%% file: main.tex
\documentclass[a4paper,USenglish,thm-restate]{lipics-v2021} 

\usepackage{multicol}
\usepackage{amsthm}
\usepackage{amssymb}
\usepackage{algorithmicx}
\usepackage[utf8]{inputenc}
\usepackage{algorithm}
\usepackage{algpseudocode}
\usepackage{graphicx}
\usepackage{float}
\usepackage{color}
\usepackage[dvipsnames]{xcolor}
\usepackage{paralist}
\usepackage[title]{appendix}
\usepackage{systeme}
\usepackage{lmodern}
\usepackage{hyperref}
\usepackage{thmtools}
\usepackage{pifont}

\newcommand{\RN}[1]{{\color{blue}{\textsc{$\bullet$ Run #1}}}}

\newcommand{\run}{{\color{black}{run}}}
\newcommand{\Run}{{\color{black}{Run}}}

\newcommand{\REG}{\textbf{R}}

\newcommand{\tsts}{temporarily stops taking steps}

\newcommand{\Wu}{\textsc{Write}}
\newcommand{\Ru}{\textsc{Read}}

\newcommand{\wu}{\textsc{w}}
\newcommand{\ru}{\textsc{r}}

\newcommand{\Iup}[1]{{\cal{I}}'_{#1}}

\newcommand{\Iu}[1]{{\cal{I}}_{#1}}

\newcommand{\I}[1]{\mathit{I_{#1}}}

\newcommand{\Is}{${\cal{I}}_{n}^{s}$}

\newcommand{\regx}{$R_{xn}$}

\newcommand{\AW}{\mathcal{I}}

\newcommand{\AWB}{${\cal{I}}_{B}$}

\newcommand{\Reg}[2]{\mbox{$[#1,#2]$-register}}

\newcommand{\tp}[1]{\langle #1 \rangle}

\newcommand{\ilns}[1]{in line~\ref{#1} of}

\newcommand{\iln}[1]{in}

\newcommand{\Vals}{\mathit{tuples}}
\newcommand{\val}{\mathit{val}}

\newcommand{\lw}{\mathit{last\_written}}
\newcommand{\lr}{\mathit{last\_read}}

\newcommand{\lc}{\mathit{previous\_k}}

\newcommand{\ready}{\textsc{prepare}}
\newcommand{\done}{\textsc{commit}}

\newcommand{\ml}{malicious}

\newcommand{\vd}{valid}

\newcommand{\Done}{\textbf{done}}

\newcommand{\rinit}{u_0}

\newcommand{\sety}{\textit{Z}}
\newcommand{\pk}[1]{P_{#1}}

\newcounter{NewCounter}

 \newtheorem{assumption}[NewCounter]{Assumption}

\algdef{SE}[SUBALG]{Indent}{EndIndent}{}{\algorithmicend\ }
\algtext*{Indent}
\algtext*{EndIndent}

\ArticleNo{50}


\title{On implementing SWMR registers from SWSR registers in systems with Byzantine failures}

\author{Xing Hu \qquad Sam Toueg}{Department of Computer Science, University of Toronto, Canada}{}{}{}
\authorrunning{X.\,Hu and S.\,Toueg}

\begin{document}
\nolinenumbers
\maketitle 

\setlength{\parindent}{0pt}
\setlength{\parskip}{2pt}

\setlength{\textheight}{24.4cm}

\input{abstract.tex}

\input{introduction.tex}

\input{resulttech.tex}

\input{model.tex}
%
\input{proof-impossibility.tex}

%
\input{section-algo-n-readers.tex}
%
\input{section-algo-n-readers-signed.tex}

%
 \input{section-impossibility-bounded-termination.tex}

\input{corollaries.tex}

%
\input{conclusion.tex}
\section*{Acknowledgments}
We thank Vassos Hadzilacos for his helpful comments on this paper.
This work was partially funded by the Natural Sciences and Engineering Research Council of Canada (Grant number: RGPIN-2014-05296).
\newpage
\bibliography{biblio}

\appendix

\input{section-algo-2-readers}

\input{proof-impossibility-bounded-termination.tex}
\end{document}

%% file: abstract.tex

\begin{abstract}
The implementation of registers from (potentially) weaker registers is a classical problem
	in the theory of distributed computing.
Since Lamport's pioneering work~\cite{Lamport86}, this problem has been extensively studied
	in the context of asynchronous processes with crash failures.
In this paper, we investigate this problem in the context of Byzantine process failures, with and without process signatures.

We first prove that, without signatures, there is no wait-free
	linearizable
	implementation of a 1-writer $n$-reader register from atomic 1-writer $1$-reader registers.
In fact, we show a stronger result, namely,
	even under the assumption that
	the writer can only crash and at most one reader can be {\ml},
	there is no linearizable
	implementation of a 1-writer $n$-reader register from atomic 1-writer $(n-1)$-reader registers
	that ensures that every
	correct process eventually completes its operations.	

In light of this impossibility result,
	we give two implementations of a 1-writer $n$-reader register from atomic 1-writer $1$-reader registers
	that work under different assumptions.
The first implementation is linearizable (under any combination of process failures),
	but it guarantees that every
	correct process eventually completes its operations
	only under the assumption that
	the writer is correct or no reader is {\ml}
	--- thus matching the impossibility result.
The second implementation assumes process signatures; it is bounded wait-free and linearizable
	under any combination of process failures.
	
Finally, we show that without process signatures, even if we assume that
	the writer is correct and at most one of the readers can be {\ml},
	it is impossible to guarantee that every correct reader completes each read operation in a \emph{bounded} number of steps.
\end{abstract}

%% file: introduction.tex

\section{Introduction}
We consider the basic problem of implementing a single-writer \emph{multi}-reader
	register from atomic single-writer \emph{single}-reader registers in a system
	where processes are subject to \emph{Byzantine failures}.
In particular, (1) we give an implementation that works under some failure assumptions, and (2) we prove a matching impossibility result for the case when these assumptions do not hold.
We also consider systems where processes can use unforgeable signatures, and give an implementation
	that works for any number of faulty processes.
We now describe our motivation and results in detail.

\vspace{-2mm}

\subsection{Motivation}
Implementing shared registers from weaker primitives is a fundamental problem
that has been thoroughly studied in distributed computing~\cite{abd,Bloom1988,Burns1987,Haldar1995,Israeli1992,Lamport86,Newman1987,Peterson1984,Peterson1987,Singh1987,Vidyasankar1988,Vidyasankar1991,Paul1986}.
In particular, it is well-known that in systems where processes are subject to \emph{crash} failures,
	it is possible to implement a 
	$m$-writer $n$-reader register (henceforth denoted $\Reg{m}{n}$)
	from atomic 1-writer 1-reader registers (denoted $\Reg{1}{1}$s).

 In this paper, we consider
	the problem of implementing \emph{multi}-reader registers from \emph{single}-reader registers
	in systems where processes are subject to \emph{Byzantine} failures.\
In particular, we consider the following basic questions:

\vspace{-2mm}
\begin{itemize}
\item
Is it possible to implement a 
	$\Reg{1}{n}$
	from atomic $\Reg{1}{1}$s in systems
	with Byzantine processes?

\item If so, under which assumption(s) such an implementation exist?
\end{itemize}

The above questions are also motivated by the growing
	interest in shared-memory or hybrid systems where processes are subject to 
	Byzantine failures. 
For example, Cohen and Keidar~\cite{CohenKeidar2021} 
	give $f$-resilient implementations of several objects
	(namely, \emph{reliable broadcast}, \emph{atomic snapshot}, and \emph{asset transfer} objects)
	using atomic $\Reg{1}{n}$s in systems with Byzantine failures where at most $f <n/2$ processes are faulty.
As another example, Aguilera \emph{et al.}
	use atomic $\Reg{1}{n}$s to solve some agreement problems
	in hybrid systems with Byzantine process failures~\cite{Aguilera2019}.
Moreover, Most\'efaoui \emph{et al.}~\cite{Mostefaoui2016}
	prove that,
	in \emph{message-passing} systems with Byzantine process failures,
	there is a $f$-resilient 
	implementation of a $\Reg{1}{n}$ 
	if and only if at most $f < n/3$ processes are faulty.

\vspace{-3.8mm}
\subsection{Description of the results}

\vspace{-1mm}

In this section, when we write ``implementation'', we mean an implementation
	that is both:
	(a)~``safe'', i.e., it is linearizable~\cite{CohenKeidar2021,linearizability,Mostefaoui2016}, and
	(b)~``live'', i.e., it ensures that every
	correct process eventually completes its operations
	(possibly under some failure assumptions).

To simplify the exposition of our results,
	we first state them in terms of two process groups:
	\emph{correct} processes that do not fail and \emph{faulty} ones.
We show that in a system with Byzantine failures the following matching impossibility and possibility results hold.
 For all $n \ge 3$:

\begin{compactitem}
\item[\textbf{(A)}] If the writer 
	\emph{and} some readers (even if only one reader) can be faulty,
	then there is no implementation of a $\Reg{1}{n}$ from atomic $\Reg{1}{n-1}$s.
	
\item[\textbf{(B)}] If the writer
	\emph{or} some readers (any number of readers), but \emph{not both}, can be faulty, 
	then there is an implementation of a $\Reg{1}{n}$ from atomic $\Reg{1}{1}$s.
\end{compactitem}

Note that result \textbf{(A)} implies
	that there is no \emph{wait-free}	
	implementation of a $\Reg{1}{n}$ from atomic $\Reg{1}{n-1}$s.\footnote{Recall that a wait-free implementation guarantees that every correct process eventually completes its operations, \emph{regardless of the execution speeds or failures of the other processes}~\cite{herlihy91}.}

This simple version of the results, however, leaves some questions open.
One reason is because these results do not distinguish between
	the different types of faulty processes
	(recall that Byzantine failures encompass
	all the possible failure behaviours, from simple crash to ``malicious'' behaviour).
For example we may ask: what happens 
	if we can
	assume that some processes (say the writer) are subject to crash failures \emph{only},
	while some other processes (say the readers) can fail in ``malicious''~ways?
Is an implementation of a $\Reg{1}{n}$ from atomic $\Reg{1}{1}$s now possible?

To answer this and similar questions, we partition processes into \emph{three} separate groups:
	(a)~ those that do not fail, called \emph{correct} processes,
	(b)~those that fail \emph{only} by crashing, and 
	(c) ~those that fail in any other way, called \emph{{\ml}} processes.
In systems with a mix of such process failures, we prove the following: 

\begin{compactitem}
\item[\textbf{(1)}] For all $n \ge 3$,
	there is no implementation $\Iu{n}$ of a $\Reg{1}{n}$ from atomic $\Reg{1}{n-1}$s,
	even if we assume that the writer 
	can only crash and at most one of the readers can be {\ml}. 
\end{compactitem}

In fact, we show that this impossibility result holds even if \emph{every} reader is given atomic $\Reg{1}{n}$s that it can write and \emph{all} processes can read, and the writer is the only process that does not have atomic $\Reg{1}{n}$s.

Note that 
	the above results consider safety and liveness as an \emph{indivisible} requirement
	of a register implementation.
But it could be useful to consider each requirement separately.
For~example, what happens if we want to implement a $\Reg{1}{n}$ with the following properties: (a)~it is \emph{always} safe (i.e., linearizable)
	and (b) it may lose its liveness  (i.e., it may block some
	read or write operations) \emph{only if} some specific ``pattern/types'' of failures occur?
We prove that in systems with a mix of process failures:

\begin{compactitem}
\item[\textbf{(2)}] For all $n \ge 3$, there is an implementation $\Iu{n}$ of a $\Reg{1}{n}$ from atomic $\Reg{1}{1}$s 
	such that:
\begin{compactitem}
\item $\Iu{n}$ is linearizable, and

\item  In every run of $\Iu{n}$ where the writer is correct or no reader is {\ml},
	every correct process completes all its~operations.

\end{compactitem}
\end{compactitem}

\noindent
So this register implementation is linearizable
	regardless of which processes fail and how they fail, i.e., it is always ``safe''.
But it guarantees ``liveness'' only if the writer is correct or no reader is {\ml}.
If the writer is correct, it tolerates any number of {\ml} readers.

\noindent
Note that \textbf{(1)} and \textbf{(2)} are matching impossibility and possibility results.
They imply the simpler results \textbf{(A)} and~\textbf{(B)} that we stated earlier
	for processes that are coarsely characterized as either correct or faulty.

If we assume that the writer is correct, the linearizable implementation 
	of result \textbf{(2)} above
	ensures that every correct reader completes each read in a \emph{finite} number of steps.
This raises the question of whether, if we assume that the writer is correct,
	there is a linearizable implementation such that
	every reader completes each read in a \emph{bounded} number of steps.
We prove that the answer is ``No''.
More precisely:

\begin{compactitem}
\item[\textbf{(3)}] For all $n \ge 3$,  even if we assume that the writer is correct and at most one reader can be {\ml},
	there is \emph{no}
	linearizable implementation of a 
	$\Reg{1}{n}$ from atomic \mbox{$\Reg{1}{n-1}$}s
	that ensures that every correct reader completes every read in a \emph{bounded} number of steps.
\end{compactitem}

The above results are for the case that the implemented register has at least $n=3$ readers.
For the special case that $n=2$,  we give a simple implementation of a $\Reg{1}{2}$ 
	from atomic $\Reg{1}{1}$s that is \emph{bounded wait-free}: 
	all correct processes are guaranteed to complete their operations
	in a bounded number of steps
	regardless of which processes fail and how they fail.

We also consider the problem of implementing a $\Reg{1}{n}$ from atomic $\Reg{1}{1}$s in systems where processes
	are subject to Byzantine failures, but they can use \emph{unforgeable signatures}.
In sharp contrast to the impossibility result~\textbf{(1)},
	we show that with signatures for all $n\ge2$,
	there is an implementation of $\Reg{1}{n}$ from atomic $\Reg{1}{1}$s that is \mbox{bounded wait-free}.
	
We conclude the paper with a result about implementations from \emph{regular} registers~\cite{Lamport86}.
Recall that, in contrast to atomic registers,
	regular registers allow ``new-old'' inversions in the values that processes read.
It is well-known that in systems with crash failures,
	it is easy to implement a \emph{wait-free} linearizable
	$\Reg{1}{n}$ from regular \mbox{$\Reg{1}{n}$}s.
Here~we~show~that in systems with
	Byzantine failures, such an implementation is impossible:
	for~$n \ge 3$,
	even if we assume that the writer can only crash and at most one reader can be {\ml},
	there is \emph{no}
	linearizable implementation of a 
	$\Reg{1}{n}$ from regular~\mbox{$\Reg{1}{n}$}s\footnote{So all processes,
	\emph{including the writer},
	are given regular registers that all the $n$ readers~can~read.}\linebreak
	that ensures that every correct process eventually completes its operations.

%% file: resulttech.tex

\section{Result techniques}

The techniques that we used to obtain our main possibility and impossibility results
	are also a significant contribution of this paper.
	
To prove the impossibility result \textbf{(1)},
	one cannot use a standard partitioning argument:
	all the processes
	except the writer are given atomic $\Reg{1}{n}$s that all processes can read,
	and the writer is given a $\Reg{1}{n-1}$ that all the readers except one can read;
	thus it is clear that the system cannot be partitioned.

So to prove this result we use an interesting \emph{reductio ad absurdum} technique.
Starting from an alleged implementation
	of  $\Reg{1}{n}$ from $\Reg{1}{n-1}$s,
	we consider a run where the implemented register is initialized to $0$, the writer completes a write of $1$, and then a reader reads~$1$.
By leveraging the facts that:
	(1)~in each step the writer can read or write only $\Reg{1}{n-1}$s,
	(2)~the writer may crash,
	(3)~one~of~the readers may be malicious,
	and (4)~there are at least 3 readers,
	we are able to successively remove every read or write step of the writer (one by one,
	starting from its last step) in a way that
	maintains the property that some correct reader reads 1 and at most one reader in the run is malicious.
	As we successively remove the steps of the writer,
	the identity of the reader that reads 1, and the identity of the reader that may be malicious, keep changing.
	By continuing this process, we end up with a run in which the writer takes no steps, and yet a correct reader reads~1.

Note that this proof is reminiscent of the impossibility proof for the ``Two generals' Problem''
	in message-passing systems~\cite{twogeneral}.
In that proof, one leverages the possibility of message losses to successively
	remove one message at a time.
The proof given here is much more elaborate because it leverages the subtle interaction
	between crash \emph{and} malicious failures that may occur at different processes.

For the matching possibility result \textbf{(2)}, we solve the problem of implementing a $\Reg{1}{n}$ from $\Reg{1}{1}$s
	with a \emph{recursive} algorithm: intuitively, we first give an algorithm to implement a $\Reg{1}{n}$ using $\Reg{1}{n-1}$s, rather than only $\Reg{1}{1}$s, and then recurse till $n=2$.
We do so because the recursive step of implementing a $\Reg{1}{n}$ using $\Reg{1}{n-1}$s
	is significantly easier than implementing a $\Reg{1}{n}$
	using only $\Reg{1}{1}$s.
This is explained in more detail in Section~\ref{algo-difficulty}.

%% file: model.tex
\section{Model Sketch}\label{model}

We consider systems with asynchronous processes that communicate via single-writer registers and are subject to Byzantine failures.
Recall that a single-writer $n$-reader register is denoted as a $\Reg{1}{n}$; the $n$ readers are distinct from the writer.

\subsection{Process failures}

A process that is subject to Byzantine failures can behave arbitrarily.
In particular, it may deviate from the algorithm it is supposed to execute, 
	or just stop this execution prematurely, i.e., crash.
To distinguish between these two types of failures, we partition processes as follows:
 \begin{compactitem}
 \item  Processes that do not fail, i.e., \emph{correct} processes.
 \item  Processes that fail, i.e., \emph{faulty} processes. Faulty processes are divided into two groups:
 \begin{compactitem}
 	  \item processes that just \emph{crash}, 
	   and
	  \item the remaining processes,
	  	which we call \emph{{\ml}}.
 \end{compactitem}
 \end{compactitem}

\subsection{Atomic and implemented registers}

A register is \emph{atomic} if its read and write operations are \emph{instantaneous} (i.e., indivisible);
	each read 
	must return the value of the last write
	that precedes it, or the initial value of the register if no such write~exists.
	
Roughly speaking, the \emph{implementation} of a register from a set of ``base'' registers
	is given by read/write procedures that each process can execute
	to read/write the implemented register; these procedures can access the given base registers
	(which, intuitively, may be less
	``powerful'' than the implemented register).
So each operation on an implemented register \emph{spans an interval} that starts
	with an \emph{invocation} (a procedure call)
	and completes with a corresponding \emph{response} (a value returned by the procedure).
Note that a process executes steps of a register implementation
	only when it executes its \emph{own} operations on the register,
	i.e., only within the intervals of these operations.	

\subsection{Implementation liveness properties}

All the register implementations that we consider
	satisfy the following liveness property:

\begin{definition}[Termination]\label{termination}
Every correct process completes every operation
	in a finite number of its own steps.
\end{definition}

As we will see, termination may rely on some failure assumptions.
For example, the register implementation that we
	give in Section~\ref{n-from-1}
	(Algorithm~\ref{1wnr}, Theorem~\ref{Theo-Main-Possibility})
	satisfies the Termination property under the assumption that
	either the writer is correct or no reader is {\ml}.
In contrast to the Termination property,
	\emph{wait-freedom} and  \emph{bounded wait-freedom} are liveness properties that do not rely on any failure assumptions~\cite{herlihy91}:

\begin{definition}[Wait-freedom]\label{wfdef}
Every correct process completes every operation in a finite number of its own steps, regardless of the execution speeds or failures of the other processes.\footnote{In a preliminary version of this paper~\cite{DISC2022},
	an implementation that satisfies the Termination property
	(under some failure assumption)
	was said to be wait-free (under this failure assumption).
In~particular, the register implementation given
	in Section~\ref{n-from-1}
	was said to be wait-free under the assumption that
	the writer is correct or no reader is {\ml}.
But this use of the term ``wait-free'' is not conventional and can be misleading.
Here we reserve the term ``wait-free'' for implementations that satisfy the Termination property \emph{unconditionnally}, as in~\cite{herlihy91}. }
\end{definition}

\begin{definition}[Bounded wait-freedom]\label{bwfdef}
Every correct process completes every operation in a bounded number of its own steps, regardless of the execution speeds or failures of the other processes.
\end{definition}

\vspace{-5.5mm}

\subsection{Linearizability of register implementations}

\vspace{-1.5mm}

Roughly speaking, linearizability requires that every operation on an implemented object appears
	as if it took effect instantaneously at some point
	(the ``linearization point'') in its execution interval~\cite{linearizability}.\footnote{Linearizable (implementations of) registers, however, are \emph{not} equivalent to atomic registers. In fact, Golab, Higham and Woelfel have shown that with a strong adversary, some randomized algorithms that ``work correctly'' under the assumption that processes use atomic registers, do not work if they use linearizable register implementations instead of atomic registers~\cite{sl11}.}
As noted by~\cite{CohenKeidar2021,Mostefaoui2016}, however, the precise definition of linearizability
	depends on whether we assume that processes can only crash (as it was assumed in~\cite{linearizability}), 
	or they can also fail in a ``Byzantine way''.
We now explain this for the special case of \emph{register} implementations.

\smallskip
\textbf{In systems with only crash failures.}
It is well-known that a \emph{single-writer multi-reader} register implementation is linearizable
	if and only if it satisfies two simple properties: intuitively, (1) every read operation reads the value
	written by a concurrent or immediately preceding write operation, and (2) there are no ``new-old'' inversions in the values read.	
To define these properties precisely, we first define what it means for two operations to be concurrent or for one to precede the other.

\vspace{-2mm}
\begin{definition}
Let $o$ and $o'$ be any two operations.
\begin{compactitem}
\item $o$ \emph{precedes} $o'$ if the response of $o$ occurs 
	before \mbox{the invocation of $o'$.}

\item $o$ \emph{is concurrent with} $o'$ if neither precedes the other.
\end{compactitem}
\end{definition}

\vspace{-2mm}
We say that a write operation $\wu$ \emph{immediately precedes} a read operation $\ru$
	if $\wu$ precedes~$\ru$, and there is no write operation $\wu'$
	such that $\wu$ precedes $\wu'$ and $\wu'$ precedes $\ru$.

Let $v_0$ be the \emph{initial value} of the implemented register,
	and $v_k$ be the value written by the $k$-th write operation of the writer $w$ of the implemented register
(this is well-defined because 
	each process, including the writer, applies its operations sequentially).

\vspace{-1mm}
\begin{definition}[Register Linearizability]\label{LinearizableCrash}
In a system with crash failures,
	an implementation of a $\Reg{1}{n}$ 
	is \emph{linearizable} if and only if it satisfies the following two properties:

\begin{compactitem}
\item \label{p1}  \emph{\textbf{Property 1} [Reading a ``current'' value]}
If a read operation $\ru$ returns the value $v$ then:
	
	\begin{compactitem}
	\item there is a write $v$ operation that immediately precedes $\ru$ or is concurrent with $\ru$, or	
	\item  $v = v_0$ and no write operation precedes $\ru$.
	\end{compactitem}

\item \label{p2} \emph{\textbf{Property 2} [No ``new-old'' inversion]}
If two read operations $\ru$ and $\ru'$ return values $v_k$ and $v_{k'}$, respectively,
	and  $\ru$ precedes $\ru'$, then $k \le k'$.
\end{compactitem}
\end{definition}

\vspace{-1mm}
\noindent
\textbf{In systems with Byzantine failures.}
The above definitions do not quite work for systems with Byzantine failures.
For example, it is not clear what it means for a writer $w$ of an implemented register to ``write a value $v$''
	if $w$ is {\ml}, i.e., if $w$ \emph{deviates} from the write procedure that it is supposed to execute;
	similarly, if a reader $r$ is {\ml} it is not clear what it means for $r$ to ``read a value $v$''.
The definition of linearizability for systems with Byzantine failures avoids
	the above issues by restricting the linearization requirements
	to processes that are \emph{not} {\ml}.
More precisely:

\begin{definition}[Register Linearizability]\label{LinearizableByz}
In a system with Byzantine process failures,
	an implementation of a $\Reg{1}{n}$
	is \emph{linearizable} if and only if the following holds.
	If the writer is not {\ml},~then:

\begin{compactitem}
\item \label{pb1}  \emph{\textbf{Property 1} [Reading a ``current'' value]}
If a read operation $\ru$ by a process that is not {\ml} returns the value $v$ then:
	
	\begin{compactitem}
	\item there is a write $v$ operation that immediately precedes $\ru$ or is concurrent with $\ru$, or	
	\item  $v = v_0$ and no write operation precedes $\ru$.
	\end{compactitem}

\item \label{pb2} \emph{\textbf{Property 2} [No ``new-old'' inversion]}
If two read operations $\ru$ and $\ru'$ by processes that are not {\ml} return values $v_k$ and $v_{k'}$, respectively,
	and $\ru$ precedes $\ru'$, then $k \le k'$.
\end{compactitem}
\end{definition}

Note that if the writer is correct or only crashes,
	then readers that are correct or only~crash are required to read ``current'' values
	and also avoid ``new-old'' inversions.
So in systems where faulty processes can only crash, 
	 Definition~\ref{LinearizableByz} reduces to Definition~\ref{LinearizableCrash}.

Cohen and Keidar were the first to define linearizability for \emph{arbitrary}
	objects in systems with Byzantine failures~\cite{CohenKeidar2021},
	and their definition generalizes the definition of \emph{register} linearizability in such systems given by Most\'efaoui \emph{et al.}~in~\cite{Mostefaoui2016}.
Definition~\ref{LinearizableByz} above (which is also for register linearizability)
	is consistent~with~both.\footnote{In~\cite{CohenKeidar2021,Mostefaoui2016}, however, processes that are subject to Byzantine failures are partitioned into only two groups, namely, correct processes and faulty processes.
Thus the reader of a register that just crashes is, by definition, faulty.
So, as with all other faulty processes, by the linearizability definitions in~\cite{CohenKeidar2021,Mostefaoui2016} it~is exempt from any requirement, e.g., it is allowed to read a stale value. Definition~3~avoids~this by leveraging our subdivision of faulty processes into those that only crash and those that are malicious.}

%% file: proof-impossibility.tex
 
\section{Impossibility result}\label{Impossibility-Result}

We now prove that there is no wait-free
	linearizable
	implementation of a $\Reg{1}{n}$ from atomic $\Reg{1}{n-1}$s.
In fact, we show a stronger result:
	even under the assumption that
	the writer can only crash and at most one reader can be {\ml},
	there is no linearizable
	implementation of a $\Reg{1}{n}$ from atomic $\Reg{1}{n-1}$s
	that ensures that every
	correct process eventually completes its operations.

\begin{theorem}\label{Theo-Impossibility-Result}
For all $n \ge 3$,
	in a system with $n+1$ processes that are subject to Byzantine failures,
	there is \emph{no}
	linearizable implementation of a 
	$\Reg{1}{n}$ from atomic \mbox{$\Reg{1}{n-1}$}s
	that satisfies the Termination property,
	 even if we assume that the writer of the implemented
	$\Reg{1}{n}$ can only crash and at most one reader can be {\ml}.
\end{theorem}

\begin{proof}
Let $n \ge 3$.
Suppose, for contradiction, that there is an 
	implementation $\AW$
	of a $\Reg{1}{n}$ $\REG$ from atomic $\Reg{1}{n-1}$s
	that is linearizable (i.e., it satisfies the Register Linearizabilty property)
	and ensures that all correct processes complete their operations
	(i.e., it satisfies the Termination property),
	under the assumption that the writer $w$ of $\REG$
	can only crash and at most one of the $n$ readers of $\REG$ can be {\ml}.

\begin{figure}[!t]
\vspace{-6mm} 
\minipage{0.48\textwidth}
    \centering 
    \includegraphics[width=0.9\textwidth]{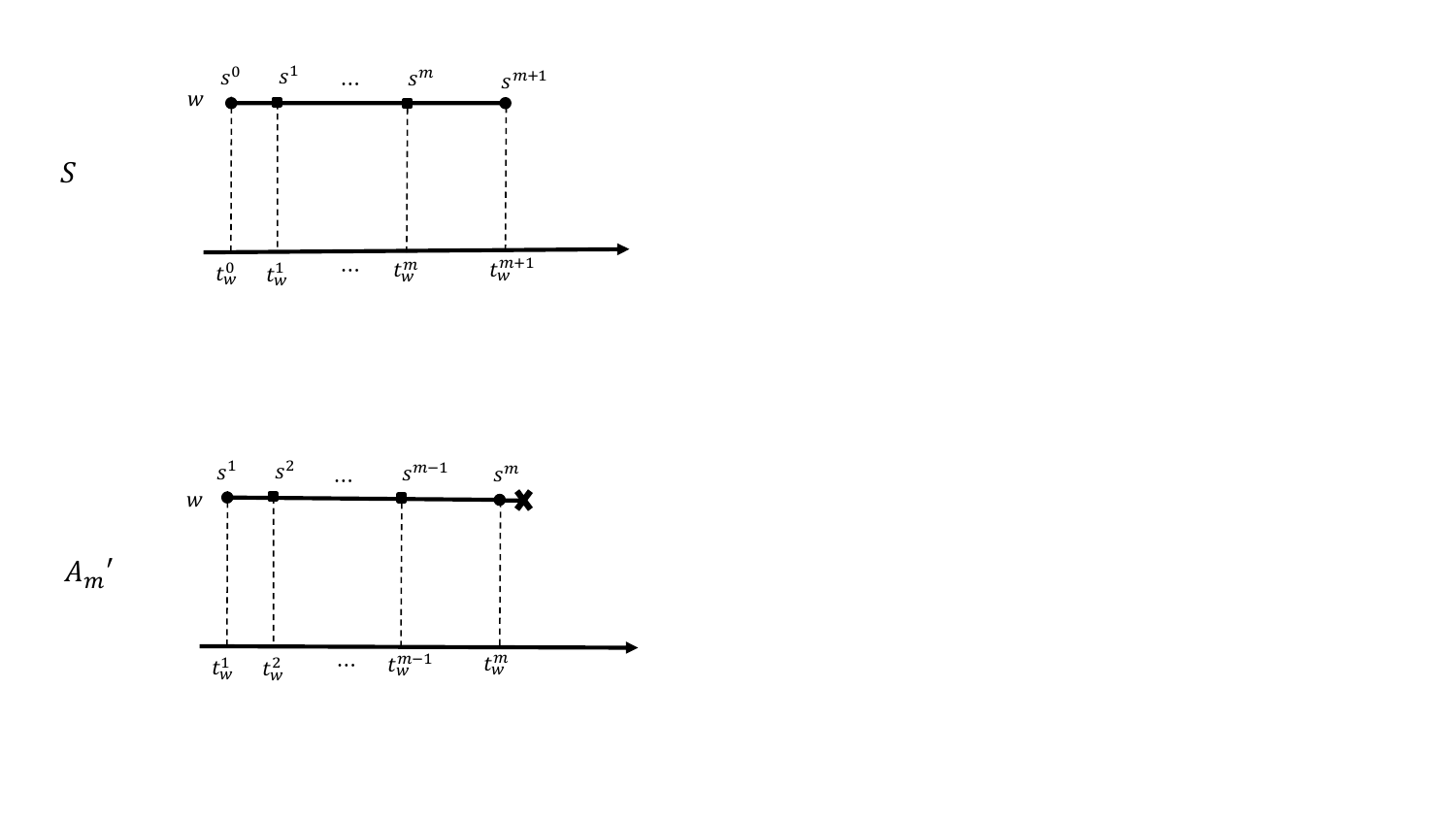}
    \caption{{\Run} $A_{m}'$} 
    \label{S}
\endminipage
\minipage{0.48\textwidth}
     \centering 
    \includegraphics[width=0.9\textwidth]{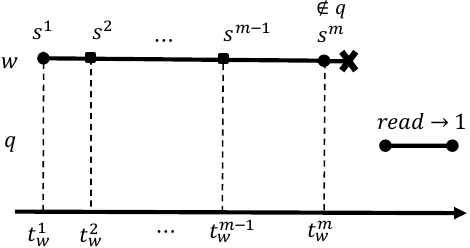}
    \caption{{\Run} $A_m$} 
    \label{Am}
\endminipage
\newline

\minipage{0.48\textwidth}
    \centering 
    \includegraphics[width=0.9\textwidth]{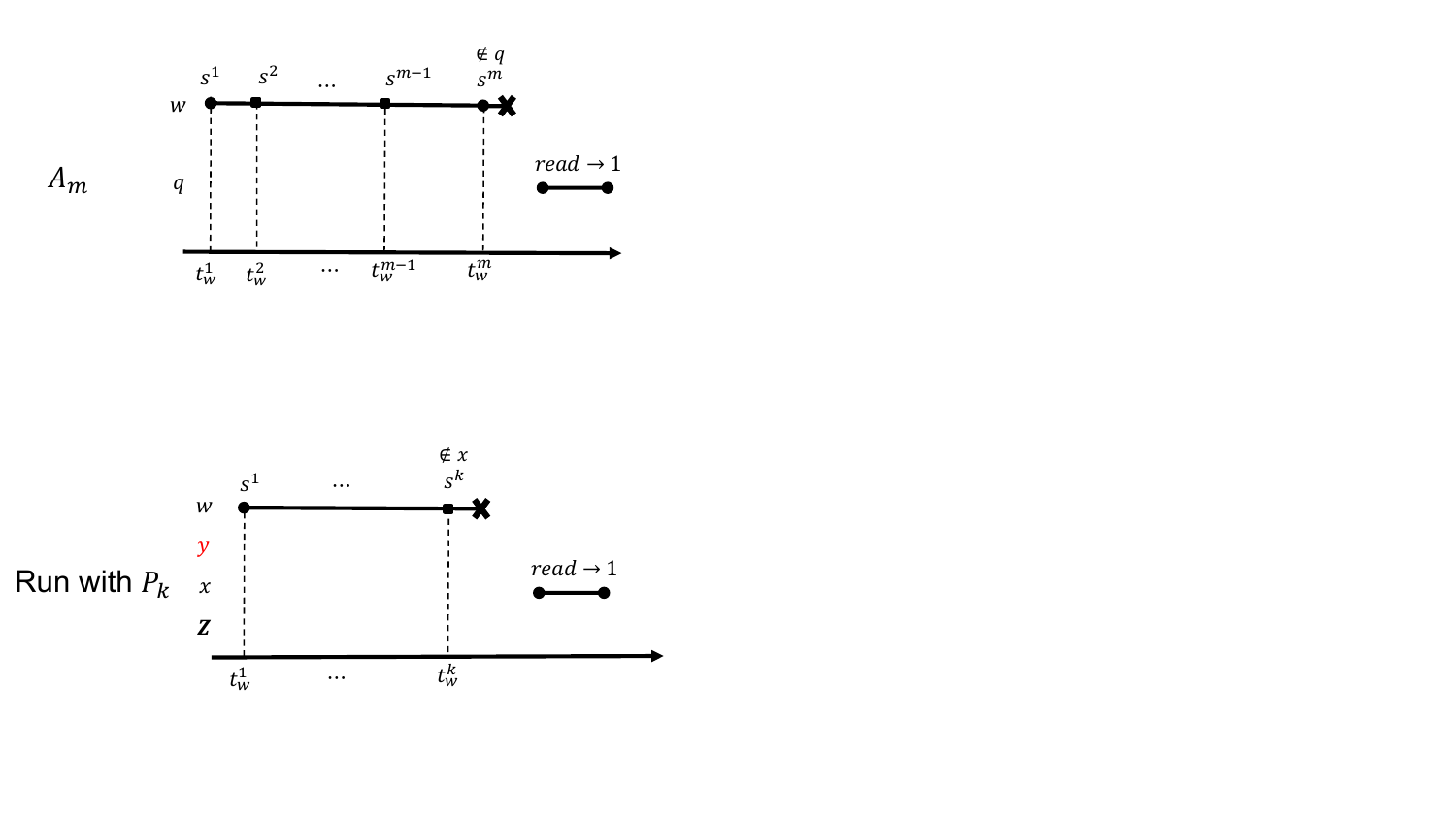}
    \caption{A {\run} with property $P_k$} 
    \label{E}
\endminipage
\minipage{0.48\textwidth}
    \centering 
    \includegraphics[width=0.9\textwidth]{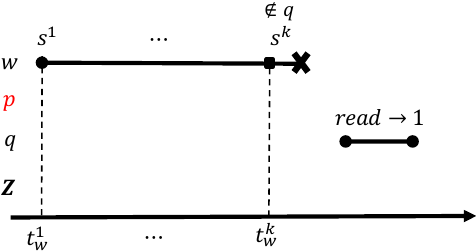}
    \caption{{\Run} $A_{k}$} 
    \label{Ak}
\endminipage\hfill
\newline

\minipage{0.48\textwidth}
    \centering 
    \includegraphics[width=0.9\textwidth]{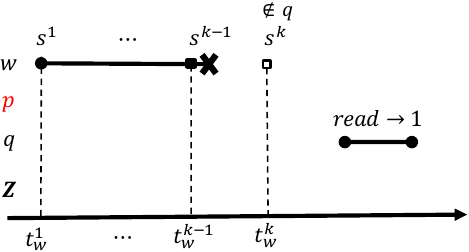}
    \caption{{\Run} $B_{k-1}$} 
    \label{Bk-1}
\endminipage
\minipage{0.48\textwidth}%
    \centering 
    \includegraphics[width=0.9\textwidth]{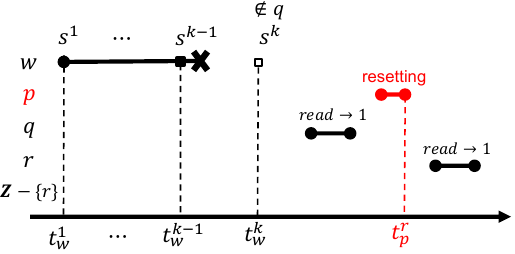}
    \caption{{\Run} $C_{k-1}^r$} 
    \label{Ck-1}
\endminipage\hfill
\newline

\minipage{0.48\textwidth}
    \centering 
    \includegraphics[width=0.9\textwidth]{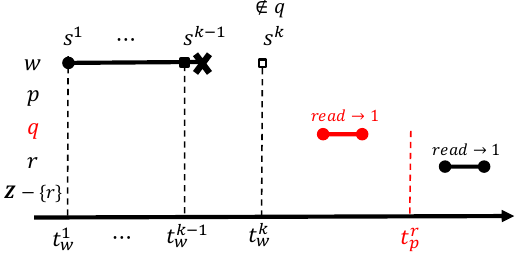}
    \caption{{\Run} $D_{k-1}^r$} 
    \label{Dk-1}
\endminipage
\minipage{0.48\textwidth}
    \centering 
    \includegraphics[width=0.9\textwidth]{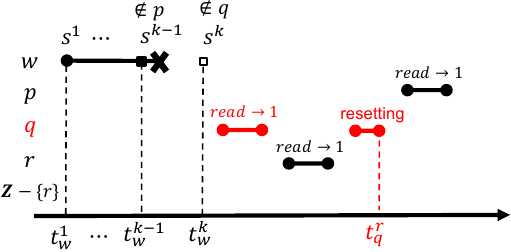}
    \caption{{\Run} $E_{k-1}^r$} 
    \label{Ek-1}
\endminipage\hfill
\newline

\minipage{0.48\textwidth}%
    \centering 
    \includegraphics[width=0.9\textwidth]{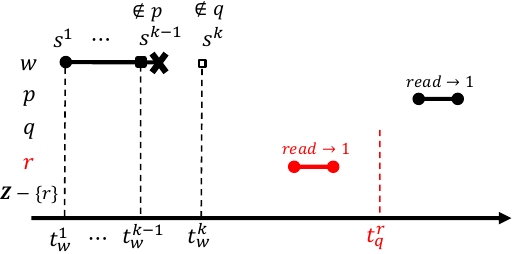}
    \caption{{\Run} $F_{k-1}^r$} 
    \label{Fk-1}
\endminipage
\minipage{0.48\textwidth}
    \centering 
    \includegraphics[width=0.9\textwidth]{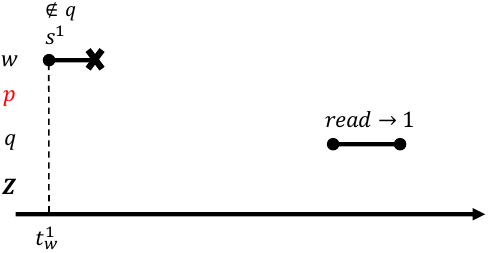}
    \caption{{\Run} $A_1$} 
    \label{A1}
\endminipage\hfill
\newline

\minipage{0.48\textwidth}
    \centering 
    \includegraphics[width=0.9\textwidth]{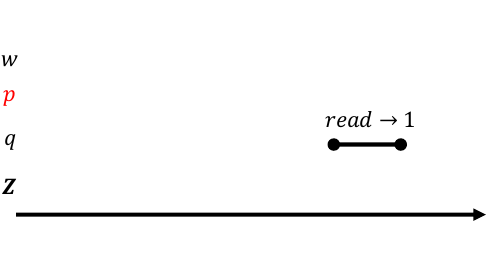}
    \caption{{\Run} $A_0$} 
    \label{A0}
\endminipage\hfill
\end{figure}

We now construct a sequence of {\run}s of $\AW$ that leads to a contradiction.
In all these {\run}s, the initial value of the implemented $\REG$ is $0$,
	the writer $w$ invokes only one operation into~$\REG$, namely a write of $1$,
	and each reader reads $\REG$ at most once (i.e., $\REG$ is only a ``one-shot'' binary register).
Moreover, in each of these {\run}s
	the writer crashes (but it is not {\ml}) and
	 there is at most one {\ml} reader;
	\mbox{the other $n-1$ readers~are~correct.}
Thus, these runs of $\AW$
	must satisfy the linearizability Properties 1 and 2 of
	Register Linearizability (Definition~\ref{LinearizableByz}),
	and every correct reader must complete any read operation that it invokes.

\begin{definition}\label{invisible}
Let $s$ be any step that the writer $w$ takes when executing the implementation $\AW$ of $\REG$.
Step $s$ is \emph{invisible
	to a reader $p$} if $s$ is either a local step of $w$, or 
	the reading or the writing of an atomic $\Reg{1}{n-1}$ that is not readable by $p$.
\end{definition}

Since there are $n$ readers,
	and the registers that $w$ can write 
	are atomic $\Reg{1}{n-1}$s,
	every write by $w$ into one of these registers
	is invisible to one of the readers.
So: 

\begin{observation}\label{invisibletoareader}
Let $s$ be any step that the writer $w$ takes when executing the implementation $\AW$ of $\REG$. 
Step $s$ is invisible to at least one of the $n$ readers.
\end{observation}

Let $A'_m$ be the following {\run} of $\AW$ (see Figure~\ref{S}):

\begin{compactitem}

\item The readers do not invoke any read operations, and so they take no steps.

\item The writer $w$ invokes an operation to write 1 on $\REG$.
By the Termination property of the implementation,
it completes this operation in a finite number of steps.

During this write operation, 
	$w$ takes a sequence of steps $s^1,...,s^{m}$
	such that each $s^i$ is either a local step, or
	the reading or the writing of an atomic $\Reg{1}{n-1}$
	($s^0$ is the invocation step of the write operation, and $s^m$ is the response step of this operation).
	
	Let $t_w^i$ be the time when step $s^i$ occurs.

\item After the time $t_w^{m}$ when $w$ completes its write operation, $w$ crashes.
	
\end{compactitem}

From the {\run} $A'_m$ of $\AW$, it is clear that the following run
	is also a {\run} of $\AW$ (see Figure~\ref{Am}):

\RN{$A_{m}$:}
\begin{compactitem}

\item The writer $w$ behaves exactly as in $A'_m$.
		
\item All the readers are correct.

\item Let $q$ be a reader such that step $s^m$ is invisible to $q$
	(by Observation~\ref{invisibletoareader}, this reader exists).

	 After the writer $w$ crashes at time $t_w^{m}$, $q$ invokes a read operation on $\REG$.
	 By the Termination property of the implementation,
	 $q$ completes its read operation.
By the linearizability properties of $\AW$,
	this read operation on $\REG$ returns 1.

\item All the other readers do not invoke any read operations, and so they take no steps.

\end{compactitem}

\begin{definition}\label{pk}
For every $k$,
 	$1 \le k \le m$,
 	a {\run} of $\AW$ has property $\pk{k}$
	 if the following holds:
\begin{compactenum}
\item Up to and including time $t_w^k$, all processes behave exactly as in $A_m$, that is:
\begin{compactitem}	
	\item $w$ takes steps $s^0, s^1, \ldots, s^k$
	\item All the readers take no steps.
\end{compactitem}

\item After taking step $s^k$ at time $t_w^k$, $w$ crashes before taking further steps.

\item There is a reader $x$ that is correct such that step $s^k$ is invisible to $x$.
After time $t_w^k$, reader~$x$ starts and completes a read operation on $\REG$ that returns~$1$.

\item There is a reader $y \neq x$ that may be correct or malicious. After time $t_w^k$, reader $y$ may or may not take steps.

\item There is a set {\sety} of $n-2$ distinct readers other than $x$ and $y$ that are correct and~take~no~steps.
\end{compactenum}
\end{definition}

Note that since $n\ge 3$, the set {\sety} contains at least one reader.
Furthermore, all the readers that take steps do so after time $t_w^k$.

A {\run} of $\AW$ with property $\pk{k}$ 
	 is shown in Figure~\ref{E}.
In this figure and all the subsequent ones,
	correct readers are in black font,
	while the reader that may be {\ml} is colored~{\color{red} red}
	(this~reader may have taken some steps after time $t_w^k$, but these are \emph{not} shown in the figure).
The ``$\notin x$'' on top of a step $s^i$ means that $s^i$
	is invisible to the reader~$x$.
The symbol \ding{54} indicates where the crash of the writer $w$ occurs.

Note that the {\run} $A_m$ of $\AW$ satisfies property $\pk{m}$:
	the reader denoted $x$ in property $\pk{m}$ is the reader $q$ of {\run} $A_m$,
	the reader $y$ of $\pk{m}$ is an arbitrary reader other than $q$ in $A_m$,
	and the set $Z$ of $\pk{m}$ is the set of the remaining $n-2$ readers in $A_m$.
So we have:

\begin{observation}\label{baserun}
Run $A_m$ of $\AW$ has property $\pk{m}$.
\end{observation}

\begin{claim}~\label{induction}
	 For every $k$, $1 \le k \le m$, there is a {\run} of $\AW$ that has property $\pk{k}$.
\end{claim}

\begin{proof}
We prove the claim by a backward induction on $k$, starting from $k=m$.

\smallskip\noindent
\textbf{Base Case:} $k=m$. This follows directly from Observation~\ref{baserun}.

\noindent
\textbf{Induction Step}: Let $k$ be such that $1 < k \le m$.

\RN{$A_{k}$.} Suppose
	there is a {\run} $A_{k}$
	of $\AW$ that has property $\pk{k}$ (this is the induction hypothesis).
We now show that there is a {\run} $A_{k-1}$ of $\AW$ that has property $\pk{k-1}$.

Since {\run} $A_{k}$ of $\AW$ satisfies $\pk{k}$, the following holds in $A_{k}$ (see Figure~\ref{Ak}):
\begin{compactitem}  

\item Up to and including time $t_w^k$, all processes behave exactly as in $A_m$.

\item After taking step $s^k$ at time $t_w^k$, $w$ crashes before taking further steps.

\item There is a reader $q$ that is correct such that step $s^k$ is invisible to $q$.
After time $t_w^k$, 
reader~$q$~starts and completes a read operation on $\REG$ that returns~$1$.

\item There is a reader $p \neq q$ that may be correct or malicious. After time $t_w^k$, reader $p$ may or may not take steps.\footnote{These steps are not shown in Figure~\ref{Ak}.}

\item There is a set {\sety} of $n-2$ distinct readers other than $p$ and $q$ that are correct and~take~no~steps.

\end{compactitem}

\RN{$B_{k-1}$.}
From the {\run} $A_{k}$ of $\AW$ we construct the following {\run} $B_{k-1}$ of $\AW$ (Figure~\ref{Bk-1}).
Intuitively, $B_{k-1}$ is exactly like $A_k$
	except that $w$ crashes just before taking step $s^k$ (so $B_{k-1}$ is just $A_k$ with the step $s^k$ ``removed'').
Run $B_{k-1}$ is possible because:
	(1) even though $p$ may have ``noticed'' the removal of step $s^k$,
		$p$~may be {\ml} (all the other readers are correct in this {\run}),
		and $p$ behaves exactly as in $A_k$,
	and (2) $q$ cannot distinguish between $A_k$ and $B_{k-1}$
		because $s^k$ is invisible to $q$, and $p$ and all the readers in {\sety} behave as in $A_k$;
		so $q$ behaves as in $A_k$, and in particular $q$ reads 1 in $B_{k-1}$ as in $A_{k}$.

More precisely in $B_{k-1}$:

\begin{compactitem}   
\item All processes behave exactly as in $A_k$ up to and including time $t_w^{k-1}$.

\item After taking step $s^{k-1}$ at time $t_w^{k-1}$, the writer $w$ crashes before taking step $s^k$.
	
\item All the readers in {\sety} are correct and take no steps, exactly as in $A_k$.

\item $p$ behaves exactly as in $A_k$.
	This is possible because
		even though $p$ may have ``noticed'' the removal of step $s^k$,
		$p$ may be {\ml} (all the other readers are correct in this {\run}).	
\item $q$ behaves exactly as in $A_k$.
	In particular, after time $t_w^k$, 
	$q$~starts and completes a read operation on $\REG$ that returns~$1$.
	This is possible because $q$ cannot distinguish between $A_k$ and $B_{k-1}$:
	 $s^k$ is invisible to $q$, and $p$ and all the readers in {\sety} behave exactly as in $A_k$.
\end{compactitem}

Note that in $B_{k-1}$ all processes behave exactly as in $A_m$ up to and including time $t_w^{k-1}$.

There are two cases:

\textbf{Case 1:} 
\emph{$s^{k-1}$ is invisible to $q$.}
Then $B_{k-1}$ is a {\run} of $\AW$ that has the property $\pk{k-1}$, as we wanted to show.

\textbf{Case 2:} 
\emph{$s^{k-1}$ is visible to $q$.}
Then, by Observation~\ref{invisibletoareader}, $s^{k-1}$ is invisible to $p$ or to~some~$r' \in ~{\sety}$.

\noindent
\RN{$C_{k-1}^r$.}
Let $r$ be \emph{any} reader in {\sety}.
From the {\run} $B_{k-1}$ of $\AW$ we construct the following {\run} $C_{k-1}^r$ of $\AW$ (Figure~\ref{Ck-1}).
$C_{k-1}^r$ is a continuation of $B_{k-1}$ where, after the correct reader $q$ reads 1, {\ml} $p$ wipes out any trace of the write steps that it may have taken so far, and then correct reader $r\in \sety$ reads~1 (this is the only value that $r$ can read, since correct $q$ previously read~1).
More precisely:

\begin{compactitem}   
\item $C_{k-1}^r$ is an extension of $B_{k-1}$:
	all processes behave exactly as in $B_{k-1}$ up to and including the time when $q$ completes its read operation on $\REG$.

\item All the readers in $\sety-\{r\}$ are correct and take no steps\footnote{If $n=3$, then the set $\sety-\{r\}$ is empty.}.

\item After the correct reader $q$ completes its read operation on $\REG$:

\begin{compactitem}   

	\item $q$ takes no steps.

	\item 
	$p$ resets all the atomic registers that it can write to their initial values.
	Process $p$ can do so because it may be {\ml} (all the other readers are correct in this {\run}).
	Let~$t_p^r$~be the time when
	$p$ completes all the register resettings.

	\item Correct reader $r$
	starts a read operation on $\REG$ after time $t_p^r$.
	It takes no steps before this read.
	By the Termination property of the implementation,
		$r$ completes its read operation.
	Since $w$ is not {\ml}, and
		the read operation by correct $q$ precedes the read operation by $r$ and returns~$1$,
		by the linearizability of $\AW$,
		the read operation by correct reader $r$~also~returns~$1$.
\end{compactitem}
\end{compactitem}

Note that in $C_{k-1}^r$ all processes behave exactly as in $A_m$ up to and including time $t_w^{k-1}$.

\RN{$D_{k-1}^r$.}
We can now construct the following {\run} $D_{k-1}^r$ of $\AW$ (Figure~\ref{Dk-1}).
Intuitively,	we obtain $D_{k-1}^r$ from $C_{k-1}^r$ by removing all the steps of $p$.
So reader $p$ (which was {\ml} in $C_{k-1}^r$)
	is now a correct process that takes no steps.
Despite this removal, $q$ behaves exactly as in $C_{k-1}^r$ because $q$
		(which was correct in $C_{k-1}^r$)
		may now be {\ml}.
The writer $w$ also behaves exactly as in $C_{k-1}^r$ because it cannot see the removal of $p$'s steps:
	\mbox{they all occur after time~$t_w^{k-1}$.}
Correct reader $r$ behaves exactly as in $C_{k-1}^r$
	because it also cannot see the removal of $p$'s steps:
	in both $C_{k-1}^r$ and $D_{k-1}^r$,
	$r$ does not ``see'' any steps of $p$.
\mbox{So $r$ reads 1 in $D_{k-1}^r$ as in $C_{k-1}^r$.}

More precisely in~$D_{k-1}^r$:

\begin{compactitem}   
\item $w$ behaves exactly as in $C_{k-1}^r$.

\item All the readers in $\sety-\{r\}$ are correct and take no steps, as in $C_{k-1}^r$.

\item $p$ is correct and it takes no steps.
	So all the
	atomic
	registers that it can write retain their initial~values.

\item $q$ behaves exactly as in $C_{k-1}^r$.
	  This is possible because
		even though $q$ may have ``noticed'' the removal of $p$'s steps,
		$q$ may be {\ml} (all the other readers are correct in this {\run}).
	
\item $r$ behaves exactly as in $C_{k-1}^r$.
	In particular,
	after time $t_p^r$ reader $r$~starts and completes a read operation on $\REG$ that returns~$1$.
	This is possible because $r$ cannot distinguish between $C_{k-1}^r$ and~$D_{k-1}^r$:
	$r$ cannot see the removal of $p$'s steps, 
		and $q$ and all the readers in $\sety-\{r\}$ behave exactly as in $C_{k-1}^r$.
\end{compactitem}

Note that in $D_{k-1}^r$ all processes behave exactly as in $S$ up to and including time $t_w^{k-1}$.

If $s^{k-1}$ is invisible to reader $r$, it is clear that the {\run} $D_{k-1}^r$ of $\AW$ has property $\pk{k-1}$.

Recall that (1) the reader $r$ above is an \emph{arbitrary} reader in {\sety},
	and
	(2) $s^{k-1}$ is invisible to~$p$ or to some reader $r' \in \sety$.
So there are two cases:

\textbf{Subcase 2a:} \emph{$s^{k-1}$ is invisible to some reader $r' \in \sety$.}
In the above we proved that the {\run} $D_{k-1}^{r'}$ of $\AW$ has property $\pk{k-1}$,
	as we wanted to show.
	
\textbf{Subcase 2b:} \emph{$s^{k-1}$ is invisible to $p$.}

\RN{$E_{k-1}^r$.}
We construct the continuation $E_{k-1}^r$ of $D_{k-1}^r$ shown in Figure~\ref{Ek-1}:
	after $r$ reads~1, {\ml} process $q$
	wipes out any trace of the write steps that it has taken so far,
	and then correct reader $p$ 
	starts a read operation on $\REG$.
		 By the Termination property of the implementation,
		this read operation by $p$ must complete. 
	Since correct $r$ previously read~1,
		by the linearizability of $\AW$, 
		$p$ must also read 1.

More precisely in $E_{k-1}^r$:
\begin{compactitem} 

\item $E_{k-1}^r$ is an extension of the {\run} $D_{k-1}^r$:
		all processes behave exactly as in $D_{k-1}^r$ up to and including the time when $r$ completes its read operation on $\REG$.

\item All the readers in $\sety-\{r\}$ are correct and take no steps, as in $D_{k-1}^r$.

\item After the correct reader $r$ completes its read operation on $\REG$:

\begin{compactitem}   

	\item $r$ takes no steps.

	\item 
	$q$ resets all the atomic registers that it can write to their initial values.
	Process $q$ can do so because it may be {\ml} (all the other readers are correct in this {\run}).
	Let~$t_q^r$~be the time when
	$q$ completes all the register resettings.

	\item Correct reader $p$
	starts a read operation on $\REG$ after time $t_q^r$.
	It takes no steps before this read.
	 By the Termination property of the implementation,
		$p$ completes its read operation.
		Since $w$ is not {\ml},
		and the read operation by correct $r$ precedes the read operation by $p$ and returns~$1$,
		by the linearizability of $\AW$,
		the read operation by correct reader $p$~also~returns~$1$.
\end{compactitem}

\end{compactitem}

Note that in $E_{k-1}^r$ all processes behave exactly as in $A_m$ up to and including time $t_w^{k-1}$.

\RN{$F_{k-1}^r$.}
Finally, we construct the {\run} $F_{k-1}^r$ of $\AW$ by removing
	all the steps of $q$ from $E_{k-1}^r$ (see Figure~\ref{Fk-1}); so $q$ (which was {\ml} in $E_{k-1}^r$)
	is now a correct process that takes no steps.
Despite this removal, $r$ behaves exactly as in $E_{k-1}^r$ because $r$ (which was correct in $E_{k-1}^r$) may now be {\ml}.
The writer $w$ also behaves exactly as in $E_{k-1}^r$ because it cannot see the removal of $q$'s steps:
	they all occur after time $t_w^{k-1}$.
Finally, correct $p$ behaves exactly as in $E_{k-1}^r$ because it also cannot see the removal of $q$'s steps:
	in both $E_{k-1}^r$ and $F_{k-1}^r$,
	$p$ does not ``see'' any steps of $q$. 
So $p$ reads 1 in $F_{k-1}^r$ as in $E_{k-1}^r$.

More precisely in $F_{k-1}^r$:

\begin{compactitem} 

\item $w$ behaves exactly as in $E_{k-1}^r$.

\item All the readers in $\sety-\{r\}$ are correct and take no steps, as in $E_{k-1}^r$.

\item $q$ is correct and it takes no steps. So all the atomic registers that it can write retain their initial values.
	
\item $r$ behaves exactly as in $E_{k-1}^r$.
	  This is possible because
		even though $r$ may have ``noticed'' the removal of $q$'s steps,
		$r$ may be {\ml} (all the other readers are correct in this {\run}).
	
\item $p$ behaves exactly as in $E_{k-1}^r$.
	In particular, after time $t_q^r$ reader
	$p$~starts and completes a read operation on $\REG$ that returns~$1$.
	This is possible because $p$ cannot distinguish between $E_{k-1}^r$ and~$F_{k-1}^r$:
	$p$ cannot see the removal of $q$'s steps, 
		and $r$ and all the readers in $\sety-\{r\}$ behave exactly as in $E_{k-1}^r$.

\end{compactitem}

Note that in $F_{k-1}^r$ all processes behave exactly as in $A_m$ up to and including time $t_w^{k-1}$.

Since $s^{k-1}$ is invisible to $p$, it is clear that the {\run} $F_{k-1}^r$ of $\AW$ has property $\pk{k-1}$.

The above concludes the proof of the Induction Step of Claim~\ref{induction}:
	we proved that, in all possible cases, there is a {\run} of ${\AW}$ that has property $\pk{k-1}$, as we needed to show.
\end{proof}

By the Claim~\ref{induction} that we just proved,
	the implementation $\AW$ of $\REG$ has a {\run} $A_1$ with property~$\pk{1}$.
By this property, the following holds in $A_{1}$ (see Figure~\ref{A1}):
\begin{compactitem}  

\item Up to and including time $t_w^1$, all processes behave exactly as in $A_m$.

\item After taking step $s^1$ at time $t_w^1$, $w$ crashes before taking further steps.

\item There is a reader $q$ that is correct such that step $s^1$ is invisible to $q$.
After time $t_w^1$, 
reader~$q$~starts and completes a read operation on $\REG$ that returns~$1$.

\item There is a reader $p \neq q$ that may be correct or malicious. After time $t_w^1$, reader $p$ may or may not take steps.

\item There is a set {\sety} of $n-2$ distinct readers other than $p$ and $q$ that are correct and~take~no~steps.

\end{compactitem}

From the {\run} $A_{1}$ of $\AW$ we construct the following {\run} $A_0$ of $\AW$ (Figure~\ref{A0}).
Intuitively, $A_0$~is~the same as $A_1$ except that the writer is correct
	 and does not take any steps
	(i.e., $w$ does not invoke a write 1 operation on $\REG$), but all the readers
	behave the same as in $A_1$ and so $q$ still reads 1.
	This {\run} of $\AW$ is possible because:
	(1) even though $p$ may have ``noticed'' that $w$~does~not take the step $s^1$,
	$p$~may be {\ml} (all the other readers are correct in this {\run}),
	and $p$ behaves exactly as in $A_1$,
	and 
	(2) $q$ cannot distinguish between $A_1$ and $A_0$
		because $s^1$ is invisible to $q$, and $p$ and all the readers in {\sety} behave as in $A_1$.
So $q$ reads 1 from $\REG$ in $A_0$ exactly as in $A_{1}$.
Since the initial value of the implemented register $\REG$ is $0$,
	{\run} $A_0$ of the implementation $\AW$
	of $\REG$ violates the linearizability~of~$\AW$
	--- a contradiction that concludes the proof of Theorem~\ref{Theo-Impossibility-Result}.
\end{proof}

\vspace{-2mm}

It is easy to verify that the above proof
	holds (without any change)
	even if every reader is given atomic $\Reg{1}{n}$s that it can write and \emph{all} other processes can read, and the writer is the only process that does not have an atomic $\Reg{1}{n}$.		
Thus:

\begin{theorem}\label{STheo-Impossibility-Result}
For all $n \ge 3$, in a system with $n+1$ processes that are subject to Byzantine failures,
	there is \emph{no}
	linearizable implementation of a 
	$\Reg{1}{n}$
	that satisfies Termination,
	even under the assumption that:
	
\begin{compactitem}

\item the writer $w$ of the implemented
	$\Reg{1}{n}$  can only crash and at most one reader can be {\ml}, and

\item $w$ has atomic $\Reg{1}{n-1}$s, and every reader has atomic $\Reg{1}{n}$s.

\end{compactitem}

\end{theorem} 

\vspace{-4mm}

%% file: section-algo-n-readers.tex

\section{Register implementation algorithm}
We now give an implemention of a $\Reg{1}{n}$ 
	from atomic $\Reg{1}{1}$s in systems with Byzantine process failures;
	this implementation is linearizable,
	and it satisfies the Termination property 
	provided the writer of the register
	\emph{or} any number of the readers,
	\emph{but not both}, can be faulty.
More precisely, it is a \emph{\vd} implementation,  as we define below.

\begin{definition}\label{valid}
A register implementation is \emph{{\vd}} if the following holds:
\begin{compactitem}
\item \label{regpropertiesl}  \emph{It is linearizable.}

\item \label{regpropertiesw} \emph{If the writer is correct or no reader is {\ml}, it satisfies the Termination property.}
\end{compactitem}
\end{definition}

Note that, when executed in a system where processes can only crash,
	a valid register implementation is linearizable and ``terminating'' (unconditionally).

\subsection{Some difficulties to overcome}\label{algo-difficulty}
Note that in a system with Byzantine process failures, implementing a $\Reg{1}{n}$ 
	from $\Reg{1}{1}$s is non-trivial, even if the writer can only crash.
To see this, we now illustrate some of the issues that arise.
First note that with $\Reg{1}{1}$s the writer cannot \emph{simultaneously} inform all the readers about a new write.
So different readers may have different views of whether there is a write in progress:
	some readers may not see it, some readers may see it as still in progress, while other readers may see it as having completed. 
Thus readers must communicate with each other to avoid ``new-old'' inversions in the values they read.
With non-Byzantine failures, readers can easily coordinate their reads because they
	can trust the information they pass to each other.
With Byzantine failures, however, readers cannot blindly trust what other readers tell them.

For example, suppose a reader $q$
	is aware that a write $v$ operation is in progress (say because the writer $w$ directly ``told'' $q$ about it via the register that they share).
To avoid a ``new-old'' inversion, $q$ checks whether any other reader $q'$ has already read $v$ (because it is possible that from $q'$'s point of view,
	the write of $v$ already completed).
Suppose some $q'$ ``warns'' $q$ that it has already read the new value $v$, and so $q$ also reads $v$. 
But what if $q'$ is {\ml} and ``lied'' to $q$ (and only to $q$) about having read $v$?
Note that $q$ may be the \emph{only} correct reader currently aware that the write of $v$ is in progress (say because $w$ is slow).
Now suppose that
	a reader $q''$ that is \emph{not} aware of the write of $v$ also wants to read:
	if $q''$ reads the old value of the register
	this creates a ``new-old'' inversion with the newer value $v$ that $q$ previously read;
	but if $q''$ reads $v$ because $q$ warns $q''$ that it had read $v$, then $q''$ may be reading a value $v$ \emph{that was never written by the correct~writer~$w$}: $q$ itself could be {\ml} and could have ``lied'' about reading $v$!

The above is only one of many possible scenarios illustrating why it is not easy to
	implement a $\Reg{1}{n}$
	from $\Reg{1}{1}$s
	when some readers can be {\ml}, even if the writer itself is not {\ml}.
	
\subsection{A recursive solution}
To simplify this task, we do not \emph{directly} implement a $\Reg{1}{n}$ using \emph{only} $\Reg{1}{1}$s.
Instead, we first give an implementation $\I{n}$ of a $\Reg{1}{n}$
	that uses some $\Reg{1}{n-1}$s
	together with some $\Reg{1}{1}$s.
Then, by replacing the $\Reg{1}{n-1}$s with $\I{n-1}$ implementations,
	we get an implementation of the $\Reg{1}{n}$
	that uses some $\Reg{1}{n-2}$s and some $\Reg{1}{1}$s.
By recursing down to $n=2$, this gives an implementation of the $\Reg{1}{n}$ that uses only $\Reg{1}{1}$s.
In other words, we can implement a $\Reg{1}{n}$ 
	from $\Reg{1}{1}$s 
	with a \emph{recursive} construction that gradually reduces the number of readers of the base registers that it uses (all the way down to 1).
We now describe this recursive implementation and prove its correctness.

\vspace{-2mm}

\subsection{Implementing a \texorpdfstring{\Reg{1}{n}}{} from \texorpdfstring{\Reg{1}{n-1}s}{}}\label{impIn}
Algorithm~\ref{1wnr} is an implementation $\I{n}$ of a $\Reg{1}{n}$ that is writable by a process $w$
	and readable by every process in $\{p\} \cup Q$, where
	$p$ is an arbitrary reader and all remaining $n-1$ readers are in $Q$.
We distinguish $p$ from the other readers in $Q$ because $p$ and $q\in Q$ use different procedures for reading the implemented $\Reg{1}{n}$.
$\I{n}$
	uses two kinds of registers:
	\emph{atomic} $\Reg{1}{1}$s and
	\emph{implemented} $\Reg{1}{n-1}$s.
We will show that  
	$\I{n}$ is \emph{{\vd}} under the assumption that the $\Reg{1}{n-1}$ implementations that it uses are also \emph{{\vd}}.

{\large \textbf{Notation.}} Recall that if $R$ is an atomic register, all operations applied to $R$ are instantaneous,
	whereas if $R$ is an implemented register, each operation spans an interval of time, from an invocation to a response.
However, since we assume that the $\Reg{1}{n-1}$ implementations that $\I{n}$ uses are valid and therefore \emph{linearizable},
	we can think of each operation on an implemented $\Reg{1}{n-1}$ as being atomic, i.e.,
	as if it takes effect instantaneously at some point during its execution interval~\cite{linearizability}.
Thus to read or write a register $R$ we use the same notation, irrespective of whether $R$ is atomic or implemented.
In particular, in our implementation algorithm (shown in Figure~\ref{1wnr})
	we use the following notation:

\begin{compactitem}
\item ``$R \gets v$''  denotes the operation that writes $v$ into $R$.

\item ``\textbf{if} $R = \mathit{val}$ \textbf{then} $\ldots$'' means ``read register $R$ and
	if the value read  is equal to $\mathit{val}$~then~$\ldots$''
\end{compactitem}

The shared registers used by the implementation are as follows:
\begin{compactitem}

\item 
$R_{ss'}$ is an atomic $\Reg{1}{1}$
	writable by process $s$ and readable~by~process $s'$.\footnote{If $s=s'$, this ``shared register'' is actually just a \mbox{local register of process~$s$.}}
	
\item $R_{wQ}$ is an implemented $\Reg{1}{n-1}$ 
	writable by $w$ and readable by every $q \in Q$.

\item$R_{pQ}$ is an implemented $\Reg{1}{n-1}$ 
	writable by $p$ and readable by every $q \in Q$.
\end{compactitem}

{\large \textbf{Algorithm description.}} The  implementation $\I{n}$
	of a $\Reg{1}{n}$ from $\Reg{1}{n-1}$s consists of two procedures,
	namely $\Wu()$ for the writer $w$,
	and $\Ru()$ for each reader $r$ in $\{p\} \cup Q$.
To write a value $u$, the writer $w$ executes $\Wu(u)$.
If $u$ is the $k$-th value written by $w$,
	$\Wu(u)$ first forms the unique tuple $\langle k,u \rangle$
	and then it calls the lower-level write procedure $\wu(\langle k,u \rangle)$ to write this tuple.
Intuitively, $\Wu()$ tags the values that it writes with a counter value to make them unique
	and to indicate in which order they~are~written.

To read a value, a reader $r \in \{p\} \cup Q$ calls $\Ru()$,
	and this in turn calls a lower-level read procedure $\ru_r()$ that reads tuples written by $\wu()$.
There are two version of the procedure $\ru_r()$: one used when $r=p$ and one used when $r \in Q$.
If $\ru_r()$ returns a tuple of the form $\langle j,v \rangle$,
	then $\Ru()$ strips the counter $j$ from the tuple and returns the value $v$ as the value read
	(otherwise $\Ru()$ returns~$\bot$ to indicate a read failure).
	
Thus the lower-level procedures
	$\wu()$, $\ru_p()$, and $\ru_q()$ for each $q \in Q$,
	are executed to write and read unique tuples of the form $\langle k,u \rangle$.
We now describe how these procedures~work.

$\bullet$\hspace{1mm} To execute $\wu(\langle k,u  \rangle)$,
	process $w$ first writes $(\ready,\lw,\langle k,u \rangle)$ in the 
	$R_{wp}$ register
	that $p$ can read, and then in the 
	$R_{wQ}$ register that every process in $Q$ can read;
	$\lw$ is the \emph{last} tuple written by $w$ before $\langle k,u \rangle$ (so
	$\lw = \langle k-1,u'\rangle$ for some $u'$).
Then, $w$ writes $(\done,\langle k,u \rangle)$ into $R_{wp}$ and then~into~$R_{wQ}$.

\input{code-algo-n-readers}

$\bullet$\hspace{1mm} To execute $\ru_p()$, process $p$ reads $R_{wp}$ (line~\ref{pCommit}).
If $p$ reads $(\done,\langle k,u \rangle)$ 
	with a $k$ at least as big as those it saw before,
	 it returns $\langle k,u \rangle$ as the tuple read (line~\ref{pCrt});
	just before doing so, however, it writes $\langle k,u \rangle$
	in the $R_{pQ}$ register that every process $q \in Q$ can read (line~\ref{tellq}): intuitively, this is to ``warn'' them that
	$p$ read a ``new'' tuple, to help avoid ``new-old'' inversions in the tuples read.
	
If $p$ reads $(\ready,\lw,-)$ (line~\ref{pPre}), then it returns $\lw$ as the tuple read
	(without giving any ``warning'' about this to processes in $Q$).

If $p$ reads anything else from $R_{wp}$, then it returns~$\bot$ (the writer is surely {\ml}).

$\bullet$\hspace{1mm}To execute $\ru_q()$, process $q\in Q$ reads $R_{wQ}$.
If $q$ reads $(\done,\langle k,u \rangle)$ (line~\ref{qCommit}),
	it just returns $\langle k,u \rangle$ as the tuple read in~line~\ref{qCrt}
	(without ``warning'' other processes).
	
If $q$ reads $(\ready,\lw,\langle k,u \rangle)$ (line~\ref{qPre}), then
	$q$ \emph{cannot} simply return $\lw$ as the tuple read:
	this is because $p$ could have already read $(\done,\langle k,u \rangle)$ from $R_{wp}$
	and so $p$ could have already read the ``newer'' tuple $\langle k,u \rangle$ with $\ru_p()$.
So $q$ must determine whether to return $\lw$ or $\langle k,u \rangle$.
To do so, $q$ forks two threads and executes them in parallel (we will explain why below).

If $q$ does not read values of the form $(\ready,\lw,\langle k,u \rangle)$ or $(\done,\langle k,u \rangle)$ from $R_{wQ}$, then  $q$ just returns $\bot$ in line~\ref{qrtbot} ($w$ is surely {\ml}).

In \textsc{Thread 1},
	process $q$ keeps reading $R_{wQ}$: if it ever reads
	$(\done,\langle k',- \rangle)$ with $k' \ge k$,
	or $(\ready, - ,\langle k', - \rangle)$ with $k'>k$,
	it simply returns $\langle k,u \rangle$ as the tuple read.
Note that \emph{if the writer $w$ is correct},
	then $q$ \emph{cannot} spin forever in this thread without returning $\langle k,u \rangle$.

In \textsc{Thread 2},
	process $q$ first reads the register $R_{pQ}$ to see whether $p$ ``warned'' processes in $Q$
	that it read a tuple at least as ``new'' as $\langle k,u \rangle$.

\textbf{-}\hspace{1mm} If $q$ sees that $R_{pQ}$ contains a tuple
	at least as ``new'' as $\langle k,u \rangle$  (line~\ref{qfromp}),
	then $q$ returns $\langle k,u \rangle$
	as the tuple read (line~\ref{qPrt2});
	but before doing so, $q$ successively writes $\langle k,u \rangle$
	in each register $R_{qq'}$ such that $q' \in Q$ (line~\ref{qtor}): intuitively, this is to ``warn'' each process in $Q$ that $q$
	read this ``new'' tuple.
	
\textbf{-}\hspace{1mm} Otherwise, $q$ reads every $R_{q'q}$ register
	to avoid a new-old inversion with any tuple read by any process $q' \in Q$:
	if $q$ sees that some $R_{q'q}$ contains a tuple at least as ``new'' as $\langle k,u \rangle$  (line~\ref{rtoq}),
	then $q$ reads $R_{pQ}$ \emph{again} (line~\ref{qfrompp}) (so $q$ does \emph{not} simply
	``trust'' $q'$ and return $\langle k,u \rangle$!).
If $q$ sees that $R_{pQ}$ contains a tuple at least as ``new'' as $\langle k,u \rangle$ (line~\ref{qfrompp}),
	then
	$q$ successively writes $\langle k,u \rangle$
	to every register $R_{qq'}$ such that $q' \in Q$ (line~\ref{qtorp}),
	and it returns $\langle k,u \rangle$ as the tuple read (line~\ref{qPrt3});
	else $q$ exits \textsc{Thread 2} (so in this case only \textsc{Thread 1} remains).

\textbf{-}\hspace{1mm} Finally, if $q$ does not see that $R_{pQ}$  or  any $R_{qq'}$ contains a tuple at least as ``new'' as $\langle k,u \rangle$ (in lines~\ref{qfromp} and~\ref{rtoq}),
	then $q$ returns $\lw$ (line~\ref{qPrtlw}).

\textbf{Why two parallel threads?} In a nutshell, this is to guarantee the Termination property of $\I{n}$
	in runs where the writer is correct or no reader is {\ml} (this property is required for the implementation to be valid).
It turns out that:
\begin{compactenum}[(A)]
\item if only \textsc{Thread 1} is executed, then a faulty writer can block correct readers even if no reader is {\ml}, and
\item if only \textsc{Thread 2} is executed, then {\ml} readers can block correct readers from returning any value in this thread even if the writer is correct.
\end{compactenum}

But if the writer is correct or no reader is {\ml}, we can show that every read operation by a correct reader is guaranteed to complete with a return value \emph{in one of the two threads}.

It is easy to see why a faulty writer (even one that just crashes) may block a correct reader in \textsc{Thread~1}.
We now explain how {\ml} readers may impede correct readers in \textsc{Thread~2}.

In \textsc{Thread~2} readers must read $R_{pQ}$ at least once (in line~\ref{qfromp}).
Recall that (a) $R_{pQ}$ is an \emph{implemented} \mbox{$\Reg{1}{n-1}$}, and
	(b) we are \emph{only} assuming that this implementation is \emph{valid}.
In particular, if the writer $p$ of $R_{pQ}$ crashes \emph{and} some readers of $R_{pQ}$ are {\ml},
	the implementation of $R_{pQ}$ does \emph{not} guarantee that correct readers complete their operations.
In other words, if $p$ crashes \emph{and} some readers of $R_{pQ}$ are {\ml}, a correct reader $q$ may block while trying to read $R_{pQ}$!

Malicious readers can also prevent a correct reader $q$
	from reading any tuple in~\textsc{Thread~2} as follows.
When $q$ executes $\ru_q()$ the following can occur:
	(1) in line~\ref{rtoq}, $q$ sees that some $R_{q'q}$ contains $\langle k',- \rangle$ with $k'\ge k$ ,
	but 
	(2) in line~\ref{qfrompp} $q$ sees that $R_{pQ}$ does \emph{not} contain $\langle k',- \rangle$ with $k'\ge k$.
We can show that this can occur \emph{only if} at least one of $p$ or $q'$ is {\ml}.
Note that if (1) and (2) indeed occur, then $q$ exits \textsc{Thread~2} \emph{without returning any tuple}.

\input{proof-algo-n-readers}

\subsection{Implementing a \texorpdfstring{\Reg{1}{n}}{} from atomic \texorpdfstring{\Reg{1}{1}s}{}}\label{n-from-1}
	
We now prove that
	in a system with Byzantine process failures,
	there is an implementation of a $\Reg{1}{n}$ from atomic $\Reg{1}{1}$s
	that is linearizable (always) and 
	satisfies the Termination property
	if the writer
	\emph{or} any number of readers, but \emph{not both}, can fail.
	This matches the impossibility result given by Theorem~\ref{Theo-Impossibility-Result} in Section~\ref{Impossibility-Result}.
	More precisely:

\begin{theorem}\label{Theo-Main-Possibility}
For all $n \ge 2$, in a system of $n+1$ processes that are subject to Byzantine failures,
	there is an implementation ${\cal{I}}_n$ of a $\Reg{1}{n}$ from atomic \mbox{$\Reg{1}{1}$s}
	such~that:
\begin{compactitem}
\item \label{regpropertieslX}  \emph{${\cal{I}}_n$ is linearizable.}
\item \label{regpropertieswX} \emph{if the writer is correct or no reader is {\ml}, ${\cal{I}}_n$ satisfies the Termination property.}
\end{compactitem}
\end{theorem} 

\begin{proof}
We must show that for all $n \ge 2$, there is a \emph{valid} implementation ${\cal{I}}_n$ of a $\Reg{1}{n}$ from atomic \mbox{$\Reg{1}{1}$s}.
We prove this by induction on $n$.

\noindent
\textsc{Base Case.} Let $n=2$.
Consider the implementation $\I{2}$ of Theorem~\ref{I-is-valid}.
Since $n=2$, the set $Q$ now contains only one process.
So each register $R_{wQ}$ and $R_{pQ}$ in $\I{2}$
	can be implemented directly by an \emph{atomic} $\Reg{1}{1}$.
Since these are \emph{valid} implementations of $R_{wQ}$ and $R_{pQ}$,
	there is a valid implementation ${\cal{I}}_2$ of a $\Reg{1}{2}$ from
	atomic $\Reg{1}{1}$s.

\noindent
\textsc{Induction Step.} Let $n >2$.
Suppose there is a valid implementation ${\cal{I}}_{n-1}$
	of a $\Reg{1}{n-1}$ that uses only atomic $\Reg{1}{1}$s.
We must show there is a valid implementation ${\cal{I}}_{n}$
	of a $\Reg{1}{n}$ that uses only atomic $\Reg{1}{1}$s.

\smallskip
\noindent
By Theorem~\ref{I-is-valid},
	there is an implementation $\I{n}$ of a $\Reg{1}{n}$
	that uses:
\begin{compactenum}
\item two implemented $\Reg{1}{n-1}$s (namely, registers $R_{wQ}$ and $R_{pQ}$),  and
\item some atomic $\Reg{1}{1}$s
\end{compactenum}

such that $\I{n}$ is valid if the implementations of the $\Reg{1}{n-1}$s $R_{wQ}$ and $R_{pQ}$ are~valid.
Implement $R_{wQ}$ and $R_{pQ}$ in $\I{n}$ using the valid implementation ${\cal{I}}_{n-1}$ 
  	(${\cal{I}}_{n-1}$ exists by our induction hypothesis).
This gives an implementation ${\cal{I}}_{n}$ of a $\Reg{1}{n}$
	that uses only atomic $\Reg{1}{1}$s (because ${\cal{I}}_{n-1}$ uses only atomic $\Reg{1}{1}$s).
Since the implementations of $R_{wQ}$ and $R_{pQ}$ are valid,
	${\cal{I}}_{n}$ is valid.
\end{proof}

For the special case that $n=2$ (i.e., there are only two readers), there is a simple implementation $\Iup{2}$
	that is stronger than the implementation $\Iu{2}$ given by Theorem~\ref{Theo-Main-Possibility}:
	in contrast to~${\cal{I}}_2$, 
	which satisfies Termination if the writer is correct or no reader is {\ml},
	$\Iup{2}$ satisfies Termination \emph{unconditionally};
	in other words $\Iup{2}$ is \emph{wait-free}, and in fact it is \emph{bounded} wait-free (Definition~\ref{bwfdef}).

\begin{restatable}{theorem}{Stwofromone}
\label{S2from1}
The implementation $\Iup{2}$ (given by Algorithm~\ref{1w2r} in Appendix~\ref{whatever})
	is a bounded wait-free linearizable implementation
	of a $\Reg{1}{2}$ from
	atomic $\Reg{1}{1}$s.
\end{restatable}

%% file: code-algo-n-readers.tex

\begin{algorithm}[!t]
\caption{Implementation $\I{n}$ of a $\Reg{1}{n}$
	writable by (an arbitrary) process~$w$ and readable by the $n$ processes in $\{p\}\cup Q$, for $n\ge2$.
It uses two $\Reg{1}{n-1}$s and some $\Reg{1}{1}$s.
}\label{1wnr} 
\ContinuedFloat

\vspace{-3mm}
{\footnotesize
\begin{multicols}{2}

\textsc{Atomic Registers}
\vspace{.7mm}

~~~~$R_{wp}$: 
	$\Reg{1}{1}$;
	initially $(\done,\langle 0,u_0 \rangle)$
	
\vspace{1mm}

~~~~For all processes $q$ and $q'$ in $Q$: 

\hspace{1cm}
	$R_{qq'}$: 
	$\Reg{1}{1}$;
	 initially $\langle 0 , \rinit \rangle$

\vspace{1.3mm}

\textsc{Implemented Registers}
\vspace{.7mm}

\hspace{3.5mm}$R_{wQ}$: 
	$\Reg{1}{n-1}$; initially $(\done,\langle 0,u_0 \rangle)$

\hspace{2.5mm} $R_{pQ}$: 
	$\Reg{1}{n-1}$;
	initially $\langle 0 , \rinit \rangle$

\columnbreak

\textsc{Local variables}

\vspace{.7mm}

~~~~$c$: variable of $w$; initially $0$

~~~~$\lw$: variable of $w$; initially $\langle 0 , \rinit \rangle$

~~~~$\lc$: variable of $p$; initially $0$

\end{multicols}
\hrule
\vspace{-2mm}
\begin{algorithmic}[1]
\begin{multicols}{2}
\Statex
\textsc{\Wu($u$):}  ~~~~~~~~$\triangleright$ executed by the writer $w$
\Indent
\State \label{c} $c \gets c+1$  
\State \label{callW} \textbf{call} \textsc{\wu($\langle c,u \rangle$)}
\State \Return $\Done$
\EndIndent

\columnbreak

\noindent
\textsc{$\Ru$():}\Comment{executed by any reader $r$ in $\{p\}\cup Q$}
\Indent
\State \label{callru}\textbf{call} $\ru_r$()
\State \textbf{if} this call returns some tuple $\langle k,u \rangle$ \textbf{then}
\Indent
\State $\Return$ $u$ 
\EndIndent
\State \textbf{else} \Return $\bot$ 
\EndIndent
\end{multicols}
\vspace{-1mm}
\hrule
\vspace{4mm}
\Statex

\textsc{\wu($\langle k,u \rangle$):} \Comment{executed by $w$ to do its $k$-th write}
\Indent
\State \label{wpP} $R_{wp}\gets (\ready,\lw,\langle k,u \rangle)$
\State \label{wqP} $R_{wQ}\gets (\ready,\lw,\langle k,u \rangle)$
\State \label{wpC} $R_{wp}\gets (\done,\langle k,u \rangle)$
\State \label{wqC} $R_{wQ}\gets (\done,\langle k,u \rangle)$
\State $\lw \gets \langle k,u \rangle$
\State \Return $\Done$
\EndIndent

\vspace*{2mm}

\noindent
\textsc{$\ru_p$():}\Comment{executed by reader $p$}
\Indent
\State \label{pCommit}\textbf{if} $R_{wp}=(\done,\langle k,u \rangle)$ for some $\langle k,u \rangle$ with $k \ge \lc$ \textbf{then}
\Indent
\State \label{tellq} \label{ptoq}$R_{pQ}\gets \langle k,u \rangle$
\State \label{lc} $\lc \gets k$
\State \label{pCrt} \Return $\langle k,u \rangle$
\EndIndent
\State \label{pPre}\textbf{elseif} $R_{wp}=(\ready, \lw, - )$  for some $\lw$ \textbf{then}
\Indent
\State \label{pPrt}\Return $\lw$
\EndIndent
\State \textbf{else} \Return $\bot$ 
\EndIndent

\vspace*{2mm}

\noindent
\textsc{$\ru_q()$:}\Comment{executed by any reader $q$ in $Q$}
\Indent
\State \label{qCommit} \textbf{if} $R_{wQ}=(\done,\langle k,u \rangle)$ for some $\langle k,u \rangle$ \textbf{then}
\Indent
\State \label{qCrt} \Return $\langle k,u \rangle$
\EndIndent

\State \label{qPre} \textbf{elseif} $R_{wQ}=(\ready,\lw,\langle k,u \rangle)$ for some $\lw$ and some $\langle k,u \rangle$ \textbf{then}

\Indent

\State {\color{blue} \textbf{cobegin}}

\textsc{\color{blue} {~~~~// Thread 1}}

\Indent

\State \label{rp} \textbf{repeat forever}

\Indent
\State \label{qCommitp} \textbf{if} $R_{wQ}=(\done,\langle k',- \rangle)$ for some $k'\ge k$ \textbf{then}
\Indent
\State \label{qPrt4}\Return $\langle k,u \rangle$
\EndIndent
\State \label{qPrep} \textbf{if} $R_{wQ}=(\ready, - ,\langle k', - \rangle)$  for some $k' > k$ \textbf{then}
\Indent
\State \label{qPrt5}\Return $\langle k,u \rangle$
\EndIndent
\EndIndent
\EndIndent

\textsc{\color{blue} {~~~~// Thread 2}}

\Indent
\State \label{qfromp}\textbf{if} $R_{pQ}=\langle k',- \rangle$ \text{for some} $k'\ge k$ \textbf{then}
\Indent
\State \label{qtor} \textbf{for every} process $q' \in Q$ \textbf{do} $R_{qq'}\gets \langle k,u \rangle$
\State \label{qPrt2}\Return $\langle k,u \rangle$
\EndIndent

\State \label{rtoq}\textbf{elseif} $R_{q'q}=\langle k',- \rangle$ \text{for some} $q' \in Q$ and some $k'\ge k$ \textbf{then}
\Indent

\State \label{qfrompp} \textbf{if} $R_{pQ}=\langle k',- \rangle$ \text{for some} $k'\ge k$ \textbf{then}
\Indent
\State \label{qtorp} \textbf{for every} process $q' \in Q$ \textbf{do} $R_{qq'}\gets \langle k,u \rangle$
\State \label{qPrt3}\Return $\langle k,u \rangle$
\EndIndent
\State \textbf{else} \textbf{exit} \textsc{Thread 2}
\EndIndent
\State \label{et2} \textbf{else}  \Return \label{qPrtlw} $\lw$
\EndIndent

\State {\color{blue} \textbf{coend}}

\EndIndent

\State \label{qrtbot} \textbf{else} \Return $\bot$

\EndIndent
\end{algorithmic}
}
 \hfill
\end{algorithm}

%% file: proof-algo-n-readers.tex

We now prove the correctness of the $\Reg{1}{n}$ implementation $\I{n}$ given in Figure~\ref{1wnr},
	more precisely, we show that if the $\Reg{1}{n-1}$s that $\I{n}$ uses are {\vd}, then $\I{n}$ is valid (Theorem~\ref{I-is-valid}).
Since this proof may be distracting, in a first reading of the paper 
	a reader may want to skip this proof and go directly to Theorem~\ref{I-is-valid}.

{\large \textbf{Correctness of the implementation \texorpdfstring{$\I{n}$}{In}.}}

We must show that
	$\I{n}$ is \emph{{\vd}} under the assumption that the $\Reg{1}{n-1}$ implementations that it uses,
	namely $R_{wQ}$ and $R_{pQ}$, are also valid.
So in this proof we assume:

\begin{assumption}\label{a0}
The implementations of the $\Reg{1}{n-1}$s $R_{wQ}$ and $R_{pQ}$ that $\I{n}$ uses are valid. 
\end{assumption}

We show that under this asssumption,
	the implementation $\I{n}$ of the $\Reg{1}{n}$ is also valid, that is:

\begin{compactitem}
\item $\I{n}$ is linearizable, and
\item 
If the writer is correct or no reader is {\ml}, $\I{n}$ satisfies the Termination property.
\end{compactitem}

\medskip\noindent
Henceforth, we consider an arbitrary {\run} $E$ of the implementation $\I{n}$ given in Figure~\ref{1wnr}.

By~Assumption~\ref{a0}, the implemented registers $R_{wQ}$ and $R_{pQ}$ that $\I{n}$ uses are linearizable;
	moreover, the atomic registers that $\I{n}$ uses are also (trivially) linearizable.
So operations on these registers appear to take effect instantaneously at some point (the ``linearization point'') in their execution intervals.
Therefore, without loss of generality, we can assume that in the {\run} $E$
	the operations on the registers that $\I{n}$ uses are sequential.

\vspace{1mm}
In the proof, we use the following notation (where $R$ is any atomic or implemented register used by $\I{n}$):

\begin{itemize}

\item ``process $x$ reads $R = v$ in line $\ell$ of $\ru()$''
	means that process $x$ reads register $R$,
	this read returns the value $v$,
	and both occur in line $\ell$ of the read procedure $\ru()$. 

\item ``process $x$ reads $R=u$
	\emph{before}
	process $y$ reads $R'=v$'' means that
	the read operation by~$x$ (which returns $u$) \emph{precedes} the read operation by~$y$ (which returns $v$).

\item ``process $x$ writes $u$ in $R$
	\emph{before}
	process $y$ writes $v$ in $R'$'' means that
	the write $u$ operation by $x$
	\emph{precedes} the write $v$ operation by $y$.
\end{itemize}

We first show that $\I{n}$ is linearizable.
 Then we prove that
 	it satisfies the Termination property if the writer is correct or no reader is {\ml}.

\textbf{Linearizability of \texorpdfstring{$\I{n}$}{In}}. We consider two cases:

 \smallskip
 \textsc{Case 1:} \emph{The writer $w$ of the register implemented by $\I{n}$ is {\ml}.}
 By Definition~\ref{LinearizableByz}, $\I{n}$ is (trivially) linearizable in this case.
 
 \textsc{Case 2:} \emph{The writer $w$ of the register implemented by $\I{n}$ is \emph{not} {\ml}.}
 
 For this case, we now prove that the read and write operations of the implemented register
 	satisfy the linearizability Properties~1 and~2 of Definition~\ref{LinearizableByz}.
In the following:
 
\begin{compactitem}
\item $u_0$ is the initial value of the register that $\I{n}$ implements.
\item For $k\ge1$, $u_k$ denotes the $k$-th value written by $w$ using
	the procedure
	$\Wu()$.
	More precisely, if $w$ calls $\Wu()$ with a value $u$ and this is its $k$-th call of $\Wu()$,
	then~$u_k$~is~$u$.
\item $v_0$ is $\langle 0, u_0 \rangle$.
\item For $k\ge1$, $v_k$ denotes the $k$-th value written by $w$ using
	the procedure
	$\wu()$. 
\end{compactitem}

\begin{observation}\label{vk}
For all $k\ge 0$, $v_k = \tp{k,u_k}$.
\end{observation}

By a slight abuse of notation:
\begin{compactitem}
\item a \emph{write operation} performed by executing the $\Wu()$ or $\wu()$ procedures
	with a value $x$ is denoted $\Wu(x)$ or $\wu(x)$, respectively.
\item A \emph{read operation} performed by executing the $\Ru()$ or $\ru()$ procedures that
	return a value $x$ is denoted $\Ru(x)$ or $\ru(x)$, respectively.
\end{compactitem}

\begin{observation}\label{wwrites}
Let $\wu(v)$ be any write operation by $w$.
Then there is a $k \ge 1$ such that $v=v_k$.
\end{observation}

\begin{observation}\label{regcontentstrongw1}
Let $R\in\{R_{wp}, R_{wQ}\}$.
If $w$ writes $x$ in $R$,
	then $x=(\done, v_k)$ for some $k \ge 1$ or
	$x=(\ready, v_k , v_{k+1})$ 
	for some $k \ge 0$.
\end{observation}

\begin{observation}\label{regcontentstrongp}
Suppose $p$ is not {\ml}.
If $p$ reads $R_{wp}=x$,
	then $x=(\done, v_k)$ or
	$x=(\ready, v_k , v_{k+1})$, for some $k \ge 0$. 
\end{observation}

\begin{observation}\label{regcontentstrongq}
Suppose $q\in Q$ is not {\ml}.
If $q$ reads $R_{wQ}=x$,
then $x=(\done, v_k)$ or
	$x=(\ready, v_k , v_{k+1})$, for some $k \ge 0$. 
\end{observation}

\begin{lemma}\label{pread}
Suppose $p$ is not {\ml}.
Let $\ru_p(v)$ be any read operation by $p$.
Then there is a $k \ge 0$ such that $v=v_k$, and
\begin{compactitem}
\item $p$ reads $R_{wp} = (\done, v_{k})$ in line~\ref{pCommit} of $\ru_p(v)$,\footnote{For brevity, we say that
	``a process $r$ reads or writes a register \emph{in line  $x$ of a $\ru_r(-)$ or a $\wu(-)$ operation}'',
	if it reads or writes this register in line $x$ of the $\ru_r()$ or $\wu()$ \emph{procedure}
	executed to do this $\ru_r(-)$ or $\wu(-)$ operation.}
	or
\item $p$ reads $R_{wp} =(\ready, v_{k}, v_{k+1})$ in line~\ref{pPre} of $\ru_p(v)$. 
\end{compactitem}
\end{lemma}

\begin{proof}
Suppose $p$ is not {\ml}.
Let $\ru_p(v)$ be any read operation by $p$.
Note that $p$ reads $R_{wp}$ in $\ru_p(v)$.
When it does so, by Observation~\ref{regcontentstrongp},
	there are two possible cases:
\begin{enumerate}
\item $p$ reads $R_{wp} = (\done, v_{k})$ for some $k\ge 0$ in line~\ref{pCommit} of $\ru_p(v)$.
 Then $\ru_p(v)$ returns $v_{k}$ in line~\ref{pCrt},
 	i.e., $v=v_k$.
\item $p$ reads $R_{wp} = (\ready, v_{k}, v_{k+1})$ for some $k\ge 0$ in line~\ref{pPre} of $\ru_p(v)$.
 Then $\ru_p(v)$ returns $v_{k}$ in line~\ref{pPrt},
 	i.e., $v=v_k$.
\end{enumerate}
\end{proof}

\begin{lemma}\label{qread}
Suppose $q \in Q$ is not {\ml}.
Let $\ru_q(v)$ be any read operation by $q$.
Then there is a $k \ge 0$ such that $v=v_k$, and
\begin{compactitem}
\item $q$ reads $R_{wQ} = (\done, v_{k})$ in line~\ref{qCommit} of $\ru_q(v)$, 
\item $q$ reads $R_{wQ} =(\ready, v_{k-1}, v_{k})$ in line~\ref{qPre} of $\ru_q(v)$, or 
\item $q$ reads $R_{wQ} =(\ready, v_{k}, v_{k+1})$ in line~\ref{qPre} of $\ru_q(v)$. 
\end{compactitem}
\end{lemma}

\begin{proof}
Suppose $q \in Q$ is not {\ml}.
Let $\ru_q(v)$ be any read operation by $q$.
Note that $q$ reads $R_{wQ}$ in $\ru_q(v)$.
When it does so, by Observation~\ref{regcontentstrongq},
	there are two possible cases:
\begin{enumerate}
\item $q$ reads $R_{wQ} = (\done, v_{k})$ for some $k\ge 0$ in line~\ref{qCommit} of $\ru_q(v)$.
Then $\ru_q(v)$ returns $v_k$ in line~\ref{qCrt},
	i.e., $v=v_k$.
\item $q$ reads $R_{wQ} = (\ready, v_{k}, v_{k+1})$ for some $k\ge 0$ in line~\ref{qPre} of $\ru_q(v)$.
Then there are two subcases:
	\begin{compactenum}
	\item $\ru_q(v)$ returns $v_{k+1}$ in line~\ref{qPrt4},~\ref{qPrt5},~\ref{qPrt2}, or~\ref{qPrt3},
		i.e., $v=v_{k+1}$.
	Let $k'=k+1$.
	Then in this case, $q$ reads $R_{wQ} =(\ready, v_{k'-1}, v_{k'})$ in line~\ref{qPre} of $\ru_q(v)$ and $v=v_{k'}$.

	\item $\ru_q(v)$ returns $v_{k}$  in line~\ref{qPrtlw},
	i.e., $v=v_{k}$.
	\end{compactenum}
\end{enumerate}
\end{proof}

\begin{observation}\label{monowpw}
Let $R$ be any register in $\{R_{wp}, R_{wQ}\}$.

\begin{enumerate}[(1)]
\item \label{monowpw1} If $w$ writes $(\ready, v_{k-1}, v_k)$ in $R$ before $w$ writes $(\done,v_{k'})$ in $R$,
	then \mbox{$k \le k'$}.

\item \label{monowpw2} If $w$ writes $(\ready, v_{k-1}, v_k)$ in $R$ before $w$ writes $(\ready, v_{k'-1}, v_{k'})$ in $R$, 
	then $k < k'$.

\item \label{monowpw2.5} If $w$ writes $(\done,v_k)$ in $R$ before $w$ writes $(\done,v_{k'})$ in $R$,
	then $k < k'$.

\item \label{monowpw3} If $w$ writes $(\done,v_k)$ in $R$ before $w$ writes $(\ready, v_{k'-1}, v_{k'})$ in $R$,
	then $k < k'$.
\end{enumerate}
\end{observation}

\begin{observation}\label{monowpwr}
Let $R$ be any register in $\{R_{wp}, R_{wQ}\}$.
Suppose $r \in \{p\}\cup Q$ is not {\ml}.
\begin{enumerate}[(1)]
\item \label{monowpwr1} If $w$ writes $(\ready, v_{k-1}, v_k)$ in $R$ before $r$ reads $(\done,v_{k'})$ in $R$,
	then \mbox{$k \le k'$}.

\item \label{monowpwr2} If $w$ writes $(\ready, v_{k-1}, v_k)$ in $R$ before $r$ reads $(\ready, v_{k'-1}, v_{k'})$ in $R$, 
	then \mbox{$k \le k'$}.

\item \label{monowpwr2.5} If $w$ writes $(\done,v_k)$ in $R$ before $r$ reads $(\done,v_{k'})$ in $R$,
	then \mbox{$k \le k'$}.

\item \label{monowpwr3} If $w$ writes $(\done,v_k)$ in $R$ before $r$ reads $(\ready, v_{k'-1}, v_{k'})$ in $R$,
	then $k < k'$.
\end{enumerate}
\end{observation}

\begin{observation}\label{monowp}
Let $R$ be any register in $\{R_{wp}, R_{wQ}\}$.
Suppose $r$ and $r'$ are non-{\ml} processes in $\{p\} \cup Q$.

\begin{enumerate}[(1)]
\item \label{monowp1} If $r$ reads $R=(\ready, v_{k-1}, v_k)$ before $r'$ reads $R=(\done,v_{k'})$,
	then \mbox{$k \le k'$}.

\item \label{monowp2} If $r$ reads $R=(\ready, v_{k-1}, v_k)$ before $r'$ reads $R=(\ready, v_{k'-1}, v_{k'})$, 
	then \mbox{$k \le k'$}.

\item \label{monowp2.5} If $r$ reads $R=(\done,v_k)$ before $r'$ reads $R=(\done,v_{k'})$,
	then \mbox{$k \le k'$}.

\item \label{monowp3} If $r$ reads $R=(\done,v_k)$ before $r'$ reads $R=(\ready, v_{k'-1}, v_{k'})$,
	then $k < k'$.
\end{enumerate}
\end{observation}

\textsc{Proof of linearizability Property 1.}
We now prove that the write and read operations of the register that $\I{n}$ implements satisfy Property~1 of Definition~\ref{LinearizableByz}, i.e., processes read the ``current'' value of the register.
 To do so,
 	we first prove this for the writes and reads of the lower-level procedures $\wu()$ and $\ru_r()$ for all readers $r$
	(Lemma~\ref{p1r}),
 	and then prove it for the writes and reads of the high-level procedures $\Wu()$ and $\Ru()$ (Lemma~\ref{Cp1}).

\begin{lemma}\label{p1r}
If $\ru_r(v)$ is a read operation by a non-{\ml} process $r\in \{p\}\cup Q$ then:
\begin{compactitem}
	\item there is a $\wu(v)$ operation that immediately precedes $\ru_r(v)$ or is concurrent with $\ru_r(v)$, or
	\item $v = v_0$ and no $\wu(-)$ operation precedes $\ru_r(v)$.
\end{compactitem}
\end{lemma}

\begin{proof}
Suppose $r\in\{ p \} \cup Q$ is not {\ml}.
Let $\ru_r(v)$ be any read operation by $r$.

 By Lemmas~\ref{pread} and~\ref{qread}, $v = v_k$ for some $k \ge 0$.
 We now show that:
 \begin{compactitem}
 \item if $k=0$ then no $\wu(-)$ operation precedes $\ru_r(v_k)$, and
 \item if $k>0$ then a $\wu(v_{k})$ operation immediately precedes~$\ru_r(v_k)$ or is concurrent with $\ru_r(v_k)$.
 \end{compactitem}
\vspace{2mm}
There are two cases: $r=p$~or~$r \in Q$.
\begin{itemize}
\item Case 1: $r=p$.
	By Lemma~\ref{pread},
	there are two cases:
		
	\begin{enumerate}[1)]
 	\item $p$ reads $R_{wp} = (\done, v_{k})$ in line~\ref{pCommit} of $\ru_p(v_k)$.
	There are two cases:
		\begin{enumerate}[i.]
		\item $k=0$.
			Suppose, for contradiction, 
			that there is a $\wu(v)$ operation that precedes $\ru_p(v_0)$.
			By Observation~\ref{wwrites},
			$v = v_i$ for some $i \ge 1$.
		So $w$ writes $(\done, v_{i})$ into $R_{wp}$ in line~\ref{wpC} of $\wu(v_i)$ before $p$ reads $R_{wp} = (\done, v_{0})$ in line~\ref{pCommit} of $\ru_p(v_k)$.
		By Observation~\ref{monowpwr}(\ref{monowpwr2.5}),
			$i \le 0$ --- a contradiction.
		So no $\wu(-)$ operation~precedes~$\ru_p(v_0)$.

		\item $k>0$.
		Then $w$ writes $(\done,v_k)$ into $R_{wp}$ in line~\ref{wpC} of $\wu(v_k)$ 
		before $p$ reads $R_{wp}= (\done,v_k)$ in $\ru_p(v_k)$.
		So the $\wu(v_k)$ operation precedes $\ru_p(v_k)$ or is concurrent with $\ru_p(v_k)$.
		We now show that if $\wu(v_k)$ precedes $\ru_p(v_k)$,
		then $\wu(v_k)$ immediately precedes $\ru_p(v_k)$.
		Suppose, for contradiction, 
		that $\wu(v_k)$ precedes $\ru_p(v_k)$ but does not immediately precede $\ru_p(v_k)$.
		Then there is a $\wu(v_i)$ operation that immediately precedes $\ru_p(v_k)$.
		Clearly, the $\wu(v_k)$ operation precedes the $\wu(v_i)$ operation, and so $i > k$.
		Furthermore, $w$ writes $(\done, v_{i})$ into $R_{wp}$ in line~\ref{wpC} of $\wu(v_i)$ before $p$ reads $R_{wp} = (\done, v_{k})$ in line~\ref{pCommit} of $\ru_p(v_k)$.
		By Observation~\ref{monowpwr}(\ref{monowpwr2.5}),
			$i \le k$ --- a contradiction.
		 Therefore the $\wu(v_k)$ operation immediately precedes $\ru_p(v_k)$ or is concurrent with $\ru_p(v_k)$.
		\end{enumerate}

	\item $p$ reads $R_{wp} =(\ready, v_{k}, v_{k+1})$ in line~\ref{pPre} of $\ru_p(v_k)$.
	Then this read occurs \emph{after}
	$w$ writes $(\ready, v_{k}, v_{k+1})$ in $R_{wp}$ in line~\ref{wpP} of the $\wu(v_{k+1})$ operation.
	Furthermore, by Observation~\ref{monowpwr}(\ref{monowpwr3}), 
	this read occurs \emph{before} $w$ writes $(\done, v_{k+1})$ in $R_{wp}$
	in line~\ref{wpC} of the $\wu(v_{k+1})$ operation.
	Therefore the $\wu(v_{k+1})$ operation is concurrent with $\ru_p(v_k)$.
	There are two cases:

	\begin{enumerate}[i.]
	\item $k=0$.
	Since $\wu(v_1)$ is concurrent with $\ru_p(v_0)$,
	no $\wu(-)$ operation precedes $\ru_p(v_0)$.

	\item $k>0$. 
	Since $\wu(v_{k+1})$ is concurrent with $\ru_p(v_k)$, 
		$\wu(v_{k})$ immediately precedes~$\ru_p(v_k)$ or is concurrent with $\ru_p(v_k)$.
	\end{enumerate}
\end{enumerate}

\item Case 2: $r = q \in Q$.
	By Lemma~\ref{qread},
	there are three cases:
\begin{enumerate}[1)]
\item $q$ reads $R_{wQ} = (\done, v_{k})$ in line~\ref{qCommit} of $\ru_q(v_k)$.
	There are two cases:
		\begin{enumerate}[i.]
		
		\item $k=0$.
			Suppose, for contradiction, 
			that there is a $\wu(v)$ operation that precedes $\ru_p(v_0)$.
			By Observation~\ref{wwrites},
			$v = v_i$ for some $i \ge 1$.
		So $w$ writes $(\done, v_{i})$ into $R_{wQ}$ in line~\ref{wqC} of $\wu(v_i)$ before $q$ reads $R_{wQ} = (\done, v_{0})$ in line~\ref{qCommit} of $\ru_q(v_k)$.
		By Observation~\ref{monowpwr}(\ref{monowpwr2.5}),
			$i \le 0$ --- a contradiction.
		So no $\wu(-)$ operation~precedes~$\ru_q(v_0)$.
		
		\item $k>0$.
		Then $w$ writes $(\done,v_k)$ into $R_{wQ}$ in line~\ref{wqC} of $\wu(v_k)$
		 before $q$ reads $R_{wQ}= (\done,v_k)$ in line~\ref{qCommit} of $\ru_q(v_k)$.
		So the $\wu(v_k)$ operation precedes $\ru_q(v_k)$ or is concurrent with $\ru_q(v_k)$.
		We now show that if $\wu(v_k)$ precedes $\ru_q(v_k)$,
		then $\wu(v_k)$ immediately precedes $\ru_q(v_k)$.
		Suppose, for contradiction, 
			that $\wu(v_k)$ precedes $\ru_q(v_k)$ but does not immediately precede $\ru_q(v_k)$.
		Then there is a $\wu(v_i)$ operation
		that immediately precedes $\ru_p(v_k)$.
		Clearly, the $\wu(v_k)$ operation precedes the $\wu(v_i)$ operation, and so $i > k$.		
		Furthermore, $w$ writes $(\done, v_{i})$ into $R_{wQ}$ in line~\ref{wqC} of $\wu(v_i)$ before $q$ reads $R_{wQ} = (\done, v_{k})$ in line~\ref{qCommit} of $\ru_q(v_k)$.
		By Observation~\ref{monowpwr}(\ref{monowpwr2.5}),
			$i \le k$ --- a contradiction.
		 Therefore the $\wu(v_k)$ operation immediately precedes $\ru_q(v_k)$ or is concurrent with $\ru_q(v_k)$.
		\end{enumerate}
		
\item $q$ reads $R_{wQ} =(\ready, v_{k-1}, v_{k})$ in line~\ref{qPre} of $\ru_q(v_k)$.
	Then this read occurs \emph{after}
	$w$ writes $(\ready, v_{k-1}, v_{k})$ in $R_{wQ}$ in line~\ref{wqP} of the $\wu(v_k)$ operation.
	Furthermore, by Observation~\ref{monowpwr}(\ref{monowpwr3}),
	this read occurs \emph{before} $w$ writes $(\done,v_k)$ in $R_{wQ}$
	in line~\ref{wqC} of the $\wu(v_k)$ operation.
	Therefore the $\wu(v_k)$ operation is concurrent with $\ru_q(v_k)$.

\item $q$ reads $R_{wQ} =(\ready, v_{k}, v_{k+1})$ in line~\ref{qPre} of $\ru_q(v_k)$.
	Then this read occurs after
	$w$ writes $(\ready, v_{k}, v_{k+1})$ in $R_{wQ}$ in line~\ref{wqP} of the $\wu(v_{k+1})$ operation.
	Furthermore, by Observation~\ref{monowpwr}(\ref{monowpwr3}),
	this read occurs before $w$ writes
	$(\done,v_{k+1})$ in $R_{wQ}$ in line~\ref{wqC} of the $\wu(v_{k+1})$ operation.
	Therefore the $\wu(v_{k+1})$ operation is concurrent with $\ru_q(v_k)$.
	There are two cases:
	\begin{enumerate}[i.]
	
	\item $k=0$.
	Since $\wu(v_1)$ is concurrent with $\ru_q(v_0)$,
	no $\wu(-)$ operation precedes $\ru_q(v_0)$.

	\item $k>0$.
	Since $\wu(v_{k+1})$ is concurrent with $\ru_q(v_k)$, 
		$\wu(v_{k})$ is concurrent with $\ru_q(v_k)$ or immediately precedes $\ru_q(v_k)$.
	\end{enumerate}
\end{enumerate}
\end{itemize}
\end{proof}

We now prove that the write and read operations
	of the high-level procedures $\Wu()$ and $\Ru()$
	satisfy Property~1 of Definition~\ref{LinearizableByz}.

By Observation~\ref{vk}, Lemmas~\ref{pread} and~\ref{qread}, and the code of the procedure $\Ru()$,
	we have:
\begin{observation}\label{integrityR}
If $\Ru(u)$ is an operation by a non-{\ml} process $r \in \{p\}\cup Q$,
	then $u=u_k$ for some $k \ge 0$.
\end{observation}

\begin{observation}\label{Rr}
If $\Ru(u_k)$ is an operation by a non-{\ml} process $r\in\{p\}\cup Q$,
then $r$ invokes and completes a $\ru_r(v_k)$ operation in $\Ru(u_k)$.
\end{observation}

\begin{observation}\label{Ww}
If $\Wu(u_k)$ is a completed operation by $w$,
	then $w$ invokes and completes a $\wu(v_k)$ operation in $\Wu(u_k)$.
\end{observation}

We now prove that the $\Wu(-)$ and $\Ru(-)$ operations satisfy Property~1 of Definition~\ref{LinearizableByz}.

\begin{lemma}\label{Cp1} \emph{[\textbf{Property 1}: Reading a ``current'' value]}

If $\Ru(u)$ is a read operation by a non-{\ml} process $r\in\{ p \} \cup Q$
	then:
\begin{compactitem}	
	\item there is a $\Wu(u)$ operation that immediately precedes $\Ru(u)$ or is concurrent with $\Ru(u)$, or	
	\item $u = u_0$ and no $\Wu(-)$ operation precedes $\Ru(u)$.
\end{compactitem}
\end{lemma}

\begin{proof}
Let $\Ru(u)$ be any read operation by a non-{\ml} process $r\in\{ p \} \cup Q$.
By Observation~\ref{integrityR},
	$u=u_k$ for some $k\ge 0$.
There are two cases:
	\begin{enumerate}[(1)]
	\item $k=0$.
		Suppose, for contradiction, that a $\Wu(u_i)$ operation precedes $\Ru_r(u_0)$.
		Note that $i \ge 1$. 
		By Observations~\ref{Ww} and~\ref{Rr}, 
		a $\wu(v_i)$ operation precedes a $\ru_r(v_0)$ operation.
		Since process $r$ is not {\ml},
			by Lemma~\ref{p1r},
		 there is no $\wu(-)$ operation that precedes $\ru_r(v_0)$ --- a contradiction.

	\item $k>0$.
		By Observation~\ref{Rr},
		$r$ invokes and completes a $\ru_r(v_k)$ operation in $\Ru(u_k)$.
		Since $k>0$,
		 by Lemma~\ref{p1r},
		 there is a $\wu(v_{k})$ operation that immediately precedes $\ru_r(v_k)$ or is concurrent with $\ru_r(v_k)$.
		 Let $\Wu(u_{k})$ be the operation in which $w$ invokes the $\wu(v_{k})$ operation.
		Since $\wu(v_k)$ immediately precedes $\ru_r(v_k)$ or is concurrent with $\ru_r(v_k)$,
		the $\Wu(u_{k})$ operation immediately precedes $\Ru(u_k)$ or is concurrent with $\Ru(u_k)$.
	\end{enumerate}
\end{proof}

\textsc{Proof of linearizability Property 2.}
We now prove that the write and read operations of the register that $\I{n}$ implements satisfy Property~2 of Definition~\ref{LinearizableByz},
	i.e., we prove that there are no ``new-old'' inversions in the values that processes read.
 To do so,
 	we first prove this for the writes and reads of the lower-level procedures $\wu()$ and $\ru_r()$ for all readers $r$
	(Lemma~\ref{p2r}),
 	and then prove it for the writes and reads of the high-level procedures $\Wu()$ and $\Ru()$ (Lemma~\ref{Cp2}).

We first show that there are no ``new-old'' inversions in the consecutive reads of process $p$.

\begin{lemma}\label{monopread}
Suppose $p$ is not {\ml}.
If $\ru_p(v_k)$ and $\ru_p(v_{k'})$ are read operations by $p$,
	and $\ru_p(v_k)$ precedes $\ru_p(v_{k'})$, then \mbox{$k \le k'$}.
\end{lemma}

\begin{proof}
Suppose $p$ is not {\ml}.
Let $\ru_p(v_k)$ and $\ru_p(v_{k'})$ be read operations by $p$ such that $\ru_p(v_k)$ precedes $\ru_p(v_{k'})$.
By Lemma~\ref{pread}, the following occurs:

\begin{compactenum}
	\item $p$ reads $R_{wp} = (\done, v_{k})$ \ilns{pCommit} $\ru_p(v_k)$,
	or
	\item $p$ reads $R_{wp} =(\ready, v_k, v_{k+1})$ \ilns{pPre} $\ru_p(v_k)$,
\end{compactenum}

\emph{before} the following occurs:

\begin{compactenum}
	\item $p$ reads $R_{wp} = (\done, v_{k'})$ \ilns{pCommit} $\ru_p(v_{k'})$,
	or
	\item $p$ reads $R_{wp} =(\ready, v_{k'}, v_{k'+1})$ \ilns{pPre} $\ru_p(v_{k'})$.
\end{compactenum}

\smallskip
So there are four possible cases:

\begin{compactenum}
\item $p$ reads $R_{wp} = (\done, v_{k})$ \iln{pCommit} $\ru_p(v_k)$
	before
	$p$ reads $R_{wp} = (\done, v_{k'})$ \iln{pCommit} $\ru_p(v_{k'})$. 
 By Observation~\ref{monowp}(\ref{monowp2.5}),
 	 \mbox{$k \le k'$}.

\item $p$ reads $R_{wp} = (\done, v_{k})$ \iln{pCommit} $\ru_p(v_k)$
	before
	$p$ reads $R_{wp} =(\ready, v_{k'}, v_{k'+1})$ \iln{pPre} $\ru_p(v_{k'})$.
 By Observation~\ref{monowp}(\ref{monowp3}),
 	 $k < k'+1$. 
 	 So \mbox{$k \le k'$}.

\item $p$ reads $R_{wp} =(\ready, v_k, v_{k+1})$ \iln{pPre} $\ru_p(v_k)$
	before
	$p$ reads $R_{wp} = (\done, v_{k'})$ \iln{pCommit} $\ru_p(v_{k'})$.
 By Observation~\ref{monowp}(\ref{monowp1}),
 	 $k+1 \le k'$.
 	 So \mbox{$k \le k'$}.

\item $p$ reads $R_{wp} =(\ready, v_k, v_{k+1})$
	\iln{pPre} $\ru_p(v_k)$
	before $p$ reads $R_{wp} =(\ready, v_{k'}, v_{k'+1})$
	\iln{pPre} $\ru_p(v_{k'})$.
By Observation~\ref{monowp}(\ref{monowp2}), 
	$k+1 \le k'+1$. 
	So \mbox{$k \le k'$}. 
\end{compactenum}
\end{proof}

To prove that there are no ``new-old'' inversions between the reads of $p$ and those of any reader $q \in Q$, and also between the reads of any pair of readers $q, q' \in Q$, we first make some straightforward observations that are clear from the code of $\I{n}$.
We first note that the counters of the tuples in registers $R_{pQ}$ and $R_{qq'}$ do not decrease.

\begin{observation}\label{monopq}
Suppose $p$ is not {\ml}.
	If $p$ writes $v_k$ in $R_{pQ}$ before $p$ writes $v_{k'}$ in $R_{pQ}$,
	then \mbox{$k \le k'$.}
\end{observation}

\begin{observation}\label{monopqwr}
Suppose $p$ and $q\in Q$ are not {\ml}.
	If $p$ writes $v_k$ in $R_{pQ}$ before $q$ reads $R_{pQ}=v_{k'}$,
	then \mbox{$k \le k'$.}
\end{observation}

\begin{observation}\label{monorq}
Suppose $q\in Q$ is not {\ml}.
For all processes $q'\in Q$,
	if $q$ writes $v_k$ in $R_{qq'}$ before $q$ writes $v_{k'}$ in $R_{qq'}$,
	then \mbox{$k \le k'$.}
\end{observation}

\begin{observation}\label{monorqwr}
Suppose $q\in Q$ and $q'\in Q$ are not {\ml}.
	If $q$ writes $v_k$ in $R_{qq'}$ before $q'$ reads $R_{qq'}=v_{k'}$,
	then \mbox{$k \le k'$.}
\end{observation}

The following observations relate the counters of the tuples that $w$ succesively writes in registers $R_{wp}$ 
and $R_{wQ}$.

\begin{observation}\label{wpwQw}
\begin{enumerate}[(1)]
\item \label{wpwQw1} If
	$w$ writes $(\ready, v_{k-1}, v_k)$ in $R_{wp}$ before $w$ writes $(\ready, v_{k'-1}, v_{k'})$ in $R_{wQ}$,
	then \mbox{$k \le k'$}.
\item \label{wpwQw2} If $w$ writes $(\ready, v_{k-1}, v_k)$ in $R_{wp}$ before $w$ writes $(\done,v_{k'})$ in $R_{wQ}$, 
	then \mbox{$k \le k'$.}

\item \label{wpwQw1.5} If $w$ writes $(\done,v_k)$ in $R_{wp}$ before $w$ writes $(\ready, v_{k'-1}, v_{k'})$ in $R_{wQ}$,
	then \mbox{$k < k'$.}

\item \label{wpwQw3} If $w$ writes $(\done,v_k)$ in $R_{wp}$ before $w$ writes $(\done,v_{k'})$ in $R_{wQ}$,
	then \mbox{$k \le k'$}.

\item \label{wpwQw4} If $w$ writes $(\ready, v_{k-1}, v_k)$ in $R_{wQ}$ before $w$ writes $(\ready, v_{k'-1}, v_{k'})$ in $R_{wp}$,
	then \mbox{$k < k'$.}

\item \label{wpwQw4.5} If $w$ writes $(\ready, v_{k-1}, v_k)$ in $R_{wQ}$ before $w$ writes $(\done,v_{k'})$ in $R_{wp}$,
	then \mbox{$k \le k'$.}

\item \label{wpwQw6} If $w$ writes $(\done,v_k)$ in $R_{wQ}$ before $w$ writes $(\ready, v_{k'-1}, v_{k'})$ in $R_{wp}$,
	then \mbox{$k < k'$.}

\item \label{wpwQw5} If $w$ writes $(\done,v_k)$ in $R_{wQ}$ before $w$ writes $(\done,v_{k'})$ in $R_{wp}$, 
	then \mbox{$k < k'$.}
\end{enumerate}
\end{observation}

The next observations relate the counters of the tuples that $p$ and processes $q \in Q$ read
	from $R_{wp}$ and $R_{wQ}$, respectively.

\begin{observation}\label{wpwQ}
Suppose $p$ and $q \in Q$ are not {\ml}.
\begin{enumerate}[(1)]
\item \label{wpwQ1} If $p$ reads $R_{wp}=(\ready, v_{k-1}, v_k)$ before $q$ reads $R_{wQ}=(\ready, v_{k'-1}, v_{k'})$,
	then \mbox{$k \le k'$.}
\item \label{wpwQ2} If $p$ reads $R_{wp}=(\ready, v_{k-1}, v_k)$ before $q$ reads $R_{wQ}=(\done,v_{k'})$, 
	then \mbox{$k-1 \le k'$.}

\item \label{wpwQ1.5} If $p$ reads $R_{wp}=(\done,v_k)$ before $q$ reads $R_{wQ}=(\ready, v_{k'-1}, v_{k'})$,
	then \mbox{$k \le k'$.}

\item \label{wpwQ3} If $p$ reads $R_{wp}=(\done,v_k)$ before $q$ reads $R_{wQ}=(\done,v_{k'})$,
	then \mbox{$k \le k'$.}

\item \label{wpwQ4} If $q$ reads $R_{wQ}=(\ready, v_{k-1}, v_k)$ before $p$ reads $R_{wp}=(\ready, v_{k'-1}, v_{k'})$,
	then \mbox{$k \le k'$.}

\item \label{wpwQ4.5} If $q$ reads $R_{wQ}=(\ready, v_{k-1}, v_k)$ before $p$ reads $R_{wp}=(\done,v_{k'})$,
	then \mbox{$k \le k'$.}

\item \label{wpwQ6} If $q$ reads $R_{wQ}=(\done,v_k)$ before $p$ reads $R_{wp}=(\ready, v_{k'-1}, v_{k'})$,
	then \mbox{$k+1 \le k'$.}

\item \label{wpwQ5} If $q$ reads $R_{wQ}=(\done,v_k)$ before $p$ reads $R_{wp}=(\done,v_{k'})$, 
	then \mbox{$k \le k'$.}
\end{enumerate}
\end{observation}

Now we prove that there is no ``new-old'' inversion for a read by $p$ that precedes a read by a process $q \in Q$.

\begin{lemma} \label{nopq}
If $\ru_p(v_k)$ and $\ru_q(v_{k'})$ are read operations by non-{\ml} processes $p$ and $q\in Q$ respectively,
	and $\ru_p(v_k)$ precedes $\ru_q(v_{k'})$, then \mbox{$k \le k'$}.
\end{lemma}

\begin{proof}
Suppose processes $p$ and $q \in Q$ are not {\ml}.
Let $\ru_p(v_k)$ and $\ru_q(v_{k'})$ be read operations by $p$ and $q$ respectively,
	 such that $\ru_p(v_k)$ precedes $\ru_q(v_{k'})$.
By Lemmas~\ref{pread} and~\ref{qread}, the following occurs:

\begin{compactenum}
	\item $p$ reads $R_{wp} = (\done, v_{k})$ \ilns{pCommit} $\ru_p(v_k)$,
	or
	\item $p$ reads $R_{wp} =(\ready, v_k, v_{k+1})$ \ilns{pPre} $\ru_p(v_k)$
\end{compactenum}

\emph{before} the following occurs:

\begin{compactenum}
	\item $q$ reads $R_{wQ} = (\done, v_{k'})$ \ilns{qCommit} $\ru_q(v_{k'})$, or
	
	\item $q$ reads $R_{wQ} =(\ready, v_{k'-1}, v_{k'})$ \ilns{qPre} $\ru_q(v_{k'})$, or

	\item $q$ reads $R_{wQ} =(\ready, v_{k'}, v_{k'+1})$ \ilns{qPre} $\ru_q(v_{k'})$.
\end{compactenum}

\smallskip
So there are six possible cases:

\begin{compactenum}

\item $p$ reads $R_{wp} = (\done, v_{k})$
	\iln{pCommit} $\ru_p(v_k)$
	before $q$ reads $R_{wQ} = (\done, v_{k'})$
	\iln{qCommit} $\ru_q(v_{k'})$.
 By Observation~\ref{wpwQ}(\ref{wpwQ3}),
 	 \mbox{$k \le k'$}.

\item $p$ reads $R_{wp} = (\done, v_{k})$
	\iln{pCommit} $\ru_p(v_k)$
	before $q$ reads $R_{wQ} =(\ready, v_{k'-1}, v_{k'})$
	\iln{qPre} $\ru_q(v_{k'})$.
 By Observation~\ref{wpwQ}(\ref{wpwQ1.5}),
 	 \mbox{$k \le k'$}.

\item $p$ reads $R_{wp} = (\done, v_{k})$
	\iln{pCommit} $\ru_p(v_k)$
	before $q$ reads $R_{wQ} =(\ready, v_{k'}, v_{k'+1})$
	\iln{qPre} $\ru_q(v_{k'})$.
 By Observation~\ref{wpwQ}(\ref{wpwQ1.5}),
 	 $k \le k'+1$.
 
\begin{compactenum}[i.]
\item $k < k'+1$. Then $k\le k'$.

\item $k = k'+1$. 
	We now show that this case is impossible.
	Since $k = k'+1$,
	$q$ reads $R_{wQ} =(\ready, v_{k-1}, v_{k})$ in line~\ref{qPre} of $\ru_q(v_{k-1})$,
	 and $\ru_q(v_{k-1})$ returns in line~\ref{qPrtlw}.
So $q$ read $R_{pQ}$ in line~\ref{qfromp} of $\ru_q(v_{k-1})$ before $\ru_q(v_{k-1})$ returns in line~\ref{qPrtlw}.

Since $p$ reads $R_{wp} = (\done, v_{k})$ in $\ru_p(v_k)$,
	$p$ writes $v_{k}$ in $R_{pQ}$ in line~\ref{ptoq} of $\ru_p(v_k)$.
Since $\ru_p(v_k)$ precedes $\ru_q(v_{k-1})$, 
	 $p$ writes $v_{k}$ in $R_{pQ}$ in $\ru_p(v_k)$ before $q$ reads $R_{pQ}$ in line~\ref{qfromp} of $\ru_q(v_{k-1})$.
By Observation~\ref{monopqwr}, 
	$q$ reads $R_{pQ}=v_{\ell}$,
	 for some $\ell \ge k$,
	 in line~\ref{qfromp} of $\ru_q(v_{k-1})$.
So $\ru_q(v_{k-1})$ returns $v_k$ in line~\ref{qPrt2}, rather than in line~\ref{qPrtlw} --- a contradiction.
\end{compactenum}

\item $p$ reads $R_{wp} =(\ready, v_{k}, v_{k+1})$
	\iln{pPre} $\ru_p(v_k)$
	before $q$ reads $R_{wQ} = (\done, v_{k'})$
	\iln{qCommit} $\ru_q(v_{k'})$.
 By Observation~\ref{wpwQ}(\ref{wpwQ2}), $(k+1)-1 \le k'$. So \mbox{$k \le k'$}.

\item $p$ reads $R_{wp} =(\ready, v_{k}, v_{k+1})$
	\iln{pPre} $\ru_p(v_k)$
	before $q$ reads $R_{wQ} =(\ready, v_{k'-1}, v_{k'})$
	\iln{qPre} $\ru_q(v_{k'})$.
By Observation~\ref{wpwQ}(\ref{wpwQ1}), $k+1 \le k'$. So \mbox{$k \le k'$}.

\item $p$ reads $R_{wp} =(\ready, v_{k}, v_{k+1})$
	\iln{pPre} $\ru_p(v_k)$
	before $q$ reads $R_{wQ} =(\ready, v_{k'}, v_{k'+1})$
	\iln{qPre} $\ru_q(v_{k'})$.
By Observation~\ref{wpwQ}(\ref{wpwQ1}), $k+1 \le k'+1$. So \mbox{$k \le k'$}.
\end{compactenum}
\end{proof}

Now we prove that there is no ``new-old'' inversion for a read by a process $q \in Q$ that precedes a read by $p$.

\begin{lemma}\label{noqp}
	If $\ru_q(v_k)$ and $\ru_p(v_{k'})$ are read operations by non-{\ml} processes $q\in Q$ and $p$ respectively,
	and $\ru_q(v_k)$ precedes $\ru_p(v_{k'})$, then \mbox{$k \le k'$}.
\end{lemma}

\begin{proof}
Suppose processes $q\in Q$ and $p$ are not {\ml}.
Let $\ru_q(v_k)$ and $\ru_p(v_{k'})$ be two read operations by $q$ and $p$ respectively,
	such that $\ru_q(v_k)$ precedes $\ru_p(v_{k'})$
By Lemmas~\ref{pread} and~\ref{qread}, the following occurs:
\begin{compactenum}
	\item $q$ reads $R_{wQ} = (\done, v_{k})$ \ilns{qCommit} $\ru_q(v_k)$,
	or
	
	\item $q$ reads $R_{wQ} =(\ready, v_{k-1}, v_{k})$ \ilns{qPre} $\ru_q(v_k)$,
	or

	\item $q$ reads $R_{wQ} =(\ready, v_{k}, v_{k+1})$ \ilns{qPre} $\ru_q(v_k)$
\end{compactenum}

\emph{before} the following occurs:

\begin{compactenum}
	\item $p$ reads $R_{wp} = (\done, v_{k'})$ \ilns{pCommit} $\ru_p(v_{k'})$,
	or
	\item $p$ reads $R_{wp} =(\ready, v_{k'}, v_{k'+1})$ \ilns{pPre} $\ru_p(v_{k'})$.
\end{compactenum}

\smallskip
So there are six possible cases:

\begin{compactenum}

\item $q$ reads $R_{wQ} = (\done, v_{k})$ \iln{qCommit} $\ru_q(v_k)$
	before $p$ reads $R_{wp} = (\done, v_{k'})$ \iln{pCommit} $\ru_p(v_{k'})$.
 By Observation~\ref{wpwQ}(\ref{wpwQ5}),
 	 \mbox{$k \le k'$}.

 \item $q$ reads $R_{wQ} = (\done, v_{k})$ \iln{qCommit} $\ru_q(v_k)$
 	before $p$ reads $R_{wp} =(\ready, v_{k'}, v_{k'+1})$ \iln{pPre} $\ru_p(v_{k'})$.
 By Observation~\ref{wpwQ}(\ref{wpwQ6}),
 	 $k+1 \le k'+1$. So \mbox{$k \le k'$}.

 \item $q$ reads $R_{wQ} =(\ready, v_{k-1}, v_{k})$ \iln{qPre} $\ru_q(v_k)$
 	before $p$ reads $R_{wp} = (\done, v_{k'})$ \iln{pCommit} $\ru_p(v_{k'})$.
 By Observation~\ref{wpwQ}(\ref{wpwQ4.5}),
  \mbox{$k \le k'$}.

\item $q$ reads $R_{wQ} =(\ready, v_{k-1}, v_{k})$ \iln{qPre} $\ru_q(v_k)$
	before $p$ reads $R_{wp} =(\ready, v_{k'}, v_{k'+1})$ \iln{pPre} $\ru_p(v_{k'})$.
By Observation~\ref{wpwQ}(\ref{wpwQ4}), 
	$k \le k'+1$. 

\begin{compactenum}[i.]
\item $k < k'+1$. Then $k\le k'$.

\item $k = k'+1$.
	We now show that this case is impossible.
	Since $k = k'+1$,
	 $p$ reads $R_{wp} =(\ready, v_{k-1}, v_{k})$ in $\ru_p(v_{k-1})$, and $\ru_p(v_{k-1})$ returns in line~\ref{pPrt}.
Since $q$ reads $R_{wQ} =(\ready, v_{k-1}, v_{k})$ in line~\ref{qPre} of $\ru_q(v_k)$,
	$\ru_q(v_k)$
	 returns 
	in line~\ref{qPrt4},~\ref{qPrt5},~\ref{qPrt2}, or~\ref{qPrt3}.
We now consider each one of these cases.
	\begin{compactenum}[a.]
	\item $\ru_q(v_k)$
	 returns in line~\ref{qPrt4}.
	Then $q$ reads $R_{wQ} =(\done, v_{\ell})$ for some $\ell \ge k$ in line~\ref{qCommitp} of $\ru_q(v_k)$.
	Since $\ru_q(v_k)$
	 precedes $\ru_p(v_{k-1})$,
		$q$ reads $R_{wQ} =(\done, v_{\ell})$
		before $p$ reads $R_{wp} =(\ready, v_{k-1}, v_{k})$ in $\ru_p(v_{k-1})$.
By Observation~\ref{wpwQ}(\ref{wpwQ6}), $\ell<k$ --- a contradiction.

	\item $\ru_q(v_k)$
	 returns in line~\ref{qPrt5}.
	 So $q$ read $R_{wQ}=(\ready, - ,v_{\ell})$ for some $\ell > k$ in line~\ref{qPrep} of $\ru_q(v_k)$.
	 Since $\ru_q(v_k)$
	 precedes $\ru_p(v_{k-1})$,
		$q$ reads $R_{wQ}=(\ready, - ,v_{\ell})$ in $\ru_q(v_k)$
		before $p$ reads $R_{wp} =(\ready, v_{k-1}, v_{k})$ in $\ru_p(v_{k-1})$.
By Observation~\ref{wpwQ}(\ref{wpwQ4}), $\ell \le k$ --- a contradiction.
		 
	\item $\ru_q(v_k)$
	 returns
		in line~\ref{qPrt2} or \ref{qPrt3}.
	Then $q$ reads $R_{pQ}=v_{\ell}$ for some $\ell \ge k$ in line~\ref{qfromp} or~\ref{qfrompp} of $\ru_q(v_k)$.
So $p$ writes $v_{\ell}$ to $R_{pQ}$
	in line~\ref{ptoq} of some $\ru_p(-)$ operation
	before $q$ reads $R_{pQ}$ in $\ru_q(v_k)$.
Thus, $p$ read $R_{wp} =(\done, v_{\ell})$ in line~\ref{pCommit}
	before $q$ reads $R_{pQ}$ in $\ru_q(v_k)$.
Since $\ru_q(v_k)$
	 precedes $\ru_p(v_{k-1})$, 
	 $q$ read $R_{pQ}$ in $\ru_q(v_k)$ before
	$p$ reads $R_{wp} =(\ready, v_{k-1}, v_{k})$ in $\ru_p(v_{k-1})$.
So $p$ read $R_{wp} =(\done, v_{\ell})$ before $p$ reads $R_{wp} =(\ready, v_{k-1}, v_{k})$ in $\ru_p(v_{k-1})$.
By Observation~\ref{wpwQ}(\ref{wpwQ6}), $\ell < k$ --- a contradiction.

	\end{compactenum}

\end{compactenum}

\item $q$ reads $R_{wQ} =(\ready, v_{k}, v_{k+1})$ \iln{qPre} $\ru_q(v_k)$
	before $p$ reads $R_{wp} = (\done, v_{k'})$ \iln{pCommit} $\ru_p(v_{k'})$.
 By Observation~\ref{wpwQ}(\ref{wpwQ4.5}), 
 	$k+1 \le k'$. So \mbox{$k \le k'$}.

\item $q$ reads $R_{wQ} =(\ready, v_{k}, v_{k+1})$ \iln{qPre} $\ru_q(v_k)$
	before $p$ reads $R_{wp} =(\ready, v_{k'}, v_{k'+1})$ \iln{pPre} $\ru_p(v_{k'})$.
By Observation~\ref{wpwQ}(\ref{wpwQ4}), $k+1 \le k'+1$. So \mbox{$k \le k'$}.
\end{compactenum}
\end{proof}

Finally, we prove that there are no ``new-old'' inversions between the reads of processes~in~$Q$.

\begin{lemma} \label{norq}
		If $\ru_q(v_k)$ and $\ru_{q'}(v_{k'})$ are read operations by non-{\ml} processes $q\in Q$ and $q'\in Q$ respectively,
	and $\ru_q(v_k)$ precedes $\ru_{q'}(v_{k'})$, then \mbox{$k \le k'$}.
\end{lemma}

\begin{proof}
Suppose processes $q\in Q$ and $q'\in Q$ are not {\ml}.
Let $\ru_q(v_k)$ and $\ru_{q'}(v_{k'})$ be read operations by $q$ and $q'$ respectively,
	 such that $\ru_q(v_k)$ precedes $\ru_{q'}(v_{k'})$.
By Lemma~\ref{qread}, the following occurs:
\begin{compactenum}
	\item $q$ reads $R_{wQ} = (\done, v_{k})$ \ilns{qCommit} $\ru_q(v_k)$, or
	
	\item $q$ reads $R_{wQ} =(\ready, v_{k-1}, v_{k})$ \ilns{qPre} $\ru_q(v_k)$, or

	\item $q$ reads $R_{wQ} =(\ready, v_{k}, v_{k+1})$ \ilns{qPre} $\ru_q(v_k)$
\end{compactenum}

\emph{before} the following occurs:

\begin{compactenum}
	\item $q'$ reads $R_{wQ} = (\done, v_{k'})$ \ilns{qCommit} $\ru_{q'}(v_{k'})$, or
	
	\item $q'$ reads $R_{wQ} =(\ready, v_{k'-1}, v_{k'})$ \ilns{qPre} $\ru_{q'}(v_{k'})$, or

	\item $q'$ reads $R_{wQ} =(\ready, v_{k'}, v_{k'+1})$ \ilns{qPre} $\ru_{q'}(v_{k'})$.
\end{compactenum}

\smallskip
So there are nine possible cases:

\begin{compactenum}

\item $q$ reads $R_{wQ} = (\done, v_{k})$ \iln{qCommit} $\ru_q(v_k)$
	before
	$q'$ reads $R_{wQ} = (\done, v_{k'})$ \iln{qCommit} $\ru_{q'}(v_{k'})$.
  By Observation~\ref{monowp}(\ref{monowp2.5}),
  	 \mbox{$k \le k'$}.

 \item $q$ reads $R_{wQ} = (\done, v_{k})$ \iln{qCommit} $\ru_q(v_k)$
 	before
	$q'$ reads $R_{wQ} =(\ready, v_{k'-1}, v_{k'})$ \iln{qPre} $\ru_{q'}(v_{k'})$.
  By Observation~\ref{monowp}(\ref{monowp3}),
  	 $k < k'$. So \mbox{$k \le k'$}.

 \item $q$ reads $R_{wQ} = (\done, v_{k})$ \iln{qCommit} $\ru_q(v_k)$
 	before
	$q'$ reads $R_{wQ} =(\ready, v_{k'}, v_{k'+1})$ \iln{qPre} $\ru_{q'}(v_{k'})$.
 By Observation~\ref{monowp}(\ref{monowp3}),
 	 $k < k'+1$. So \mbox{$k \le k'$}.

 \item $q$ reads $R_{wQ} =(\ready, v_{k-1}, v_{k})$ \iln{qPre} $\ru_q(v_k)$
 	before
	$q'$ reads $R_{wQ} = (\done, v_{k'})$ \iln{qCommit} $\ru_{q'}(v_{k'})$.
 By Observation~\ref{monowp}(\ref{monowp1}),
 	 \mbox{$k \le k'$}.

\item $q$ reads $R_{wQ} =(\ready, v_{k-1}, v_{k})$ \iln{qPre} $\ru_q(v_k)$
	before
	$q'$ reads $R_{wQ} =(\ready, v_{k'-1}, v_{k'})$ \iln{qPre} $\ru_{q'}(v_{k'})$.
By Observation~\ref{monowp}(\ref{monowp2}),
	 \mbox{$k \le k'$}.

 \item $q$ reads $R_{wQ} =(\ready, v_{k-1}, v_{k})$ \iln{qPre} $\ru_q(v_k)$
 	before
	$q'$ reads $R_{wQ} =(\ready, v_{k'}, v_{k'+1})$ \iln{qPre} $\ru_{q'}(v_{k'})$.
 By Observation~\ref{monowp}(\ref{monowp2}),
 	 $k \le k'+1$.
	\begin{compactenum}[i.]
	\item $k < k'+1$. Then $k\le k'$.

	\item $k = k'+1$.
	We now show that this case is impossible.
	Since $k = k'+1$, 
		$q'$ reads $R_{wQ} =(\ready, v_{k-1}, v_{k})$ in $\ru_{q'}(v_{k-1})$,
	and $\ru_{q'}(v_{k-1})$ returns in line~\ref{qPrtlw}.
	So $q'$ reads
	$R_{qq'}$ in line~\ref{rtoq} of $\ru_{q'}(v_{k-1 })$ before $\ru_{q'}(v_{k-1 })$ returns in line~\ref{qPrtlw}.
Since $q$ reads $R_{wQ} =(\ready, v_{k-1}, v_{k})$ in $\ru_q(v_k)$,
	$\ru_q(v_k)$ returns
	in line~\ref{qPrt4},~\ref{qPrt5},~\ref{qPrt2}, or~\ref{qPrt3}.
We now consider each one of these cases.

		\begin{compactenum}[a.]

		\item  $\ru_q(v_k)$ returns in line~\ref{qPrt4}.
	Then $q$ reads $R_{wQ} =(\done, v_{\ell})$ for some $\ell \ge k$ in line~\ref{qCommitp} of $\ru_q(v_k)$.
	Since $\ru_q(v_k)$ precedes $\ru_{q'}(v_{k-1})$,
		$q$ read $R_{wQ} =(\done, v_{\ell})$ in $\ru_q(v_k)$
		before $q'$ reads $R_{wQ} =(\ready, v_{k-1}, v_{k})$ in $\ru_{q'}(v_{k-1})$.
By Observation~\ref{monowp}(\ref{monowp3}), $\ell <k$ --- a contradiction.

	\item $\ru_q(v_k)$
	 returns in line~\ref{qPrt5}.
	 Then $q$ reads $R_{wQ}=(\ready, - ,v_{\ell})$  for some $\ell > k$ in line~\ref{qPrep} of $\ru_q(v_k)$.
	 Since $\ru_q(v_k)$
	 precedes $\ru_{q'}(v_{k-1})$,
		$q$ read $R_{wQ}=(\ready, - ,v_{\ell})$ in $\ru_q(v_k)$
		before $q'$ reads $R_{wQ} =(\ready, v_{k-1}, v_{k})$ in $\ru_{q'}(v_{k-1})$.
By Observation~\ref{monowp}(\ref{monowp2}), $\ell \le k$ --- a contradiction.

		\item $\ru_q(v_k)$ returns in line~\ref{qPrt2} or \ref{qPrt3}.
		Then $q$ writes $v_{k}$ in $R_{qq'}$ in line~\ref{qtor} or \ref{qtorp} of $\ru_q(v_k)$.
Since $\ru_q(v_k)$ precedes $\ru_{q'}(v_{k-1})$,
	$q$ writes $v_{k}$ in $R_{qq'}$ in $\ru_q(v_k)$ before $q'$ reads 
	$R_{qq'}$ in line~\ref{rtoq} of $\ru_{q'}(v_{k-1 })$.
	Thus, by Observation~\ref{monorqwr}, 
	$q'$ reads $R_{qq'}=v_{\ell}$ for some $\ell \ge k$ in line~\ref{rtoq} of $\ru_{q'}(v_{k-1 })$.
So $\ru_{q'}(v_{k-1 })$ returns in line~\ref{qPrt3}, rather than in line~\ref{qPrtlw} --- a contradiction

		\end{compactenum}

	\end{compactenum}

\item $q$ reads $R_{wQ} =(\ready, v_{k}, v_{k+1})$ \iln{qPre} $\ru_q(v_k)$
	before
	$q'$ reads $R_{wQ} = (\done, v_{k'})$ \iln{qCommit} $\ru_{q'}(v_{k'})$.
 By Observation~\ref{monowp}(\ref{monowp1}), $k+1 \le k'$. So \mbox{$k \le k'$}.

\item $q$ reads $R_{wQ} =(\ready, v_{k}, v_{k+1})$ \iln{qPre} $\ru_q(v_k)$
	before
	$q'$ reads $R_{wQ} =(\ready, v_{k'-1}, v_{k'})$ \iln{qPre} $\ru_{q'}(v_{k'})$.
By Observation~\ref{monowp}(\ref{monowp2}), $k+1 \le k'$. So \mbox{$k \le k'$}.

\item $q$ reads $R_{wQ} =(\ready, v_{k}, v_{k+1})$ \iln{qPre} $\ru_q(v_k)$
	before
	$q'$ reads $R_{wQ} =(\ready, v_{k'}, v_{k'+1})$ \iln{qPre} $\ru_{q'}(v_{k'})$.
By Observation~\ref{monowp}(\ref{monowp2}), $k+1 \le k'+1$. So \mbox{$k \le k'$}.
\end{compactenum}
\end{proof}

We now prove that the writes and reads of the lower-level procedures $\wu()$ and $\ru_r()$ for all readers $r$ satisfy
	Property~2 of Definition~\ref{LinearizableByz}.

\begin{lemma}\label{p2r}
Let $\ru_r(v_k)$ and $\ru_{r'}(v_{k'})$ be any read operations by some non-{\ml} processes $r$ and $r'$
	in $\{p\} \cup Q$.
If $\ru_r(v_k)$ precedes $\ru_{r'}(v_{k'})$ then \mbox{$k \le k'$}.
\end{lemma}

\begin{proof}
Immediate from Lemmas~\ref{monopread},~\ref{nopq},~\ref{noqp},~\ref{norq}.
\end{proof}

Finally, we prove that the write and read operations
	of the high-level procedures $\Wu()$ and $\Ru()$
	satisfy Property~2 of Definition~\ref{LinearizableByz}.

\begin{lemma}\label{Cp2} \emph{[\textbf{Property 2}: No ``new-old'' inversion]}

Let $\Ru(u_k)$ and $\Ru(u_{k'})$ be any read operations by some non-{\ml} processes in $\{p\}\cup Q$.
If $\Ru(u_k)$ precedes $\Ru(u_{k'})$ then \mbox{$k \le k'$}.
\end{lemma}

\begin{proof}
Let $\Ru(u_k)$ and $\Ru(u_{k'})$ be any read operations by some non-{\ml} processes $r$ and $r'$ in $\{p\}\cup Q$.
Suppose $\Ru(u_k)$ precedes $\Ru(u_{k'})$.
By Observation~\ref{Rr},
	in the $\Ru(u_k)$ and $\Ru(u_{k'})$ operations,
	processes $r$ and $r'$ invoke and complete a $\ru_r(v_k)$ and $\ru_{r'}(v_{k'})$ operation, respectively.
Since $\Ru(u_k)$ precedes $\Ru(u_{k'})$,
	$\ru_r(v_k)$ precedes $\ru_{r'}(v_{k'})$.
By Lemma~\ref{p2r},
	\mbox{$k \le k'$}.
\end{proof}

By Lemmas~\ref{Cp1} and~\ref{Cp2}, the $\Wu(-)$ and $\Ru(-)$ operations of the register implementation $\I{n}$
 	satisfy the linearizability Properties~1 and~2 of Definition~\ref{LinearizableByz}.
Therefore:

\begin{theorem}\label{Theo-linearizability}
For all $n\ge2$,
	the implementation $\I{n}$ is linearizable. 
\end{theorem}

{\large \textbf{Termination of \texorpdfstring{$\I{n}$}{In}}.}
We now prove the 
Termination
property of the implementation $\I{n}$.
As in the previous section, we assume that the implementations of the registers $R_{wQ}$ and $R_{pQ}$ that $\I{n}$ uses
	are {\vd} (so they are linearizable).

\smallskip
Note that if $w$ is {\ml}, it could write $(\done,\tp{k,-})$ followed by $(\done,\tp{k',-})$ in $R_{wp}$ such that $k' < k$.
To prevent $p$ from ``acting on'' a $\done$ tuple that is out of order,
	$p$ remembers in the variable $\lc$ the value $k$ of the last $(\done,\tp{k,-})$ tuple that it accepted.
Using this variable in the guard of line~\ref{pCommit} ensures
	that the $\tp{k,-}$ tuples that $p$ writes in $R_{pQ}$ in line~\ref{ptoq}  have non-decreasing values of $k$ even if the writer $w$ is {\ml}.
So if $p$ is not {\ml}, correct processes that read $R_{pQ}$, read tuples $\tp{k,-}$ with non-decreasing values of $k$.
More precisely:

\begin{observation}\label{monopqrr}
Suppose $p$, $q\in Q$, and $q'\in Q$ are not {\ml}.
	If $q$ reads $R_{pQ}=\tp{k,-}$ before $q'$ reads $R_{pQ}=\tp{k',-}$,
	then \mbox{$k \le k'$}.
\end{observation}

\begin{theorem}\label{Theo-wait-freedom}
 For all $n\ge2$,
 	the implementation $\I{n}$ 
 	satisfies the Termination property
 	if the writer is correct or no reader is {\ml}.
 \end{theorem} 

\begin{proof}
We must show that if the writer is correct or no reader is {\ml},
	then every correct process completes each operation that it invokes
	in a finite number of steps.
If the writer $w$ is correct, it is clear from the code of the write procedures $\Wu()$ and $\wu()$
	that $w$ completes every $\Wu()$ invocation with a response.
If the reader $p$ is correct, it is also clear from the code of the procedures $\Ru()$ and $\ru_p()$
	that $p$ completes every $\Ru()$ invocation with a response.
Let $q$ be a correct process in $Q$.
It remains to show that if the writer is correct or no reader is {\ml} then $q$ completes every $\Ru()$ invocation
	with a response.
Consider any execution of $\Ru()$ by $q$.
Note that  in line~\ref{callru} of $\Ru()$ process $q$ calls $\ru_q()$,
	and if $\ru_q()$ returns a response
	then $\Ru()$ also returns a response.
We now show that $\ru_q()$ returns a response. To do so we first show the following.

\begin{claim}\label{ouahed}
If the writer $w$ is correct, then $q$ does not loop forever in \textsc{Thread 1} (lines~\ref{rp}-\ref{qPrt5}).
\end{claim}

\begin{proof}
Suppose, for contradiction, that the writer $w$ is correct but $q$ loops forever in \mbox{lines~\ref{rp}-\ref{qPrt5}.}
Thus, $q$ reads $R_{wQ}=(\ready,\lw,\langle k,u \rangle)$ for some $\lw$ and some $\langle k,u \rangle$ in line~\ref{qPre}.
Since $w$ is correct: (1) $w$ previously wrote $(\ready,\lw,\langle k,u \rangle)$ into $R_{wQ}$,
	and (2) $w$ eventually writes $(\done, \langle k,u \rangle)$ into $R_{wQ}$.
Furthermore, by Observation~\ref{monowpw}, if $w$ writes into $R_{wQ}$ after writing $(\done, \langle k,u \rangle)$ into $R_{wQ}$,
	then it writes $(\ready,-,\langle k',- \rangle)$ or $(\done, \langle k',- \rangle)$ into $R_{wQ}$ with $k'>k$.
Since $q$ spins forever in the loop of lines~\ref{rp}-\ref{qPrt5}, $q$ reads $R_{wQ}$ infinitely many times.
From the above, it is clear that eventually $q$ reads $R_{wQ}=(\done, \langle k',- \rangle)$ for some $k' \ge k$ in lines~\ref{qCommitp} or $R_{wQ}=(\ready,-,\langle k',- \rangle)$
	for some $k'>k$ in lines~\ref{qPrep}, and then $q$ exits the loop by returning a tuple in lines~\ref{qPrt4} or~\ref{qPrt5} --- a contradiction.
\end{proof}

\begin{claim}\label{tenen}
Suppose no reader is {\ml}.
	Then (1) $q$ does not block in \textsc{Thread~2}, and
		(2) if $q$ evaluates the condition of
		line~\ref{qfrompp} in \textsc{Thread~2}, then $q$ finds that this condition holds.
\end{claim}

\begin{proof}
Suppose no reader is {\ml}.
So in particular no reader in $Q$ is {\ml}.
Thus, since the implementation of $R_{pQ}$ is valid,
		$q$'s read operations of $R_{pQ}$ in lines~\ref{qfromp} and~\ref{qfrompp} do not block.
So $q$ cannot block in \textsc{Thread~2}.

Suppose $q$ evaluates the condition of
		line~\ref{qfrompp} in \textsc{Thread~2}.
Since $q$	reaches line~\ref{qfrompp} in \textsc{Thread~2},
	process $q$ \emph{previously} read:
	(1) $R_{wQ}=(\ready,\lw,\langle k,u \rangle)$
	for some $\lw$ and some $\langle k,u \rangle$ in line~\ref{qPre},
	and
	(2) $R_{q'q}=\langle k',- \rangle$ \text{for some} $q' \in Q$ and some $k'\ge k$ in line~\ref{rtoq}.
Thus, since reader $q'$ is not {\ml},
	$q'$ writes $\tp{k',-}$ into $R_{q'q}$ before $q$ reached line~\ref{qfrompp}.
Note that $q'$ can write $\tp{k',-}$ into $R_{q'q}$ only 
	in line~\ref{qtor} or~\ref{qtorp} in some execution of $\ru_{q'}()$ by $q'$.
Before doing so, $q'$ must have read $R_{pQ}=\tp{k'',-}$ in line~\ref{qfromp} or~\ref{qfrompp} for some $k''\ge k'$
	(in that execution of $\ru_{q'}()$).
So this reading of $\tp{k'',-}$ from register $R_{pQ}$ occurred \emph{before} $q$ reached line~\ref{qfrompp}. 
Thus, by Observation~\ref{monopqrr},
	when $q$ reads $R_{pQ}$ in line~\ref{qfrompp},
	$q$ must read $R_{pQ} =\tp{k''',-}$ for some $k'''\ge k''\ge k'\ge k$.
So, $q$ finds that the condition of line~\ref{qfrompp} holds.
\end{proof}

We now prove that
	if the writer is correct or no reader is {\ml}, then
	the execution of $\ru_q()$ by the correct process $q$ returns.
Suppose, for contradiction, that:
	(1) writer is correct or no reader is {\ml}, but
	(2) the execution of $\ru_q()$ by $q$ does not return.
So $q$ does not return in line~\ref{qCrt} or line~\ref{qrtbot} of $\ru_q()$.
Thus $q$ enters the {cobegin}-{coend} section of the code of $\ru_q()$,
	and it executes \textsc{Thread~1} and \textsc{Thread~2} in parallel.

Consider the forever loop in \textsc{Thread~1}.
From the code of $\ru_q()$, it is clear that
	$q$ stops executing this loop if and only if
	either $q$ returns a value in lines~\ref{qPrt4} or~\ref{qPrt5} of this loop,
	or $q$ exits $\ru_q()$ altogether by returning some value in \textsc{Thread~2}.
Thus, since the execution of $\ru_q()$ by $q$ does not return,
	$q$ loops forever in \textsc{Thread~1}.

By Claim~\ref{ouahed}, this implies that the writer $w$ is not correct.
So, by the assumption~(1) on process failures,
	no reader is {\ml}.
By Claim~\ref{tenen}, $q$ does not block inside \textsc{Thread~2}.
By the code of \textsc{Thread~2}, either $q$ returns a value
	in line~\ref{qPrt2} or~\ref{qPrtlw} of \textsc{Thread~2},
	or $q$ reaches line~\ref{qfrompp} and evaluates the condition in this line.
In the latter case, by Claim~\ref{tenen}, the condition in line~\ref{qfrompp} evaluates to \textsc{true},
	and so $q$ returns a value in line~\ref{qPrt3} of  \textsc{Thread~2}.
So in all cases, $q$ exits $\ru_q()$ by returning a value in \textsc{Thread~2} --- a contradiction to assumption~(2).
\end{proof}

By Theorems~\ref{Theo-linearizability} and~\ref{Theo-wait-freedom}, we have that if
	the implementations of $R_{wQ}$ and $R_{pQ}$ are valid (Assumption~\ref{a0}), then the implementation $\I{n}$ is also valid. So we have:

\begin{restatable}{theorem}{Iisvalid}
\label{I-is-valid}
For all $n \ge2$,
	$\I{n}$ is a valid implementation of a $\Reg{1}{n}$ from
	implemented $\Reg{1}{n-1}$s and atomic $\Reg{1}{1}$s,
	provided that the implementations of the $\Reg{1}{n-1}$s that it uses
	(namely, $R_{wQ}$ and $R_{pQ}$) are also valid.
\end{restatable}

%% file: section-algo-n-readers-signed.tex

\input{code-algo-n-readers-signed.tex}

\section{Register implementation for systems with digital signatures}\label{withsignatures}

We now consider systems where processes are subject to Byzantine failures, but they can use unforgeable signatures.
Algorithm~\ref{1wnrsig} gives a wait-free linearizable
	implementation {\Is}
	of a $\Reg{1}{n}$ that is writable by process $w$
	and readable by a set $P$ of $n$ processes.
This implementation tolerates any combination and number of faulty processes, and it~works as~follows.

To write $u$, the writer $w$ calls $\Wu(u)$.
	In this procedure, $w$ adds a sequence number $k$ to form a tuple~$\tp{k,u}$,
	then it signs $\tp{k,u}$ with $w$ (the signed tuple is denoted $\tp{k,u}_w$),
	and finally it executes the lower-level write procedure $\wu(\tp{k,u}_w)$.
It is worth noting that in this algorithm, the writer $w$ is the \emph{only} process that signs values.

To read a value, a reader $p \in P$ calls $\Ru()$.
	This procedure calls a lower-level read procedure $\ru()$ that reads \emph{signed} tuples written by the $\wu()$ procedure.
If $\ru()$ returns a tuple of the form $\tp{k,u}_w$ for some $k$ and $u$,
	$\Ru()$ strips the signature $w$ and sequence number~$k$
	from the tuple, and then it returns the value $u$ as the value read
	(otherwise $\Ru()$ returns~$\bot$ to indicate a read failure).

The lower-level procedures $\wu()$ and $\ru()$ work as follows (in these procedures, $R_{ij}$ denotes an atomic $\Reg{1}{1}$ that is writable by process $i$ and readable by process $j$):

$\bullet$\hspace{1mm} To execute $\wu(\tp{k,u}_w)$,
	the writer $w$ simply writes $\tp{k,u}_w$ in
	$R_{wi}$ for every process $i \in P$.

$\bullet$\hspace{1mm} To execute $\ru()$,
	a reader $p \in P$ first reads the $\Reg{1}{1}$
	$R_{ip}$ fo every process $i\in \{w\}\cup P$ to form
	the set $\Vals$ of all the tuples validly signed by $w$ that it reads.
Then $p$ selects the tuple $\tp{k,u}_w$ with the maximum sequence number $k$ in $\Vals$,
	and returns this tuple; but before doing so
	$p$ writes $\tp{k,u}_w$ into the $\Reg{1}{1}$ $R_{pi}$ for every reader $i \in P$
	to notify them that it read $\tp{k,u}_w$.

\input{proof-algo-n-readers-signed}

%% file: code-algo-n-readers-signed.tex

\begin{algorithm}[!t]
\caption{Implementation {\Is} of a $\Reg{1}{n}$
	writable by process~$w$ and readable by a set $P$ of $n$ processes in a system with unforgeable signatures.
	{\Is} uses atomic $\Reg{1}{1}$s.}
\label{1wnrsig}

\vspace{-3mm}
{\footnotesize
\begin{multicols}{2}

\textsc{Atomic Registers}
\vspace{.7mm}

~~~~For all processes $i$ and $j $ in $\{w\} \cup P$:

\hspace{1cm}
$R_{ij}$: atomic $\Reg{1}{1}$;
	initially $\langle 0,u_0 \rangle_w$.
\vspace{2mm}

\columnbreak

\textsc{Local variables}
\vspace{.7mm}

~~~~$c$: variable of $w$; initially $0$

~~~~$\Vals$: variable of each $p$ in $P$; initially $\emptyset$.
\end{multicols}
\vspace{-3mm}
\hrule
\vspace{-2mm}
\begin{algorithmic}[1]
\begin{multicols}{2}
\Statex
\textsc{\Wu($u$):}  ~~~~~~~~$\triangleright$ executed by the writer $w$
\Indent
\State $c \gets c+1$  
\State \label{sign} \textbf{call} \textsc{\wu($\langle c,u \rangle_w$)} 
\State \Return $\Done$
\EndIndent

\columnbreak

\noindent
\textsc{$\Ru$():}\Comment{executed by any reader $p$ in $P$}
\Indent
\State \textbf{call} $\ru$()
\State \textbf{if} this call returns some tuple $\langle k,u \rangle_w$ \textbf{then}
\Indent
\State $\Return$ $u$ 
\EndIndent
\State \textbf{else} \Return $\bot$
\EndIndent
\end{multicols}
\vspace{-1mm}
\hrule
\vspace{4mm}
\Statex

\textsc{\wu($\langle k,u \rangle_w$):} \Comment{executed by $w$ to do its $k$-th write}
\Indent
\State \textbf{for every} process $i \in P$ \textbf{do}
\Indent
\State \label{wi} $R_{wi} \gets \langle k,u \rangle_w$ \Comment{$\langle k,u \rangle$ signed by $w$}
\EndIndent
\State \Return $\Done$
\EndIndent

\vspace*{4mm}

\noindent
\textsc{$\ru$():}\Comment{executed by any reader $p$ in $P$}
\Indent
\State $\Vals \gets \emptyset$
\State \textbf{for every} process $i \in \{w\} \cup P$ \textbf{do}
\Indent
\State \label{readall} \textbf{if} $R_{ip}=\langle \ell,\val \rangle_w$ for some $\langle \ell,\val \rangle$ validly signed by $w$ \textbf{then}
\Indent
\State \label{val}$\Vals \gets \Vals \cup \{\langle \ell,\val \rangle_w\}$
\EndIndent
\EndIndent
\State \label{max} $\langle k,u\rangle_w \gets \mbox{tuple $\langle \ell,\val \rangle_w$ with maximum sequence number $\ell$ in $\Vals$}$
\State \textbf{for every} process $i \in P$ \textbf{do}
\Indent
\State \label{pi}$R_{pi} \gets \langle k,u \rangle_w$
\EndIndent
\State \Return $\langle k,u \rangle_w$
\EndIndent

\end{algorithmic}
}
\end{algorithm} 

\vspace{-3mm}

%% file: proof-algo-n-readers-signed.tex

We now prove that the implementation {\Is} given by Algorithm~\ref{1wnrsig} is wait-free and linearizable.

\textbf{Wait-freedom.} This is trivial: the code of Algorithm~\ref{1wnrsig} does not contain any loop or wait statement, so 
	every call to the \Wu() and \Ru() procedures by any correct process terminates with a return value in a bounded number of its own steps.
	Thus:

\begin{observation}\label{sigwaitfree}
For all $n\ge 2$, the implementation {\Is} is bounded wait-free.
\end{observation}

\textbf{Linearizability.}  To prove that the implementation {\Is} given by Algorithm~\ref{1wnrsig} is linearizable,
	we must show that \emph{if the writer $w$ of the register implemented by {\Is} is not {\ml}}
	then Properties 1 and 2 of Definition~\ref{LinearizableByz} hold.
So for the rest of this section we assume that the writer $w$ is not {\ml}, and
	henceforth we omit to repeat this assumption in our observations, lemmas, and theorem.

As in the previous section:

\begin{compactitem}
\item $u_0$ is the initial value of the register that Algorithm~\ref{1wnrsig} implements.
\item For $k\ge1$, $u_k$ denotes the $k$-th value written by $w$ using
	the procedure
	$\Wu()$.
	More precisely, if $w$ calls $\Wu()$ with a value $u$ and this is its $k$-th call of $\Wu()$,
	then~$u_k$~is~$u$.
\item $v_0$ is $\tp{0,u_0}_w$.

\item For $k \ge1$, $v_k$ denotes the $k$-th value written by $w$ using
	the procedure
	$\wu()$. 
\end{compactitem}

We first prove the linearizability Properties~1 and~2 of Definition~\ref{LinearizableByz}
are satisfied by the writes and reads of the lower-level procedures $\wu()$ and $\ru()$
	(Lemmas~\ref{sigProp0} and~\ref{sigProp1}),
 	and then prove they are also satisfied by the writes and reads of the high-level procedures $\Wu()$ and $\Ru()$ (Lemmas~\ref{sigProp2r} and~\ref{sigProp2}).

\begin{observation}\label{vks}
For all $k\ge 0$, $v_k = \tp{k,u_k}_w$.
\end{observation}

\begin{observation}\label{snwrites}
Let $\wu(v)$ be any write operation by $w$.
Then there is a $k \ge 1$ such that $v=v_k$.
\end{observation}

Note that a correct reader enters a value $v$ into its set $\Vals$
	only if $v$ is a tuple $\langle \ell,\val \rangle$ validly signed by $w$, 
	i.e., $\langle \ell,\val \rangle_w$.
Since $w$ is not {\ml}, and signatures are unforgeable, it must be that $\langle \ell,\val \rangle_w$ is $v_{\ell}$.
So we have:

\begin{observation}\label{tuplesonlyvk}
For every non-{\ml} reader $p \in P$, if $v \in \Vals$
	then $v= v_k$ for some $k\ge 0$.
\end{observation}

\begin{observation}\label{ru}
For every non-{\ml} reader $p \in P$,
	if $\ru(v)$ is an operation by $p$,
	then $v=v_k$ for some $k\ge 0$.
Furthermore, if $k>0$ then $w$ invokes $\wu(v_k)$ before $\ru(v_k)$ returns.
So the operation $\wu(v_k)$ precedes $\ru(v_k)$ or is concurrent with $\ru(v_k)$.
\end{observation}

From Observation~\ref{tuplesonlyvk} and lines~\ref{max}-\ref{pi} of the procedure $\ru()$,
	if a non-{\ml} reader $p \in P$ writes a value $v$ in $R_{pi}$,
	then $v=v_k$ for some $k\ge 0$.
Furthermore, By Observation~\ref{snwrites}, if the non-{\ml} writer $w$
	writes a value $v$ in $R_{wi}$,
	then $v=v_k$ for some $k\ge 1$.
So:

\begin{observation}\label{onlyvk}
Suppose a process $i\in \{w\}\cup P$ is not {\ml}.
For every process $j\in \{w\}\cup P$,
	if $i$ writes $v$ in $R_{ij}$,
	then $v=v_k$ for some $k\ge 0$.
\end{observation}

\begin{observation}\label{onlyvkr}
Suppose processes $i \in \{w\}\cup P$ and $j \in \{w\}\cup P$ are not {\ml}.
If $j$ reads $R_{ij}=v$,
	then $v=v_k$ for some $k\ge 0$.
\end{observation}

\begin{observation}\label{monow}
      Suppose a process $i\in \{w\}\cup P$ is not {\ml}.
For every process $j\in \{w\}\cup P$,
	if $i$ writes $v_k$ in $R_{ij}$ before $i$ writes $v_{k'}$ in $R_{ij}$,
	then \mbox{$k \le k'$}.
\end{observation}

\begin{observation}\label{monowwr}
Suppose processes $i \in \{w\}\cup P$ and $j\in \{w\}\cup P$ are not {\ml}.
If $i$ writes $v_k$ in $R_{ij}$ before $j$ reads $R_{ij}=v_{k'}$,
	then \mbox{$k \le k'$}.
\end{observation}

\begin{lemma}\label{wuru}
Suppose a reader $p\in P$ is not {\ml}.
If a $\wu(v_k)$ operation precedes a $\ru(v_{k'})$ operation by $p$,
	then $k'\ge k$.
\end{lemma}
\begin{proof}
Suppose a $\wu(v_k)$ operation precedes a $\ru(v_{k'})$ operation by some non-{\ml}~reader~$p$.
So $w$ writes $v_k$ into $R_{wp}$ in line~\ref{wi} of $\wu(v_k)$ before $p$ reads $R_{wp}$ in line~\ref{readall} of $\ru(v_{k'})$.
Note that $k\ge 1$.
By Observations~\ref{onlyvkr} and~\ref{monowwr},
	$p$ reads $R_{wp} = v_{k''}$ for some $k'' \ge k$ in line~\ref{readall} of $\ru(v_{k'})$.
	 Then $p$ adds $v_{k''}$ to $\Vals$ in line~\ref{val} of $\ru(v_{k'})$.
By line~\ref{max} of $\ru(v_{k'})$,
	$v_{k'}$ is the tuple with the maximum sequence number in $\Vals$ and so $k' \ge k''$.
Since $k'' \ge k$,
	$k' \ge k$.
\end{proof}

\begin{lemma}\label{sigProp0}
	If $\ru(v)$ is an operation by a non-{\ml} reader $p\in P$
	then
\begin{compactitem}	
	\item there is a $\wu(v)$ operation that immediately precedes $\ru(v)$ or is concurrent with $\ru(v)$, or	
	\item  $v=v_0$ and no $\wu(-)$ operation precedes $\ru(v)$.
\end{compactitem}
\end{lemma}

\begin{proof}
Suppose $p\in P$ is not {\ml}.
Let $\ru(v)$ be any read operation by $p$.
By Observation~\ref{ru}, $v = v_k$ for some $k \ge 0$.
There are two cases: 

Case $k=0$.
Suppose, for contradiction, 
			that there is a $\wu(v)$ operation that precedes $\ru(v_0)$.
			By Observation~\ref{snwrites},
			$v = v_i$ for some $i \ge 1$.
Since $\wu(v_i)$ precedes~$\ru(v_0)$,
	by Lemma~\ref{wuru},
	$i \le 0$ --- a contradiction.
So no $\wu(-)$ operation precedes $\ru(v_0)$.

Case $k>0$.
By Observation~\ref{ru},
the operation $\wu(v_k)$ precedes $\ru(v_k)$ or is concurrent with~$\ru(v_k)$.
We now show that if $\wu(v_k)$ precedes~$\ru(v_k)$,
	then $\wu(v_k)$ immediately precedes~$\ru(v_k)$.
Suppose, for contradiction,
	that $\wu(v_k)$ precedes~$\ru(v_k)$ but does not immediately precede~$\ru(v_k)$.
Then there is a $\wu(v_i)$ operation that immediately precedes~$\ru(v_k)$.
Clearly, the $\wu(v_k)$ operation precedes the $\wu(v_i)$ operation, and so $i > k$.
Since $\wu(v_i)$ precedes~$\ru(v_k)$,
	by Lemma~\ref{wuru},
		$i \le k$ --- a contradiction.
		Therefore the $\wu(v_k)$ operation immediately precedes $\ru(v_k)$ or is concurrent with $\ru(v_k)$.
\end{proof}

By Observations~\ref{vks} and~\ref{ru}, and the code of procedure $\Ru()$:
\begin{observation}\label{Ru}
If $\Ru(u)$ is an operation by a non-{\ml} process $p \in P$,
then $u=u_k$ for some $k \ge 0$.
\end{observation}

\begin{observation}\label{Ruru}
If $\Ru(u_k)$ is an operation by a non-{\ml} process $p \in P$,
then $p$ invokes and completes a $\ru(v_k)$ operation in $\Ru(u_k)$.
\end{observation}

\begin{observation}\label{Wuwu}
If $\Wu(u_k)$ is a completed operation by $w$,
	then $w$ invokes and completes a $\wu(v_k)$ operation in $\Wu(u_k)$.
\end{observation}

\begin{lemma}\label{sigProp1} \emph{[\textbf{Property 1}: Reading a ``current'' value]}

~
If $\Ru(u)$ is an operation by a non-{\ml} process $p \in P$
	then:
\begin{compactitem}	
	\item there is a $\Wu(u)$ operation that immediately precedes $\Ru(u)$ or is concurrent with $\Ru(u)$, or	

	\item $u=u_0$ and no $\Wu(-)$ operation precedes $\Ru(u)$.

\end{compactitem}
\end{lemma}

\begin{proof}
Let $\Ru(u)$ be any read operation by a non-{\ml} process $p\in P$.
By Observation~\ref{Ru}, $u = u_k$ for some $k \ge 0$.
There are two cases: 

Case $k=0$.
	Suppose, for contradiction, that 
	a $\Wu(u_i)$ operation precedes $\Ru(u_0)$.
	Note that $i \ge 1$.
	By Observations~\ref{Ruru} and~\ref{Wuwu},
	a $\wu(v_i)$ operation precedes a $\ru(v_0)$ operation.
	Since process $p$ is not {\ml},
			by Lemma~\ref{sigProp0},
		 there is no $\wu(-)$ operation that precedes $\ru(v_0)$ --- a contradiction.

Case $k>0$.
	By Observation~\ref{Ruru},
	$p$ invokes and completes a $\ru(v_k)$ operation in $\Ru(u_k)$.
	Since $k>0$,
	$v_k\ne v_0$.
	So,
	by Lemma~\ref{sigProp0},
	there is a $\wu(v_k)$ operation that immediately precedes $\ru(v_k)$ or is concurrent with $\ru(v_k)$.
	Let $\Wu(u_k)$ be the operation in which $w$ invokes the $\wu(v_k)$ operation.
	Since $\wu(v_k)$ operation immediately precedes $\ru(v_k)$ or is concurrent with $\ru(v_k)$,
	the $\Wu(u_k)$ operation immediately precedes $\Ru(u_k)$ or is concurrent with $\Ru(u_k)$.
\end{proof}

\begin{lemma}\label{sigProp2r}
Let $\ru(v_k)$ and $\ru(v_{k'})$ be any read operations by non-{\ml} processes $p$ and $p'$ in $P$, respectively.
If $\ru(v_k)$ precedes $\ru(v_{k'})$, then \mbox{$k \le k'$}.
\end{lemma}

\begin{proof}
Let $\ru(v_k)$ and $\ru(v_{k'})$ be any read operations by non-{\ml} processes $p$ and $p'$ in $P$,
	respectively.
Suppose that $\ru(v_k)$ precedes $\ru(v_{k'})$.
Then $p$ writes $v_k$ in $R_{pp'}$ in line~\ref{pi} of $\ru(v_k)$ before $p'$ reads $R_{pp'}$ in line~\ref{readall} of $\ru(v_{k'})$.
By Observations~\ref{onlyvkr} and~\ref{monowwr},
	$p'$ reads $R_{pp'}=v_{k''}$ for some $k''\ge k$ in line~\ref{readall} of $\ru(v_{k'})$.
	Then $p'$
	adds $v_{k''}$ to $\Vals$ in line~\ref{val} of $\ru(v_{k'})$.
By line~\ref{max} of $\ru(v_{k'})$,
	$v_{k'}$ is the tuple with the maximum sequence number in $\Vals$ and so
		$k'\ge k''$.
Since $k'' \ge k$, $k' \ge k$.  
\end{proof}

\begin{lemma}\label{sigProp2} \emph{[\textbf{Property 2}: No ``new-old'' inversion]}

Let $\Ru(u_k)$ and $\Ru(u_{k'})$ be any read operations by non-{\ml} processes in $P$.
If~$\Ru(u_k)$ precedes $\Ru(u_{k'})$ then \mbox{$k \le k'$}.
\end{lemma}

\begin{proof}
Let $\Ru(u_k)$ and $\Ru(u_{k'})$ be any read operations by non-{\ml} processes $p$ and $p'$ in $P$.
Suppose $\Ru(u_k)$ precedes $\Ru(u_{k'})$.
By Observation~\ref{Ruru},
 in the $\Ru(u_k)$ and $\Ru(u_{k'})$ operations,
	processes $p$ and $p'$ invoke and complete a $\ru(v_k)$ and $\ru(v_{k'})$ operation, respectively.
Since $\Ru(u_k)$ precedes $\Ru(u_{k'})$,
	$\ru(v_k)$ precedes $\ru(v_{k'})$.
By Lemma~\ref{sigProp2r},
	\mbox{$k \le k'$}.
\end{proof}

By Lemmas~\ref{sigProp1} and~\ref{sigProp2},
 the $\Wu(-)$ and $\Ru(-)$ operations of the register implementation {\Is}
 	satisfy the linearizability Properties~1 and~2 of Definition~\ref{LinearizableByz}.
This proves:

\begin{theorem}\label{siglinear}
For all $n\ge 2$, 
	the implementation {\Is} is linearizable.
\end{theorem}

By Observation~\ref{sigwaitfree} and Theorem~\ref{siglinear}, 
we have the following:

\begin{restatable}{theorem}{sigmaintheorem}
\label{sigmaintheorem}
Consider a system where processes are subject to Byzantine failures and can use unforgeable signatures.
For every $n\ge 2$, 
	{\Is} is a bounded wait-free linearizable implementation of a $\Reg{1}{n}$
	from atomic $\Reg{1}{1}$s.
\end{restatable}

%% file: section-impossibility-bounded-termination.tex

\section{Register implementations with bounded termination}

The linearizable register implementations given in Section~\ref{withsignatures}
	(Algorithm~\ref{1wnrsig}, Theorem~\ref{sigmaintheorem})
	and in Appendix~\ref{whatever} (Algorithm~\ref{1w2r}, Theorem~\ref{S2from1})
	guarantee that
	every correct process completes every operation in a \emph{bounded} number of steps (regardless of which processes fail or how~they~fail).

In contrast, the linearizable register implementation
	given in Section~\ref{n-from-1}
	(Algorithm~\ref{1wnr}, Theorem~\ref{Theo-Main-Possibility})
	satisfies the Termination property, namely every correct process completes
	every operation in a \emph{finite} number of steps,
	and it does so under the assumption that
	the writer of the register is correct or no reader is {\ml}.
This raises the question of whether, under the same failure assumption,
	there is a register implementation that satisfies the following stronger termination property:

\begin{definition}[Bounded Termination]\label{Btermination}
Every correct process completes every operation
	in a bounded number of its steps.
\end{definition}

It turns out that the answer is ``No'', even if we assume that
	the writer is not faulty and at most one of the readers can fail.
	More precisely:

\begin{restatable}{theorem}{BoundedImpoResult}
\label{Bounded-Theo-Impossibility-Result}
	For all $n \ge 3$,
	in a system with $n+1$ processes that are subject to Byzantine failures,
	there is \emph{no}
	linearizable implementation of a 
	$\Reg{1}{n}$ from atomic \mbox{$\Reg{1}{n-1}$}s
	that satisfies the Bounded Termination property,
	 even if we assume that the writer of the implemented
	$\Reg{1}{n}$ is correct and at most one reader can be {\ml}. 
\end{restatable}

This is in sharp contrast to Theorem~\ref{Theo-Main-Possibility} which implies that,
	if we assume that the writer~is~correct, 
	there \emph{is} a linearizable implementation of a 
	$\Reg{1}{n}$ from atomic \mbox{$\Reg{1}{1}$}s that satisfies the Termination property
	and tolerates any number of {\ml} readers.
Thus, Theorems~\ref{Theo-Main-Possibility} and~\ref{Bounded-Theo-Impossibility-Result}
	imply that there is an inherent difference
	between achieving Termination and achieving Bounded Termination
	in systems with Byzantine failures.

To prove Theorem~\ref{Bounded-Theo-Impossibility-Result},
	it is easy to modify the impossibility
	proof of Theorem~\ref{Theo-Impossibility-Result} given in Section~\ref{Impossibility-Result}, as we now explain.
First note that in the successive runs that we construct in the proof of Theorem~\ref{Theo-Impossibility-Result},
	we leverage the fact that a correct reader
	that starts a read operation
	cannot wait for the writer to complete a concurrent write,
	\emph{even if the reader is aware that this write is in progress}:
	in our runs, the writer actually crashes so such waiting is not possible.
So a correct reader must complete its read operation even if the writer stops taking~steps.

However, if \emph{the writer is assumed to be correct} and we require ``only'' Termination,
	a reader that is aware
	that a write operation is in progress \emph{can} wait for the writer to complete this operation
	to determine what value to read (so the impossibility proof breaks down in this case, as it should).
But if we require the register implementation to satisfy Bounded Termination, then we are back to a situation
	where a reader must complete its operation without waiting for the writer to take steps.
So the proof of Theorem~\ref{Bounded-Theo-Impossibility-Result}
	can use this fact exactly as the proof of Theorem~\ref{Theo-Impossibility-Result} does.

Thus,
	we can modify the proof of Theorem~\ref{Theo-Impossibility-Result} to obtain
	one for Theorem~\ref{Bounded-Theo-Impossibility-Result} as follows.
Roughly speaking,
	whenever the writer $w$ crashes
	in the proof of Theorem~\ref{Theo-Impossibility-Result},
	the writer is correct but just \emph{pauses}
	in the proof of Theorem~\ref{Bounded-Theo-Impossibility-Result};
	the writer later resumes taking steps and completes its write operation,
	but it does so only \emph{after} the readers complete their read operations.
Since the delayed steps of the writer are not seen by the readers,
	they~behave as in the proof of Theorem~\ref{Theo-Impossibility-Result}.

More precisely,
	in the successive runs that we~construct~in~the~proof:
	\begin{compactitem}
	\item If, in the proof of Theorem~\ref{Theo-Impossibility-Result},
	the writer $w$ crashes after taking a step $s^k$ and before taking further steps,
	\item then, in the proof of Theorem~\ref{Bounded-Theo-Impossibility-Result},
	the correct writer $w$ \emph{temporarily} stops taking steps after taking a step $s^k$ and before taking further steps.
	\end{compactitem}
	
The readers behave the same in the corresponding runs of both proofs.
After the reads by correct readers are completed,
	the writer resumes taking steps and completes its operation.

Since the proof of Theorem~\ref{Bounded-Theo-Impossibility-Result} is mostly a verbatim repetition of the proof of Theorem~\ref{Theo-Impossibility-Result}, we relegate it to Appendix~\ref{bounded-termination-impo}.

%% file: corollaries.tex
\newcommand{\x}{x}
\newcommand{\y}{y}

\section{Implementations from regular registers}

In a seminal work~\cite{Lamport86}, Lamport considered the problem of
	implementing ``atomic'' registers from \emph{regular} registers,
	in systems where processes may \emph{crash}.
Recall that, as with an atomic register, a regular register ensures that a reader reads the ``current'' value of the register,
	but in contrast to an atomic register,
	a regular register allows ``new-old'' inversions in the values read.
In other words, a regular register must satisfy only Property 1 
	of the register linearizability Definition~\ref{LinearizableCrash}.

We now consider this problem for systems where processes are subject to Byzantine failures.
To~do so, we must first define what it means for
	a register implementation to be an implementation of a \emph{regular} register in systems where processes can be {\ml}.
Intuitively, we require that if the writer is not {\ml} then
	non-{\ml} readers must read the ``current'' value of the register (this is Property 1 of Definition~\ref{LinearizableByz}):
	
\begin{definition}[Register Regularity]\label{Regularity}
In a system with Byzantine process failures,
	an implementation of a $\Reg{1}{n}$
	is \emph{a regular register implementation} if and only if, when the writer is not {\ml},
	the following property holds:
	
\emph{[Reading a ``current'' value]}
If a read operation $\ru$ by a process that is not {\ml} returns the value $v$ then:
	\begin{compactitem}
	\item there is a write $v$ operation that immediately precedes $\ru$ or is concurrent with $\ru$, or	
	\item  $v = v_0$ and no write operation precedes $\ru$.\footnote{Recall that $v_0$ is the initial value of the implemented register.}
	\end{compactitem}
\end{definition}

{\large \textbf{An impossibility result.}} Recall that, by Theorem~\ref{Impossibility-Result}, in a system with Byzantine failures there is
	no linearizable implementation of a 
	$\Reg{1}{n}$ from \emph{atomic} \mbox{$\Reg{1}{n-1}$}s
	that satisfies Termination.
This raises the following question: What happens if all processes, \mbox{\emph{including the writer},} are~given regular \mbox{$\Reg{1}{n}$}s instead of atomic \mbox{$\Reg{1}{n-1}$}s? Note that these regular registers \emph{can be read by all the $n$ readers}, as in the desired register implementation, but they are~``only''~regular.
This question is answered by the following:

\begin{theorem}\label{Corollary-Impossibility-Result}
For all $n \ge 3$,
	in a system with $n+1$ processes that are subject to Byzantine failures,
	there is \emph{no}
	linearizable implementation of a 
	$\Reg{1}{n}$ from \emph{regular} \mbox{$\Reg{1}{n}$}s
	that satisfies the Termination property,
	 even if we assume that the writer of the implemented
	$\Reg{1}{n}$ can only crash and at most one reader can be {\ml}.
\end{theorem} 	

The above impossibility result is in sharp contrast to a corresponding possibility result in the case of systems with only crash failures: in such systems it is easy to implement a wait-free linearizable
	$\Reg{1}{n}$ from regular \mbox{$\Reg{1}{n}$}s.
	
\begin{proof}
Let $n \ge 3$.
Consider a system with $n+1$ processes with Byzantine~failures.

\begin{claim}\label{Gianni}
For every process $\x$, there is a wait-free implementation 
	of a \emph{regular}~$\Reg{1}{n}$~{\regx}, 
	writable by $\x$ and readable by the other $n$ processes,
	from atomic $\Reg{1}{1}$s.
\end{claim}

\begin{proof}
The implementation of {\regx}
	is very simple.
For each reader $p$, the writer $\x$ has an atomic $\Reg{1}{1}$ $R_{{\x}p}$ that $\x$ can write and $p$ can read.
\begin{compactitem}
\item To \emph{write a value $v$ into {\regx}}, the writer $\x$ successively writes $v$ into $R_{{\x}p}$ for every reader~$p$.
\item To \emph{read a value from {\regx}}, a reader $p$ reads register $R_{{\x}p}$ and returns the value read.
\end{compactitem}
It is clear that if the writer $\x$ is not {\ml}, then
	any non-{\ml} reader $p$ that reads~{\regx}, reads
	the value written into {\regx} by a write operation by $x$ that
	immediately precedes the read of {\regx} or is concurrent with this read;
	more precisely, this implementation of {\regx} satisfies the property of regular registers, namely, Property~1 of Definition~\ref{Regularity}.
\end{proof}

Let $w$ be any process.
Assume $w$ can only crash and at most one of the remaining $n$ processes can be {\ml}.
For contradiction, suppose that
	using regular \mbox{$\Reg{1}{n}$}s
	there is an implementation $\Iu{n}$
	of a $\Reg{1}{n}$, writable by $w$ and readable by the other $n$ processes,
	such that:
	(1) $\Iu{n}$ is linearizable, and
	(2) $\Iu{n}$ satisfies the Termination property.
	
By Claim~\ref{Gianni}, every regular $\Reg{1}{n}$
	used by $\Iu{n}$ has a \emph{wait-free} implementation
	from atomic $\Reg{1}{1}$s.
Thus, by replacing every regular $\Reg{1}{n}$ used by $\Iu{n}$ with its corresponding wait-free implementation, 
	we obtain an implementation $\Iup{n}$ of a $\Reg{1}{n}$, writable by $w$ and readable by the other $n$ processes, 	\emph{from atomic $\Reg{1}{1}$s}.
It is clear that
 	like~$\Iu{n}$:
	(1) $\Iup{n}$ is linearizable, and
	(2) $\Iup{n}$ satisfies the Termination property.
Therefore $\Iup{n}$ contradicts Theorem~\ref{Theo-Impossibility-Result}.
\end{proof}

%% file: conclusion.tex

\vspace{-5mm}

\section{Concluding remarks}
\vspace{-2mm}

The implementation of registers from weaker registers is a basic problem in distributed computing
	that has been extensively studied
	in the context of processes with crash failures.
In this paper, we investigated this problem in the context of Byzantine processes failures, with and without process signatures.

We first proved that, without signatures, there is no wait-free
	linearizable
	implementation of a $\Reg{1}{n}$ from atomic $\Reg{1}{n-1}$s.
In fact, we showed a stronger result, namely,
	even under the assumption that
	the writer can only crash and at most one reader can be {\ml},
	there is no linearizable
	implementation of a $\Reg{1}{n}$ from atomic $\Reg{1}{n-1}$s
	that ensures that every
	correct process eventually completes its operations.

In light of this strong impossibility result,
	we gave an implementation of a $\Reg{1}{n}$ from atomic $\Reg{1}{1}$s
	that is ``safe'' (i.e, it is linearizable)
	under any combination of Byzantine process failures,
	but it is ``live'' (i.e., it ensures that every
	correct process eventually completes its operations)
	only under the assumption that
	the writer is correct or no reader is {\ml}; this
	matches the impossibility result.

If we assume that the writer is correct,
	with the above implementation (which tolerates any number of {\ml} readers)
	every reader completes each read in a \emph{finite} number of steps.
We showed that is impossible to ensure they do so in a \emph{bounded} number of steps, even if we make the additional assumption that at most one reader can be {\ml}.
	
In sharp contrast with the above results,
	for the case that processes can use signatures,
	we gave a \emph{bounded wait-free} linearizable implementation
	of a $\Reg{1}{n}$ from atomic $\Reg{1}{1}$s which does not rely on any failure assumptions.

Perhaps surprisingly, none of the above results
	refers to a ratio of faulty vs. correct processes,
	such as $n/3$ or $n/2$, that we typically encounter in results that involve Byzantine processes.
For example, Most\'efaoui \emph{et al.}~\cite{Mostefaoui2016}
	prove that one can implement a linearizable $f$-resilient $\Reg{1}{n}$ in \emph{message-passing}
	systems with Byzantine process failures
	if and only if $f < n/3$.
As an other example, Cohen and Keidar~\cite{CohenKeidar2021}
	show that if $f <n/2$, one can use atomic $\Reg{1}{n}$s to get a linearizable
	$f$-resilient implementations of {reliable broadcast},
	{atomic snapshot}, and {asset transfer} objects in systems with Byzantine process failures.

It is worth noting that,
	since atomic $\Reg{1}{1}$s can simulate message-passing channels,
	one can use the $f$-resilient implementation of a $\Reg{1}{n}$ for \emph{message-passing} systems given in~\cite{Mostefaoui2016},
	to obtain an \emph{$f$-resilient} implementation of a $\Reg{1}{n}$ using atomic $\Reg{1}{1}$s.
But $f$-resilient implementations
	 (such as the ones given in~\cite{CohenKeidar2021,Mostefaoui2016})
	require every correct process
	to help the execution of every operation,
	even the operations of \emph{other} processes.
In~contrast, with 
	object implementations in shared-memory systems,
	a common assumption is that
	processes that do not have ongoing operations 
	take no steps;
	so a process that executes an operation
	cannot count on getting help from any process that is
	\emph{not} currently executing its own operation.

%% file: section-algo-2-readers.tex
\section{A wait-free linearizable implementation of a \texorpdfstring{\Reg{1}{2}}{} from
	atomic \texorpdfstring{\Reg{1}{1}s}{}}\label{whatever}

 Algorithm~\ref{1w2r} gives a wait-free linearizable implementation $\Iup{2}$ of a $\Reg{1}{2}$ from
	atomic $\Reg{1}{1}$s.
This algorithm is a simpler version of Algorithm~\ref{1wnr} for the \emph{valid} implementation $\I{n}$ of a $\Reg{1}{n}$ (Section~\ref{impIn}):
	$\Iup{2}$ has only two readers, namely $p$ and $q$,
	so preventing new-old inversions among readers is easier.
In contrast to Algorithm~\ref{1wnr}, the code of Algorithm~\ref{1w2r} has no parallel threads.
We now prove the correctness of~$\Iup{2}$.

\input{code-algo-2-readers}

\input{proof-algo-2-readers}

%% file: code-algo-2-readers.tex

\begin{algorithm}[t]
\caption{Implementation $\Iup{2}$ of a $\Reg{1}{2}$
	writable by~$w$ and readable by $p$ and $q$.
$\Iu{2}$~uses atomic $\Reg{1}{1}$s.
}\label{1w2r} 
\ContinuedFloat

\vspace{-3mm}
{\footnotesize
\begin{multicols}{2}

\textsc{Atomic Registers}
\vspace{.7mm}

~~~$R_{wp}$: 
	$\Reg{1}{1}$;
	initially $(\done,\langle 0,u_0 \rangle)$

~~~$R_{wq}$: 
	$\Reg{1}{1}$; initially $(\done,\langle 0,u_0 \rangle)$

~~~$R_{pq}$: 
	$\Reg{1}{1}$;
	initially $\langle 0 , \rinit \rangle$

\columnbreak

\textsc{Local variables}

\vspace{.7mm}

~~~~$c$: variable of $w$; initially $0$

~~~~$\lw$: variable of $w$; initially $\langle 0 , \rinit \rangle$

~~~~$\lr$: variable of $q$
	 initially $\langle 0 , \rinit \rangle$

\end{multicols}
\hrule
\vspace{-2mm}
\begin{algorithmic}[1]
\begin{multicols}{2}
\Statex
\textsc{\Wu($u$):}  ~~~~~~~~$\triangleright$ executed by the writer $w$
\Indent
\State \label{Sc} $c \gets c+1$  
\State \label{ScallW} \textbf{call} \textsc{\wu($\langle c,u \rangle$)} 
\State \Return $\Done$
\EndIndent

\columnbreak

\noindent
\textsc{$\Ru$():}\Comment{executed by any reader $r \in \{p,q\}$}
\Indent
\State \label{Scallru}\textbf{call} $\ru_r$()
\State \textbf{if} this call returns some tuple $\langle k,u \rangle$ \textbf{then}
\Indent
\State $\Return$ $u$ 
\EndIndent
\State \textbf{else} \Return $\bot$ 
\EndIndent
\end{multicols}
\vspace{-1mm}
\hrule
\vspace{4mm}
\Statex

\textsc{\wu($\langle k,u \rangle$):} \Comment{executed by $w$ to do its $k$-th write}
\Indent
\State \label{SwpP} $R_{wp}\gets (\ready,\lw,\langle k,u \rangle)$
\State \label{SwqP} $R_{wq}\gets (\ready,\lw,\langle k,u \rangle)$
\State \label{SwpC} $R_{wp}\gets (\done,\langle k,u \rangle)$
\State \label{SwqC} $R_{wq}\gets (\done,\langle k,u \rangle)$
\State $\lw \gets \langle k,u \rangle$
\State \Return $\Done$
\EndIndent

\vspace*{2mm}

\noindent
\textsc{$\ru_p$():}\Comment{executed by reader $p$}
\Indent
\State \label{SpCommit}\textbf{if} $R_{wp}=(\done,\langle k,u \rangle)$ for some $\langle k,u \rangle$ \textbf{then}
\Indent
\State \label{Stellq} \label{Sptoq}$R_{pq}\gets \langle k,u \rangle$
\State \label{SpCrt} \Return $\langle k,u \rangle$
\EndIndent
\State \label{SpPre}\textbf{elseif} $R_{wp}=(\ready, \lw, - )$  for some $\lw$ \textbf{then}
\Indent
\State \label{SpPrt}\Return $\lw$
\EndIndent
\State \textbf{else} \Return $\bot$ 
\EndIndent

\vspace*{2mm}

\noindent
\textsc{$\ru_q()$:}\Comment{executed by reader $q$}
\Indent
\State \label{SqCommit} \textbf{if} $R_{wq}=(\done,\langle k,u \rangle)$ for some $\langle k,u \rangle$ \textbf{then}
\Indent
\State \label{SqCrt} \Return $\langle k,u \rangle$
\EndIndent

\State \label{SqPre} \textbf{elseif} $R_{wq}=(\ready,\lw,\langle k,u \rangle)$ for some $\lw$ and some $\langle k,u \rangle$ \textbf{then}

\Indent

\EndIndent

\Indent
\State \label{Sqfromp}\textbf{if} $R_{pq}=\langle k',- \rangle$ \text{for some} $k'\ge k$ \textbf{then}
\Indent
\State \label{Sqtor} $\lr\gets \langle k,u \rangle$
\State \label{SqPrt2}\Return $\langle k,u \rangle$
\EndIndent

\State \label{Srtoq}\textbf{elseif} $\lr=\langle k',- \rangle$ and some $k'\ge k$ \textbf{then}
\Indent

\State \label{SqPrt3}\Return $\langle k,u \rangle$

\EndIndent

\State \label{Set2} \textbf{else}
\Indent
\State\Return \label{SqPrtlw} $\lw$
\EndIndent
\EndIndent

\State \label{Sqrtbot} \textbf{else} \Return $\bot$ 
\EndIndent

\end{algorithmic}
}
\end{algorithm}

%% file: proof-algo-2-readers.tex

Since the code of Algorithm~\ref{1w2r} does not contain any loop or wait statement, it is clear that
	every call to the \Wu() and \Ru() procedures by any correct process terminates
	with a return value in a bounded number of its own steps.
	Thus:

\begin{observation}\label{Obs-algo2iswf}
 The implementation $\Iup{2}$ is bounded wait-free.
\end{observation}

The proof that $\Iup{2}$ is linearizable is in many parts similar (or even identical) to the proof that
	that the register implementation $\I{n}$ is linearizable (Theorem~\ref{Theo-linearizability} in Section~\ref{impIn}).
It is given here for completeness.

\newpage
To prove that $\Iup{2}$ is linearizable, we consider two cases:
 
 \smallskip
 \textsc{Case 1:} \emph{The writer $w$ of the register implemented by $\Iup{2}$ is {\ml}.}
 By Definition~\ref{LinearizableByz}, $\Iup{2}$ is (trivially) linearizable in this case.
 
 \textsc{Case 2:} \emph{The writer $w$ of the register implemented by $\Iup{2}$ is \emph{not} {\ml}.}
 
 For this case, we now prove that the read and write operations of the implemented register
 	satisfy the linearizability Properties~1 and~2 of Definition~\ref{LinearizableByz}.
	
In the following:
 
\begin{compactitem}
\item $u_0$ is the initial value of the register that $\Iup{2}$ implements.
\item For $k\ge1$, $u_k$ denotes the $k$-th value written by $w$ using
	the procedure
	$\Wu()$.
	More precisely, if $w$ calls $\Wu()$ with a value $u$ and this is its $k$-th call of $\Wu()$,
	then~$u_k$~is~$u$.
\item $v_0$ is $\langle 0, u_0 \rangle$.
\item For $k\ge1$, $v_k$ denotes the $k$-th value written by $w$ using
	the procedure
	$\wu()$. 
\end{compactitem}

\begin{observation}\label{Svk}
For all $k\ge 0$, $v_k = \tp{k,u_k}$.
\end{observation}

\begin{observation}\label{Swwrites}
Let $\wu(v)$ be any write operation by $w$.
Then there is a $k \ge 1$ such that $v=v_k$.
\end{observation}

\begin{observation}\label{Sregcontentstrongw1}
Let $R\in\{R_{wp}, R_{wq}\}$.
If $w$ writes $x$ in $R$,
	then $x=(\done, v_k)$ for some $k \ge 1$ or
	$x=(\ready, v_k , v_{k+1})$ 
	for some $k \ge 0$.
\end{observation}

\begin{observation}\label{Sregcontentstrongp}
Suppose $p$ is not {\ml}.
If $p$ reads $R_{wp}=x$,
	then $x=(\done, v_k)$ or
	$x=(\ready, v_k , v_{k+1})$, for some $k \ge 0$. 
\end{observation}

\begin{observation}\label{Sregcontentstrongq}
Suppose $q$ is not {\ml}.
If $q$ reads $R_{wq}=x$,
then $x=(\done, v_k)$ or
	$x=(\ready, v_k , v_{k+1})$, for some $k \ge 0$. 
\end{observation}

\begin{lemma}\label{Spread}
Suppose $p$ is not {\ml}.
Let $\ru_p(v)$ be any read operation by $p$.
Then there is a $k \ge 0$ such that $v=v_k$, and
\begin{compactitem}
\item $p$ reads $R_{wp} = (\done, v_{k})$ in line~\ref{SpCommit} of $\ru_p(v)$,\footnote{For brevity, we say that
	``a process $r$ reads or writes a register \emph{in line  $x$ of a $\ru_r(-)$ or a $\wu(-)$ operation}'',
	if it reads or writes this register in line $x$ of the $\ru_r()$ or $\wu()$ \emph{procedure}
	executed to do this $\ru_r(-)$ or $\wu(-)$ operation.}
	or
\item $p$ reads $R_{wp} =(\ready, v_{k}, v_{k+1})$ in line~\ref{SpPre} of $\ru_p(v)$. 
\end{compactitem}
\end{lemma}

\begin{proof}
Suppose $p$ is not {\ml}.
Let $\ru_p(v)$ be any read operation by $p$.
Note that $p$ reads $R_{wp}$ in $\ru_p(v)$.
When it does so, by Observation~\ref{Sregcontentstrongp},
	there are two possible cases:
\begin{enumerate}
\item $p$ reads $R_{wp} = (\done, v_{k})$ for some $k\ge 0$ in line~\ref{SpCommit} of $\ru_p(v)$.
 Then $\ru_p(v)$ returns $v_{k}$ in line~\ref{SpCrt},
 	i.e., $v=v_k$.
\item $p$ reads $R_{wp} = (\ready, v_{k}, v_{k+1})$ for some $k\ge 0$ in line~\ref{SpPre} of $\ru_p(v)$.
 Then $\ru_p(v)$ returns $v_{k}$ in line~\ref{SpPrt},
 	i.e., $v=v_k$.
\end{enumerate}
\end{proof}

\begin{lemma}\label{Sqread}
Suppose $q $ is not {\ml}.
Let $\ru_q(v)$ be any read operation by $q$.
Then there is a $k \ge 0$ such that $v=v_k$, and
\begin{compactitem}
\item $q$ reads $R_{wq} = (\done, v_{k})$ in line~\ref{SqCommit} of $\ru_q(v)$, 
\item $q$ reads $R_{wq} =(\ready, v_{k-1}, v_{k})$ in line~\ref{SqPre} of $\ru_q(v)$, or 
\item $q$ reads $R_{wq} =(\ready, v_{k}, v_{k+1})$ in line~\ref{SqPre} of $\ru_q(v)$. 
\end{compactitem}
\end{lemma}

\begin{proof}
Suppose $q $ is not {\ml}.
Let $\ru_q(v)$ be any read operation by $q$.
Note that $q$ reads $R_{wq}$ in $\ru_q(v)$.
When it does so, by Observation~\ref{Sregcontentstrongq},
	there are two possible cases:
\begin{enumerate}
\item $q$ reads $R_{wq} = (\done, v_{k})$ for some $k\ge 0$ in line~\ref{SqCommit} of $\ru_q(v)$.
Then $\ru_q(v)$ returns $v_k$ in line~\ref{SqCrt},
	i.e., $v=v_k$.
\item $q$ reads $R_{wq} = (\ready, v_{k}, v_{k+1})$ for some $k\ge 0$ in line~\ref{SqPre} of $\ru_q(v)$.
Then there are two subcases:
	\begin{compactenum}
	\item $\ru_q(v)$ returns $v_{k+1}$ in line~\ref{SqPrt2} or~\ref{SqPrt3},
		i.e., $v=v_{k+1}$.
	Let $k'=k+1$.
	Then in this case, $q$~reads $R_{wq} =(\ready, v_{k'-1}, v_{k'})$ in line~\ref{SqPre} of $\ru_q(v)$ and $v=v_{k'}$.

	\item $\ru_q(v)$ returns $v_{k}$  in line~\ref{SqPrtlw},
	i.e., $v=v_{k}$.
	\end{compactenum}
\end{enumerate}
\end{proof}

\begin{observation}\label{Smonowpw}
Let $R$ be any register in $\{R_{wp}, R_{wq}\}$.

\begin{enumerate}[(1)]
\item \label{Smonowpw1} If $w$ writes $(\ready, v_{k-1}, v_k)$ in $R$ before $w$ writes $(\done,v_{k'})$ in $R$,
	then \mbox{$k \le k'$}.

\item \label{Smonowpw2} If $w$ writes $(\ready, v_{k-1}, v_k)$ in $R$ before $w$ writes $(\ready, v_{k'-1}, v_{k'})$ in $R$, 
	then $k < k'$.

\item \label{Smonowpw2.5} If $w$ writes $(\done,v_k)$ in $R$ before $w$ writes $(\done,v_{k'})$ in $R$,
	then $k < k'$.

\item \label{Smonowpw3} If $w$ writes $(\done,v_k)$ in $R$ before $w$ writes $(\ready, v_{k'-1}, v_{k'})$ in $R$,
	then $k < k'$.
\end{enumerate}
\end{observation}

\begin{observation}\label{Smonowpwr}
Let $R$ be any register in $\{R_{wp}, R_{wq}\}$.
Suppose $r \in \{p, q \}$ is not {\ml}.

\begin{enumerate}[(1)]
\item \label{Smonowpwr1} If $w$ writes $(\ready, v_{k-1}, v_k)$ in $R$ before $r$ reads $(\done,v_{k'})$ in $R$,
	then \mbox{$k \le k'$}.

\item \label{Smonowpwr2} If $w$ writes $(\ready, v_{k-1}, v_k)$ in $R$ before $r$ reads $(\ready, v_{k'-1}, v_{k'})$ in $R$, 
	then \mbox{$k \le k'$}.

\item \label{Smonowpwr2.5} If $w$ writes $(\done,v_k)$ in $R$ before $r$ reads $(\done,v_{k'})$ in $R$,
	then \mbox{$k \le k'$}.

\item \label{Smonowpwr3} If $w$ writes $(\done,v_k)$ in $R$ before $r$ reads $(\ready, v_{k'-1}, v_{k'})$ in $R$,
	then $k < k'$.

\end{enumerate}
\end{observation}

\begin{observation}\label{Smonowp}
Let $R$ be any register in $\{R_{wp}, R_{wq}\}$.
Suppose $r$ and $r'$ are non-{\ml} processes in $\{p, q\}$.

\begin{enumerate}[(1)]
\item \label{Smonowp1} If $r$ reads $R=(\ready, v_{k-1}, v_k)$ before $r'$ reads $R=(\done,v_{k'})$,
	then \mbox{$k \le k'$}.

\item \label{Smonowp2} If $r$ reads $R=(\ready, v_{k-1}, v_k)$ before $r'$ reads $R=(\ready, v_{k'-1}, v_{k'})$, 
	then \mbox{$k \le k'$}.

\item \label{Smonowp2.5} If $r$ reads $R=(\done,v_k)$ before $r'$ reads $R=(\done,v_{k'})$,
	then \mbox{$k \le k'$}.

\item \label{Smonowp3} If $r$ reads $R=(\done,v_k)$ before $r'$ reads $R=(\ready, v_{k'-1}, v_{k'})$,
	then $k < k'$.
\end{enumerate}
\end{observation}

\textsc{Proof of linearizability Property 1.}
We now prove that the write and read operations of the register that $\Iup{2}$ implements satisfy Property~1 of Definition~\ref{LinearizableByz},
	i.e., processes read the ``current'' value of the register.
 To do so,
 	we first prove this for the writes and reads of the lower-level procedures $\wu()$ and $\ru_r()$ for all readers $r \in \{p, q\}$
	(Lemma~\ref{Sp1r}),
 	and then prove it for the writes and reads of the high-level procedures $\Wu()$ and $\Ru()$ (Lemma~\ref{SCp1}).

\begin{lemma}\label{Sp1r}
If $\ru_r(v)$ is a read operation by a non-{\ml} process $r\in \{p,q\}$
	then:
\begin{compactitem}
	\item there is a $\wu(v)$ operation that immediately precedes $\ru_r(v)$ or is concurrent with $\ru_r(v)$, or
	\item  $v = v_0$ and no $\wu(-)$ operation precedes $\ru_r(v)$.
\end{compactitem}
\end{lemma}

\begin{proof}
Let $\ru_r(v)$ be any read operation by a non-{\ml} process $r\in\{ p,q \}$.
 By Lemmas~\ref{Spread} and~\ref{Sqread}, $v = v_k$ for some $k \ge 0$.
 We now show that:
 \begin{compactitem}
 \item if $k=0$ then no $\wu(-)$ operation precedes $\ru_r(v_k)$, and
 \item if $k>0$ then a $\wu(v_{k})$ operation immediately precedes~$\ru_r(v_k)$ or is concurrent with $\ru_r(v_k)$.
 \end{compactitem}
\vspace{2mm}
There are two cases: $r=p$~or~$r=q$.
\begin{itemize}
\item Case 1: $r=p$.
By Lemma~\ref{Spread},
	there are two cases:
		
	\begin{enumerate}[1)]
 	\item $p$ reads $R_{wp} = (\done, v_{k})$ in line~\ref{SpCommit} of $\ru_p(v_k)$.
	There are two cases:
		\begin{enumerate}[i.]
		\item $k=0$.
			Suppose, for contradiction, 
			that there is a $\wu(v)$ operation that precedes $\ru_p(v_0)$.
			By Observation~\ref{Swwrites},
			$v = v_i$ for some $i \ge 1$.
		So $w$ writes $(\done, v_{i})$ into $R_{wp}$ in line~\ref{SwpC} of $\wu(v_i)$ before $p$ reads $R_{wp} = (\done, v_{0})$ in line~\ref{SpCommit} of $\ru_p(v_k)$.
		By Observation~\ref{Smonowpwr}(\ref{Smonowpwr2.5}),
			$i \le 0$ --- a contradiction.
		So no $\wu(-)$ operation precedes $\ru_p(v_0)$.

		\item $k>0$.
		Then $w$ writes $(\done,v_k)$ into $R_{wp}$ in line~\ref{SwpC} of $\wu(v_k)$ 
		before $p$ reads $R_{wp}= (\done,v_k)$ in $\ru_p(v_k)$.
		So the $\wu(v_k)$ operation precedes $\ru_p(v_k)$ or is concurrent with $\ru_p(v_k)$.
		We now show that if $\wu(v_k)$ precedes $\ru_p(v_k)$, then $\wu(v_k)$ immediately precedes $\ru_p(v_k)$.
		Suppose, for contradiction, 
			that $\wu(v_k)$ precedes $\ru_p(v_k)$ but does not immediately precede $\ru_p(v_k)$.
		Then there is a $\wu(v_i)$ operation that immediately precedes $\ru_p(v_k)$.
		Clearly, the $\wu(v_k)$ operation precedes the $\wu(v_i)$ operation, and so $i > k$.
		Furthermore, $w$ writes $(\done, v_{i})$ into $R_{wp}$ in line~\ref{SwpC} of $\wu(v_i)$ before $p$ reads $R_{wp} = (\done, v_{k})$ in line~\ref{SpCommit} of $\ru_p(v_k)$.
		By Observation~\ref{Smonowpwr}(\ref{Smonowpwr2.5}),
			$i \le k$ --- a contradiction.
		 Therefore the $\wu(v_k)$ operation immediately precedes $\ru_p(v_k)$ or is concurrent with $\ru_p(v_k)$.
		\end{enumerate}

	\item $p$ reads $R_{wp} =(\ready, v_{k}, v_{k+1})$ in line~\ref{SpPre} of $\ru_p(v_k)$.
	Then this read occurs \emph{after}
	$w$ writes $(\ready, v_{k}, v_{k+1})$ in $R_{wp}$ in line~\ref{SwpP} of the $\wu(v_{k+1})$ operation.
	Furthermore, by Observation~\ref{Smonowpwr}(\ref{Smonowpwr3}), 
	this read occurs \emph{before} $w$ writes $(\done, v_{k+1})$ in $R_{wp}$
	in line~\ref{SwpC} of the $\wu(v_{k+1})$ operation.
	Therefore the $\wu(v_{k+1})$ operation is concurrent with $\ru_p(v_k)$.
	There are two cases:

	\begin{enumerate}[i.]
	\item $k=0$.
	Since $\wu(v_1)$ is concurrent with $\ru_p(v_0)$,
	no $\wu(-)$ operation precedes $\ru_p(v_0)$.

	\item $k>0$. 
	Since $\wu(v_{k+1})$ is concurrent with $\ru_p(v_k)$, 
		$\wu(v_{k})$ immediately precedes~$\ru_p(v_k)$ or is concurrent with $\ru_p(v_k)$.
	\end{enumerate}
\end{enumerate}

\item Case 2: $r = q $.
	By Lemma~\ref{Sqread},
	there are three cases:\begin{enumerate}[1)]
\item $q$ reads $R_{wq} = (\done, v_{k})$ in line~\ref{SqCommit} of $\ru_q(v_k)$.
	There are two cases:
		\begin{enumerate}[i.]
		
		\item $k=0$.
			Suppose, for contradiction, 
			that there is a $\wu(v)$ operation that precedes $\ru_p(v_0)$.
			By Observation~\ref{Swwrites},
			$v = v_i$ for some $i \ge 1$.
		So $w$ writes $(\done, v_{i})$ into $R_{wq}$ in line~\ref{SwqC} of $\wu(v_i)$ before $q$ reads $R_{wq} = (\done, v_{0})$ in line~\ref{SqCommit} of $\ru_q(v_k)$.
		By Observation~\ref{Smonowpwr}(\ref{Smonowpwr2.5}),
			$i \le 0$ --- a contradiction.
		So no $\wu(-)$ operation precedes $\ru_q(v_0)$.
		
		\item $k>0$.
		Then $w$ writes $(\done,v_k)$ into $R_{wq}$ in line~\ref{SwqC} of $\wu(v_k)$
		 before $q$ reads $R_{wq}= (\done,v_k)$ in line~\ref{SqCommit} of $\ru_q(v_k)$.
		So the $\wu(v_k)$ operation precedes $\ru_q(v_k)$ or is concurrent with $\ru_q(v_k)$.
		We now show that if $\wu(v_k)$ precedes $\ru_q(v_k)$, then $\wu(v_k)$ immediately precedes $\ru_q(v_k)$.
		Suppose, for contradiction, 
			that $\wu(v_k)$ precedes $\ru_q(v_k)$ but does not immediately precede $\ru_q(v_k)$.
		Then there is a $\wu(v_i)$ operation that immediately precedes $\ru_p(v_k)$.
		Clearly, the $\wu(v_k)$ operation precedes the $\wu(v_i)$ operation, and so $i > k$.		
		Furthermore, $w$ writes $(\done, v_{i})$ into $R_{wq}$ in line~\ref{SwqC} of $\wu(v_i)$ before $q$ reads $R_{wq} = (\done, v_{k})$ in line~\ref{SqCommit} of $\ru_q(v_k)$.
		By Observation~\ref{Smonowpwr}(\ref{Smonowpwr2.5}),
			$i \le k$ --- a contradiction.
		 Therefore the $\wu(v_k)$ operation immediately precedes $\ru_q(v_k)$ or is concurrent with $\ru_q(v_k)$.

		\end{enumerate}
		
\item $q$ reads $R_{wq} =(\ready, v_{k-1}, v_{k})$ in line~\ref{SqPre} of $\ru_q(v_k)$.
	Then this read occurs \emph{after}
	$w$ writes $(\ready, v_{k-1}, v_{k})$ in $R_{wq}$ in line~\ref{SwqP} of the $\wu(v_k)$ operation.
	Furthermore, by Observation~\ref{Smonowpwr}(\ref{Smonowpwr3}),
	this read occurs \emph{before} $w$ writes $(\done,v_k)$ in $R_{wq}$
	in line~\ref{SwqC} of the $\wu(v_k)$ operation.
	Therefore the $\wu(v_k)$ operation is concurrent with $\ru_q(v_k)$.

\item $q$ reads $R_{wq} =(\ready, v_{k}, v_{k+1})$ in line~\ref{SqPre} of $\ru_q(v_k)$.
	Then this read occurs after
	$w$ writes $(\ready, v_{k}, v_{k+1})$ in $R_{wq}$ in line~\ref{SwqP} of the $\wu(v_{k+1})$ operation.
	Furthermore, by Observation~\ref{Smonowpwr}(\ref{Smonowpwr3}),
	this read occurs before $w$ writes
	$(\done,v_{k+1})$ in $R_{wq}$ in line~\ref{SwqC} of the $\wu(v_{k+1})$ operation.
	Therefore the $\wu(v_{k+1})$ operation is concurrent with $\ru_q(v_k)$.
	There are two cases:
	\begin{enumerate}[i.]
	
	\item $k=0$.
	Since $\wu(v_1)$ is concurrent with $\ru_q(v_0)$,
	no $\wu(-)$ operation precedes $\ru_q(v_0)$.

	\item $k>0$.
	Since $\wu(v_{k+1})$ is concurrent with $\ru_q(v_k)$, 
		$\wu(v_{k})$ is concurrent with $\ru_q(v_k)$ or immediately precedes $\ru_q(v_k)$.
	\end{enumerate}
\end{enumerate}
\end{itemize}
\end{proof}

We now prove that the write and read operations
	of the high-level procedures $\Wu()$ and $\Ru()$
	satisfy Property~1 of Definition~\ref{LinearizableByz}.

By Observation~\ref{Svk}, Lemmas~\ref{Spread} and~\ref{Sqread}, and the code of the procedure $\Ru()$,
	we have:
\begin{observation}\label{SintegrityR}
If $\Ru(u)$ is an operation by a non-{\ml} process $r \in \{p,q\}$,
	then $u=u_k$ for some $k \ge 0$.
\end{observation}

\begin{observation}\label{SRr}
If $\Ru(u_k)$ is an operation by a non-{\ml} process $r\in\{p,q\}$,
then $r$ invokes and completes a $\ru_r(v_k)$ operation in $\Ru(u_k)$.
\end{observation}

\begin{observation}\label{SWw}
If $\Wu(u_k)$ is a completed operation by $w$,
	then $w$ invokes and completes a $\wu(v_k)$ operation in $\Wu(u_k)$.
\end{observation}

We now prove that the $\Wu(-)$ and $\Ru(-)$ operations satisfy Property~1 of Definition~\ref{LinearizableByz}.

\begin{lemma}\label{SCp1} \emph{[\textbf{Property 1}: Reading a ``current'' value]}

If $\Ru(u)$ is a read operation by a non-{\ml} process $r\in\{p,q\}$
	then:
\begin{compactitem}	
	\item there is a $\Wu(u)$ operation that immediately precedes $\Ru(u)$ or is concurrent with $\Ru(u)$, or	
	\item $u = u_0$ and no $\Wu(-)$ operation precedes $\Ru(u)$.
\end{compactitem}
\end{lemma}

\begin{proof}
Let $\Ru(u)$ be any read operation by a non-{\ml} process $r\in\{ p,q \}$.
By Observation~\ref{SintegrityR},
	$u=u_k$ for some $k\ge 0$.
There are two cases:
	\begin{enumerate}[(1)]
	\item $k=0$.
		Suppose, for contradiction, that a $\Wu(u_i)$ operation precedes $\Ru_r(u_0)$.
		Note that $i \ge 1$. 
		By Observations~\ref{SWw} and~\ref{SRr}, 
		a $\wu(v_i)$ operation precedes a $\ru_r(v_0)$ operation.
		Since process $r$ is not {\ml},
			by Lemma~\ref{Sp1r},
		 there is no $\wu(-)$ operation that precedes $\ru_r(v_0)$ --- a contradiction.

	\item $k>0$.
		By Observation~\ref{SRr},
		$r$ invokes and completes a $\ru_r(v_k)$ operation in $\Ru(u_k)$.
		Since $k>0$,
		 by Lemma~\ref{Sp1r},
		 there is a $\wu(v_{k})$ operation that immediately precedes $\ru_r(v_k)$ or is concurrent with $\ru_r(v_k)$.
		 Let $\Wu(u_{k})$ be the operation in which $w$ invokes the $\wu(v_{k})$ operation.
		Since $\wu(v_k)$ immediately precedes $\ru_r(v_k)$ or is concurrent with $\ru_r(v_k)$,
		the $\Wu(u_{k})$ operation immediately precedes $\Ru(u_k)$ or is concurrent with $\Ru(u_k)$.
	\end{enumerate}
\end{proof}

\textsc{Proof of linearizability Property 2.}
We now prove that the write and read operations of the register that $\Iup{2}$ implements satisfy Property~2 of Definition~\ref{LinearizableByz},
i.e., we prove that there are no ``new-old'' inversions in the values that processes read.
 To do so,
 	we first prove this for the writes and reads of the lower-level procedures $\wu()$, $\ru_p()$,
	and $\ru_q()$,
 	and then prove it for the writes and reads of the high-level procedures $\Wu()$ and $\Ru()$ (Lemma~\ref{SCp2}).

We first show that there are no ``new-old'' inversions in the consecutive reads of process $p$.

\begin{lemma}\label{Smonopread}
Suppose $p$ is not {\ml}.
If $\ru_p(v_k)$ and $\ru_p(v_{k'})$ are read operations by $p$,
	and $\ru_p(v_k)$ precedes $\ru_p(v_{k'})$, then \mbox{$k \le k'$}.
\end{lemma}

\begin{proof}
Suppose $p$ is not {\ml}.
Let $\ru_p(v_k)$ and $\ru_p(v_{k'})$ be read operations by $p$ such that $\ru_p(v_k)$ precedes $\ru_p(v_{k'})$.
By Lemma~\ref{Spread}, the following occurs:

\begin{compactenum}
	\item $p$ reads $R_{wp} = (\done, v_{k})$ \ilns{SpCommit} $\ru_p(v_k)$,
	or
	\item $p$ reads $R_{wp} =(\ready, v_k, v_{k+1})$ \ilns{SpPre} $\ru_p(v_k)$
\end{compactenum}

\emph{before} the following occurs:

\begin{compactenum}
	\item $p$ reads $R_{wp} = (\done, v_{k'})$ \ilns{SpCommit} $\ru_p(v_{k'})$,
	or
	\item $p$ reads $R_{wp} =(\ready, v_{k'}, v_{k'+1})$ \ilns{SpPre} $\ru_p(v_{k'})$.
\end{compactenum}

\smallskip
So there are four possible cases:

\begin{compactenum}
\item $p$ reads $R_{wp} = (\done, v_{k})$ \iln{pCommit} $\ru_p(v_k)$
	before
	$p$ reads $R_{wp} = (\done, v_{k'})$ \iln{pCommit} $\ru_p(v_{k'})$. 
 By Observation~\ref{Smonowp}(\ref{Smonowp2.5}),
 	 \mbox{$k \le k'$}.

\item $p$ reads $R_{wp} = (\done, v_{k})$ \iln{pCommit} $\ru_p(v_k)$
	before
	$p$ reads $R_{wp} =(\ready, v_{k'}, v_{k'+1})$ \iln{pPre} $\ru_p(v_{k'})$.
 By Observation~\ref{Smonowp}(\ref{Smonowp3}),
 	 $k < k'+1$. 
 	 So \mbox{$k \le k'$}.

\item $p$ reads $R_{wp} =(\ready, v_k, v_{k+1})$ \iln{pPre} $\ru_p(v_k)$
	before
	$p$ reads $R_{wp} = (\done, v_{k'})$ \iln{pCommit} $\ru_p(v_{k'})$.
 By Observation~\ref{Smonowp}(\ref{Smonowp1}),
 	 $k+1 \le k'$.
 	 So \mbox{$k \le k'$}.

\item $p$ reads $R_{wp} =(\ready, v_k, v_{k+1})$
	\iln{pPre} $\ru_p(v_k)$
	before $p$ reads $R_{wp} =(\ready, v_{k'}, v_{k'+1})$
	\iln{pPre} $\ru_p(v_{k'})$.
By Observation~\ref{Smonowp}(\ref{Smonowp2}), 
	$k+1 \le k'+1$. 
	So \mbox{$k \le k'$}. 
\end{compactenum}
\end{proof}

To prove that there are no ``new-old'' inversions between the reads of $p$ and $q$, and also between the reads of $q$ itself, we first make some straightforward observations that are clear from the code of $\Iup{2}$.
We first note that the counters of the tuples in the register $R_{pq}$ do not decrease.

\begin{observation}\label{Smonopq}
Suppose $p$ is not {\ml}.
	If $p$ writes $v_k$ in $R_{pq}$ before $p$ writes $v_{k'}$ in $R_{pq}$,
	then \mbox{$k \le k'$}.
\end{observation}

\begin{observation}\label{Smonopqwr}
Suppose $p$ and $q$ are not {\ml}.
	If $p$ writes $v_k$ in $R_{pq}$ before $q$ reads $R_{pq}=v_{k'}$,
	then \mbox{$k \le k'$}.
\end{observation}

The following observations relate the counters of the tuples that $w$ succesively writes in registers $R_{wp}$ 
and $R_{wq}$.
\begin{observation}\label{SwpwQw}
\begin{enumerate}[(1)]
\item \label{SwpwQw1} If $w$ writes $(\ready, v_{k-1}, v_k)$ in $R_{wp}$ before $w$ writes $(\ready, v_{k'-1}, v_{k'})$ in $R_{wq}$,
	then \mbox{$k \le k'$}.
\item \label{SwpwQw2} If $w$ writes $(\ready, v_{k-1}, v_k)$ in $R_{wp}$ before $w$ writes $(\done,v_{k'})$ in $R_{wq}$, 
	then \mbox{$k \le k'$}.

\item \label{SwpwQw1.5} If $w$ writes $(\done,v_k)$ in $R_{wp}$ before $w$ writes $(\ready, v_{k'-1}, v_{k'})$ in $R_{wq}$,
	then \mbox{$k < k'$}.

\item \label{SwpwQw3} If $w$ writes $(\done,v_k)$ in $R_{wp}$ before $w$ writes $(\done,v_{k'})$ in $R_{wq}$,
	then \mbox{$k \le k'$}.

\item \label{SwpwQw4} If $w$ writes $(\ready, v_{k-1}, v_k)$ in $R_{wq}$ before $w$ writes $(\ready, v_{k'-1}, v_{k'})$ in $R_{wp}$,
	then \mbox{$k < k'$}.

\item \label{SwpwQw4.5} If $w$ writes $(\ready, v_{k-1}, v_k)$ in $R_{wq}$ before $w$ writes $(\done,v_{k'})$ in $R_{wp}$,
	then \mbox{$k \le k'$}.

\item \label{SwpwQw6} If $w$ writes $(\done,v_k)$ in $R_{wq}$ before $w$ writes $(\ready, v_{k'-1}, v_{k'})$ in $R_{wp}$,
	then \mbox{$k < k'$}.

\item \label{SwpwQw5} If $w$ writes $(\done,v_k)$ in $R_{wq}$ before $w$ writes $(\done,v_{k'})$ in $R_{wp}$, 
	then \mbox{$k < k'$}.
\end{enumerate}
\end{observation}

The next observations relate the counters of the tuples that $p$ and $q$ read
	from $R_{wp}$ and $R_{wq}$, respectively.

\begin{observation}\label{SwpwQ}
Suppose $p$ and $q $ are not {\ml}.
\begin{enumerate}[(1)]
\item \label{SwpwQ1} If $p$ reads $R_{wp}=(\ready, v_{k-1}, v_k)$ before $q$ reads $R_{wq}=(\ready, v_{k'-1}, v_{k'})$,
	then \mbox{$k \le k'$}.
\item \label{SwpwQ2} If $p$ reads $R_{wp}=(\ready, v_{k-1}, v_k)$ before $q$ reads $R_{wq}=(\done,v_{k'})$, 
	then \mbox{$k-1 \le k'$}.

\item \label{SwpwQ1.5} If $p$ reads $R_{wp}=(\done,v_k)$ before $q$ reads $R_{wq}=(\ready, v_{k'-1}, v_{k'})$,
	then \mbox{$k \le k'$}.

\item \label{SwpwQ3} If $p$ reads $R_{wp}=(\done,v_k)$ before $q$ reads $R_{wq}=(\done,v_{k'})$,
	then \mbox{$k \le k'$}.

\item \label{SwpwQ4} If $q$ reads $R_{wq}=(\ready, v_{k-1}, v_k)$ before $p$ reads $R_{wp}=(\ready, v_{k'-1}, v_{k'})$,
	then \mbox{$k \le k'$}.

\item \label{SwpwQ4.5} If $q$ reads $R_{wq}=(\ready, v_{k-1}, v_k)$ before $p$ reads $R_{wp}=(\done,v_{k'})$,
	then \mbox{$k \le k'$}.

\item \label{SwpwQ6} If $q$ reads $R_{wq}=(\done,v_k)$ before $p$ reads $R_{wp}=(\ready, v_{k'-1}, v_{k'})$,
	then \mbox{$k+1 \le k'$.}

\item \label{SwpwQ5} If $q$ reads $R_{wq}=(\done,v_k)$ before $p$ reads $R_{wp}=(\done,v_{k'})$, 
	then \mbox{$k \le k'$}.
\end{enumerate}
\end{observation}

Now we prove that there is no ``new-old'' inversion for a read by $p$ that precedes a read by $q$.

\begin{lemma} \label{Snopq}
If $\ru_p(v_k)$ and $\ru_q(v_{k'})$ are read operations by non-{\ml} processes $p$ and $q$ respectively,
	and $\ru_p(v_k)$ precedes $\ru_q(v_{k'})$, then \mbox{$k \le k'$}.
\end{lemma}

\begin{proof}
Suppose processes $p$ and $q $ are not {\ml}.
Let $\ru_p(v_k)$ and $\ru_q(v_{k'})$ be read operations by $p$ and $q$ respectively,
	 such that $\ru_p(v_k)$ precedes $\ru_q(v_{k'})$.
By Lemmas~\ref{Spread} and~\ref{Sqread}, the following occurs:

\begin{compactenum}
	\item $p$ reads $R_{wp} = (\done, v_{k})$ \ilns{SpCommit} $\ru_p(v_k)$,
	or
	\item $p$ reads $R_{wp} =(\ready, v_k, v_{k+1})$ \ilns{SpPre} $\ru_p(v_k)$
\end{compactenum}

\emph{before} the following occurs:

\begin{compactenum}
	\item $q$ reads $R_{wq} = (\done, v_{k'})$ \ilns{SqCommit} $\ru_q(v_{k'})$, or
	
	\item $q$ reads $R_{wq} =(\ready, v_{k'-1}, v_{k'})$ \ilns{SqPre} $\ru_q(v_{k'})$, or

	\item $q$ reads $R_{wq} =(\ready, v_{k'}, v_{k'+1})$ \ilns{SqPre} $\ru_q(v_{k'})$.
\end{compactenum}

\smallskip
So there are six possible cases:

\begin{compactenum}

\item $p$ reads $R_{wp} = (\done, v_{k})$
	\iln{pCommit} $\ru_p(v_k)$
	before $q$ reads $R_{wq} = (\done, v_{k'})$
	\iln{qCommit} $\ru_q(v_{k'})$.
 By Observation~\ref{SwpwQ}(\ref{SwpwQ3}),
 	 \mbox{$k \le k'$}.

\item $p$ reads $R_{wp} = (\done, v_{k})$
	\iln{pCommit} $\ru_p(v_k)$
	before $q$ reads $R_{wq} =(\ready, v_{k'-1}, v_{k'})$
	\iln{qPre} $\ru_q(v_{k'})$.
 By Observation~\ref{SwpwQ}(\ref{SwpwQ1.5}),
 	 \mbox{$k \le k'$}.

\item $p$ reads $R_{wp} = (\done, v_{k})$
	\iln{pCommit} $\ru_p(v_k)$
	before $q$ reads $R_{wq} =(\ready, v_{k'}, v_{k'+1})$
	\iln{qPre} $\ru_q(v_{k'})$.
 By Observation~\ref{SwpwQ}(\ref{SwpwQ1.5}),
 	 $k \le k'+1$.
 
\begin{compactenum}[i.]
\item $k < k'+1$. Then $k\le k'$.

\item $k = k'+1$. 
	We now show that this case is impossible.
	Since $k = k'+1$,
	$q$ reads $R_{wq} =(\ready, v_{k-1}, v_{k})$ in line~\ref{SqPre} of $\ru_q(v_{k-1})$,
	 and $\ru_q(v_{k-1})$ returns in line~\ref{SqPrtlw}.
So $q$ read $R_{pq}$ in line~\ref{Sqfromp} of $\ru_q(v_{k-1})$ before $\ru_q(v_{k-1})$ returns in line~\ref{SqPrtlw}.
Thus, since $p$ reads $R_{wp} = (\done, v_{k})$ in $\ru_p(v_k)$,
	$p$ writes $v_{k}$ in $R_{pq}$ in line~\ref{Sptoq} of $\ru_p(v_k)$.
Since $\ru_p(v_k)$ precedes $\ru_q(v_{k-1})$, 
	 $p$ writes $v_{k}$ in $R_{pq}$ in $\ru_p(v_k)$ before $q$ reads $R_{pq}$ in line~\ref{Sqfromp} of $\ru_q(v_{k-1})$.
By Observation~\ref{Smonopqwr}, 
	$q$ reads $R_{pq}=v_{\ell}$,
	 for some $\ell \ge k$,
	 in line~\ref{Sqfromp} of $\ru_q(v_{k-1})$.
So $\ru_q(v_{k-1})$ returns $v_k$ in line~\ref{SqPrt2} --- a contradiction.
\end{compactenum}

\item $p$ reads $R_{wp} =(\ready, v_{k}, v_{k+1})$
	\iln{pPre} $\ru_p(v_k)$
	before $q$ reads $R_{wq} = (\done, v_{k'})$
	\iln{qCommit} $\ru_q(v_{k'})$.
 By Observation~\ref{SwpwQ}(\ref{SwpwQ2}), $(k+1)-1 \le k'$. So \mbox{$k \le k'$}.

\item $p$ reads $R_{wp} =(\ready, v_{k}, v_{k+1})$
	\iln{pPre} $\ru_p(v_k)$
	before $q$ reads $R_{wq} =(\ready, v_{k'-1}, v_{k'})$
	\iln{qPre} $\ru_q(v_{k'})$.
By Observation~\ref{SwpwQ}(\ref{SwpwQ1}), $k+1 \le k'$. So \mbox{$k \le k'$}.

\item $p$ reads $R_{wp} =(\ready, v_{k}, v_{k+1})$
	\iln{pPre} $\ru_p(v_k)$
	before $q$ reads $R_{wq} =(\ready, v_{k'}, v_{k'+1})$
	\iln{qPre} $\ru_q(v_{k'})$.
By Observation~\ref{SwpwQ}(\ref{SwpwQ1}), $k+1 \le k'+1$. So \mbox{$k \le k'$}.
\end{compactenum}
\end{proof}

Now we prove that there is no ``new-old'' inversion for a read by $q$ that precedes a read by $p$.

\begin{lemma}\label{Snoqp}
	If $\ru_q(v_k)$ and $\ru_p(v_{k'})$ are read operations by non-{\ml} processes $q$ and $p$ respectively,
	and $\ru_q(v_k)$ precedes $\ru_p(v_{k'})$, then \mbox{$k \le k'$}.
\end{lemma}

\begin{proof}
Suppose processes $q$ and $p$ are not {\ml}.
Let $\ru_q(v_k)$ and $\ru_p(v_{k'})$ be two read operations by $q$ and $p$ respectively,
	such that $\ru_q(v_k)$ precedes $\ru_p(v_{k'})$
By Lemmas~\ref{Spread} and~\ref{Sqread}, the following occurs:
\begin{compactenum}
	\item $q$ reads $R_{wq} = (\done, v_{k})$ \ilns{SqCommit} $\ru_q(v_k)$,
	or
	
	\item $q$ reads $R_{wq} =(\ready, v_{k-1}, v_{k})$ \ilns{SqPre} $\ru_q(v_k)$,
	or

	\item $q$ reads $R_{wq} =(\ready, v_{k}, v_{k+1})$ \ilns{SqPre} $\ru_q(v_k)$
\end{compactenum}

\emph{before} the following occurs:

\begin{compactenum}
	\item $p$ reads $R_{wp} = (\done, v_{k'})$ \ilns{SpCommit} $\ru_p(v_{k'})$,
	or
	\item $p$ reads $R_{wp} =(\ready, v_{k'}, v_{k'+1})$ \ilns{SpPre} $\ru_p(v_{k'})$.
\end{compactenum}

\smallskip
So there are six possible cases:

\begin{compactenum}

\item $q$ reads $R_{wq} = (\done, v_{k})$ \iln{qCommit} $\ru_q(v_k)$
	before $p$ reads $R_{wp} = (\done, v_{k'})$ \iln{pCommit} $\ru_p(v_{k'})$.
 By Observation~\ref{SwpwQ}(\ref{SwpwQ5}),
 	 \mbox{$k \le k'$}.

 \item $q$ reads $R_{wq} = (\done, v_{k})$ \iln{qCommit} $\ru_q(v_k)$
 	before $p$ reads $R_{wp} =(\ready, v_{k'}, v_{k'+1})$ \iln{pPre} $\ru_p(v_{k'})$.
 By Observation~\ref{SwpwQ}(\ref{SwpwQ6}),
 	 $k+1 \le k'+1$. So \mbox{$k \le k'$}.

 \item $q$ reads $R_{wq} =(\ready, v_{k-1}, v_{k})$ \iln{qPre} $\ru_q(v_k)$
 	before $p$ reads $R_{wp} = (\done, v_{k'})$ \iln{pCommit} $\ru_p(v_{k'})$.
 By Observation~\ref{SwpwQ}(\ref{SwpwQ4.5}),
  \mbox{$k \le k'$}.

\item $q$ reads $R_{wq} =(\ready, v_{k-1}, v_{k})$ \iln{qPre} $\ru_q(v_k)$
	before $p$ reads $R_{wp} =(\ready, v_{k'}, v_{k'+1})$ \iln{pPre} $\ru_p(v_{k'})$.
By Observation~\ref{SwpwQ}(\ref{SwpwQ4}), 
	$k \le k'+1$. 

\begin{compactenum}[i.]
\item $k < k'+1$. Then $k\le k'$.

\item $k = k'+1$.
	We now show that this case is impossible.
	Since $k = k'+1$,
	 $p$ reads $R_{wp} =(\ready, v_{k-1}, v_{k})$ in $\ru_p(v_{k-1})$,
	 and $\ru_p(v_{k-1})$ returns in line~\ref{SpPrt}.
	 
Note that $q$ reads $R_{pq}$ in line~\ref{Sqfromp} of $\ru_q(v_k)$.

\begin{claim}\label{above}
Process $p$ writes $v_{\ell}$ into $R_{pq}$ with some $\ell \ge k$
	before $q$ reads $R_{pq}$ in $\ru_q(v_k)$.
\end{claim}

\begin{proof}
Since $q$ reads $R_{wq} =(\ready, v_{k-1}, v_{k})$ in line~\ref{SqPre} of $\ru_q(v_k)$,
	$\ru_q(v_k)$
	 returns 
	in line~\ref{SqPrt2} or~\ref{SqPrt3}.
So there are two~cases:

	\begin{compactenum}[a.]
		 
	\item $\ru_q(v_k)$
	 returns
		in line~\ref{SqPrt2}.
	Then $q$ reads $R_{pq}=v_{\ell}$ for some $\ell \ge k$ in line~\ref{Sqfromp} of $\ru_q(v_k)$.
	Thus, since $k\ge 1$ and $R_{pq}$ is initialized to $v_0$,
	$p$ wrote $v_{\ell}$ into $R_{pq}$
	before $q$ reads $R_{pq}$ in $\ru_q(v_k)$.

	\item $\ru_q(v_k)$
	 returns
		in line~\ref{SqPrt3}.
Then $q$ reads $\lr=v_{\ell'}$ for some $\ell' \ge k$ in line~\ref{Srtoq} of $\ru_q(v_k)$.
Thus, since $k\ge 1$ and $\lr$ is initialized to $v_0$,
	$q$ wrote $v_{\ell'}$ into $\lr$ in line~\ref{Sqtor} of some $\ru_q(v_{\ell'})$ operation that precedes $\ru_q(v_k)$.
So $q$ read $R_{pq}=v_{\ell}$ for some $\ell \ge \ell' \ge k$ in line~\ref{Sqfromp} of this $\ru_q(v_{\ell'})$ operation
	that precedes $\ru_q(v_k)$.
Thus $p$ wrote $v_{\ell}$ into $R_{pq}$
	before $q$ reads $R_{pq}$ in $\ru_q(v_k)$.
	\end{compactenum}
\end{proof}

From the code of $\ru_p()$ it is clear that if $p$ writes $v_{\ell}$ to $R_{pq}$ (this can occur only in line~\ref{Sptoq} of some $\ru_p(-)$ operation)
	then $p$ previously reads $R_{wp} =(\done, v_{\ell})$ (in line~\ref{SpCommit} of that $\ru_p(-)$ operation).
Thus, by Claim~\ref{above}, $p$ reads $R_{wp} =(\done, v_{\ell})$ with $\ell \ge k$
	before $q$ reads $R_{pq}$ in $\ru_q(v_k)$.
Since $\ru_q(v_k)$
	 precedes $\ru_p(v_{k-1})$, 
	$p$ reads $R_{wp} =(\done, v_{\ell})$ before $p$ reads $R_{wp} =(\ready, v_{k-1}, v_{k})$ in $\ru_p(v_{k-1})$.
By Observation~\ref{SwpwQ}(\ref{SwpwQ6}), $\ell < k$ --- a contradiction.

\end{compactenum}

\item $q$ reads $R_{wq} =(\ready, v_{k}, v_{k+1})$ \iln{qPre} $\ru_q(v_k)$
	before $p$ reads $R_{wp} = (\done, v_{k'})$ \iln{pCommit} $\ru_p(v_{k'})$.
 By Observation~\ref{SwpwQ}(\ref{SwpwQ4.5}), 
 	$k+1 \le k'$. So \mbox{$k \le k'$}.

\item $q$ reads $R_{wq} =(\ready, v_{k}, v_{k+1})$ \iln{qPre} $\ru_q(v_k)$
	before $p$ reads $R_{wp} =(\ready, v_{k'}, v_{k'+1})$ \iln{pPre} $\ru_p(v_{k'})$.
By Observation~\ref{SwpwQ}(\ref{SwpwQ4}), $k+1 \le k'+1$. So \mbox{$k \le k'$}.
\end{compactenum}
\end{proof}

Finally, we show that there are no ``new-old'' inversions in the successive reads of $q$.

To do so, we first observe that the counters of the tuples in the variable $\lr$ of $q$ do not decrease.
To see this, note that if $q$ writes $v_k$ in $\lr$ (this occurs in line~\ref{Sqtor} of $\ru_q(v_k)$)
	then $q$ previously read
	$R_{wq}=(\ready,-,v_k \rangle)$ (in line~\ref{SqPre} of $\ru_q(v_k)$).	
So, by Observation~\ref{Smonowp}(\ref{Smonowp2}), we have:

\begin{observation}\label{Smonorq}
Suppose $q$ is not {\ml}.
	If $q$ writes $v_k$ in $\lr$ before $q$ writes $v_{k'}$ in $\lr$,
	then \mbox{$k \le k'$}.
\end{observation}

\begin{observation}\label{Smonorqwr}
Suppose $q$ is not {\ml}.
	If $q$ writes $v_k$ in $\lr$ before $q$ reads $\lr=v_{k'}$,
	then \mbox{$k \le k'$}.
\end{observation}

\begin{observation}\label{Smonoqrr}
Suppose $q$ is not {\ml}.
	If $q$ reads $\lr=v_k$ before $q$ reads $\lr=v_{k'}$,
	then \mbox{$k \le k'$}.
\end{observation}

\begin{lemma} \label{Snorq}
	If $\ru_q(v_k)$ and $\ru_{q}(v_{k'})$ are read operations by non-{\ml} process $q$,
	and $\ru_q(v_k)$ precedes $\ru_{q}(v_{k'})$, then \mbox{$k \le k'$}.
\end{lemma}

\begin{proof}
Suppose process $q$ is not {\ml}.
Let $\ru_q(v_k)$ and $\ru_{q}(v_{k'})$ be read operations by~$q$,
	 such that $\ru_q(v_k)$ precedes $\ru_{q}(v_{k'})$.
By Lemma~\ref{Sqread}, the following occurs:
\begin{compactenum}
	\item $q$ reads $R_{wq} = (\done, v_{k})$ \ilns{SqCommit} $\ru_q(v_k)$, or
	
	\item $q$ reads $R_{wq} =(\ready, v_{k-1}, v_{k})$ \ilns{SqPre} $\ru_q(v_k)$, or

	\item $q$ reads $R_{wq} =(\ready, v_{k}, v_{k+1})$ \ilns{SqPre} $\ru_q(v_k)$
\end{compactenum}

\emph{before} the following occurs:

\begin{compactenum}
	\item $q$ reads $R_{wq} = (\done, v_{k'})$ \ilns{SqCommit} $\ru_{q}(v_{k'})$, or
	
	\item $q$ reads $R_{wq} =(\ready, v_{k'-1}, v_{k'})$ \ilns{SqPre} $\ru_{q}(v_{k'})$, or

	\item $q$ reads $R_{wq} =(\ready, v_{k'}, v_{k'+1})$ \ilns{SqPre} $\ru_{q}(v_{k'})$.
\end{compactenum}

\smallskip
So there are nine possible cases:

\begin{compactenum}

\item $q$ reads $R_{wq} = (\done, v_{k})$ \iln{qCommit} $\ru_q(v_k)$
	before
	$q$ reads $R_{wq} = (\done, v_{k'})$ \iln{qCommit} $\ru_{q}(v_{k'})$.
  By Observation~\ref{Smonowp}(\ref{Smonowp2.5}),
  	 \mbox{$k \le k'$}.

 \item $q$ reads $R_{wq} = (\done, v_{k})$ \iln{qCommit} $\ru_q(v_k)$
 	before
	$q$ reads $R_{wq} =(\ready, v_{k'-1}, v_{k'})$ \iln{qPre} $\ru_{q}(v_{k'})$.
  By Observation~\ref{Smonowp}(\ref{Smonowp3}),
  	 $k < k'$. So \mbox{$k \le k'$}.

 \item $q$ reads $R_{wq} = (\done, v_{k})$ \iln{qCommit} $\ru_q(v_k)$
 	before
	$q$ reads $R_{wq} =(\ready, v_{k'}, v_{k'+1})$ \iln{qPre} $\ru_{q}(v_{k'})$.
 By Observation~\ref{Smonowp}(\ref{Smonowp3}),
 	 $k < k'+1$. So \mbox{$k \le k'$}.

 \item $q$ reads $R_{wq} =(\ready, v_{k-1}, v_{k})$ \iln{qPre} $\ru_q(v_k)$
 	before
	$q$ reads $R_{wq} = (\done, v_{k'})$ \iln{qCommit} $\ru_{q}(v_{k'})$.
 By Observation~\ref{Smonowp}(\ref{Smonowp1}),
 	 \mbox{$k \le k'$}.

\item $q$ reads $R_{wq} =(\ready, v_{k-1}, v_{k})$ \iln{qPre} $\ru_q(v_k)$
	before
	$q$ reads $R_{wq} =(\ready, v_{k'-1}, v_{k'})$ \iln{qPre} $\ru_{q}(v_{k'})$.
By Observation~\ref{Smonowp}(\ref{Smonowp2}),
	 \mbox{$k \le k'$}.

 \item $q$ reads $R_{wq} =(\ready, v_{k-1}, v_{k})$ \iln{qPre} $\ru_q(v_k)$
 	before
	$q$ reads $R_{wq} =(\ready, v_{k'}, v_{k'+1})$ \iln{qPre} $\ru_{q}(v_{k'})$.
 By Observation~\ref{Smonowp}(\ref{Smonowp2}),
 	 $k \le k'+1$.
	\begin{compactenum}[i.]
	\item $k < k'+1$. Then $k\le k'$.

	\item $k = k'+1$.
	We now show that this case is impossible.
	Since $k = k'+1$, 
		$q$ reads $R_{wq} =(\ready, v_{k-1}, v_{k})$ in $\ru_{q}(v_{k-1})$,
	and 
	$\ru_{q}(v_{k-1})$ returns in line~\ref{SqPrtlw}.
	Note that $q$ reads
	$\lr$ in line~\ref{Srtoq} of $\ru_{q}(v_{k-1 })$ before
	$\ru_{q}(v_{k-1})$ returns in line~\ref{SqPrtlw}.
\begin{claim}\label{above2}
Process $q$ reads $\lr=v_{\ell}$ for some $\ell \ge k$ in line~\ref{Srtoq} of $\ru_{q}(v_{k-1 })$.
\end{claim}

\begin{proof}

Since $q$ reads $R_{wq} =(\ready, v_{k-1}, v_{k})$ in $\ru_q(v_k)$,
	$\ru_q(v_k)$ returns
	in line~\ref{SqPrt2} or~\ref{SqPrt3}.
There are two~cases:

		\begin{compactenum}[a.]

		\item $\ru_q(v_k)$ returns in line~\ref{SqPrt2} of $\ru_q(v_k)$.
		So $q$ writes $v_{k}$ in $\lr$ in line~\ref{Sqtor} of $\ru_q(v_k)$.
Since $\ru_q(v_k)$ precedes $\ru_{q}(v_{k-1})$,
	by Observations~\ref{Smonorqwr},
		when $q$ reads 
	$\lr$ in line~\ref{Srtoq} of $\ru_{q}(v_{k-1 })$,
	$q$ reads $\lr=v_{\ell}$ for some $\ell \ge k$.
	
		\item $\ru_q(v_k)$ returns in line~\ref{SqPrt3} of $\ru_q(v_k)$.
		So $q$ reads $\lr=v_{\ell'}$ with some $\ell' \ge k$ in line~\ref{Srtoq} of $\ru_q(v_k)$.
Since $\ru_q(v_k)$ precedes $\ru_{q}(v_{k-1})$,
	by Observations~\ref{Smonoqrr},
		when $q$ reads 
	$\lr$ in line~\ref{Srtoq} of $\ru_{q}(v_{k-1 })$,
	$q$ reads $\lr=v_{\ell}$ for some $\ell \ge \ell' \ge k$.
		\end{compactenum}
\end{proof}

By Claim~\ref{above2} and the code of $\ru_{q}()$, it is clear that $\ru_{q}(v_{k-1 })$ returns $v_k$ in line~\ref{SqPrt3} --- a contradiction.

	\end{compactenum}

\item $q$ reads $R_{wq} =(\ready, v_{k}, v_{k+1})$ \iln{qPre} $\ru_q(v_k)$
	before
	$q$ reads $R_{wq} = (\done, v_{k'})$ \iln{qCommit} $\ru_{q}(v_{k'})$.
 By Observation~\ref{Smonowp}(\ref{Smonowp1}), $k+1 \le k'$. So \mbox{$k \le k'$}.

\item $q$ reads $R_{wq} =(\ready, v_{k}, v_{k+1})$ \iln{qPre} $\ru_q(v_k)$
	before
	$q$ reads $R_{wq} =(\ready, v_{k'-1}, v_{k'})$ \iln{qPre} $\ru_{q}(v_{k'})$.
By Observation~\ref{Smonowp}(\ref{Smonowp2}), $k+1 \le k'$. So \mbox{$k \le k'$}.

\item $q$ reads $R_{wq} =(\ready, v_{k}, v_{k+1})$ \iln{qPre} $\ru_q(v_k)$
	before
	$q$ reads $R_{wq} =(\ready, v_{k'}, v_{k'+1})$ \iln{qPre} $\ru_{q}(v_{k'})$.
By Observation~\ref{Smonowp}(\ref{Smonowp2}), $k+1 \le k'+1$. So \mbox{$k \le k'$}.
\end{compactenum}
\end{proof}

We now prove that the writes and reads of the lower-level procedures $\wu()$, $\ru_p()$, and  $\ru_q()$ satisfy
	Property~2 of Definition~\ref{LinearizableByz}.

\begin{lemma}\label{Sp2r}
Let $\ru_r(v_k)$ and $\ru_{r'}(v_{k'})$ be any read operations by non-{\ml} processes $r$ and $r'$
	in $\{p,q\}$.
If $\ru_r(v_k)$ precedes $\ru_{r'}(v_{k'})$ then \mbox{$k \le k'$}.
\end{lemma}

\begin{proof}
Immediate from Lemmas~\ref{Smonopread},~\ref{Snopq},~\ref{Snoqp},~\ref{Snorq}.
\end{proof}

Finally, we prove that the write and read operations
	of the high-level procedures $\Wu()$ and $\Ru()$
	satisfy Property~2 of Definition~\ref{LinearizableByz}.

\begin{lemma}\label{SCp2} \emph{[\textbf{Property 2}: No ``new-old'' inversion]}

Let $\Ru(u_k)$ and $\Ru(u_{k'})$ be any read operations by non-{\ml} processes in $\{p,q\}$.
If $\Ru(u_k)$ precedes $\Ru(u_{k'})$ then \mbox{$k \le k'$}.
\end{lemma}

\begin{proof}
Let $\Ru(u_k)$ and $\Ru(u_{k'})$ be any read operations by non-{\ml} processes $r$ and $r'$ in $\{p,q\}$.
Suppose $\Ru(u_k)$ precedes $\Ru(u_{k'})$.
By Observation~\ref{SRr},
 in the $\Ru(u_k)$ and $\Ru(u_{k'})$ operations,
	processes $r$ and $r'$ invoke and complete a $\ru_r(v_k)$ and $\ru_{r'}(v_{k'})$ operation, respectively.
Since $\Ru(u_k)$ precedes $\Ru(u_{k'})$,
	$\ru_r(v_k)$ precedes $\ru_{r'}(v_{k'})$.
By Lemma~\ref{Sp2r},
	\mbox{$k \le k'$}.
\end{proof}

By Lemmas~\ref{SCp1} and~\ref{SCp2}, the $\Wu(-)$ and $\Ru(-)$ operations of the register implementation~$\Iup{2}$
 	satisfy the linearizability Properties~1 and~2 of Definition~\ref{LinearizableByz}, so $\Iup{2}$ is linearizable.
By Observation~\ref{Obs-algo2iswf}, $\Iup{2}$ is also bounded wait-free.
Thus:

\Stwofromone*

%% file: proof-impossibility-bounded-termination.tex

\section{Bounded Termination: impossibility proof}\label{bounded-termination-impo}
We now prove that in a system with $n+1$ Byzantine processes,
	there is no linearizable implementation of a $\Reg{1}{n}$ from atomic $\Reg{1}{n-1}$s
	that satisfies the Bounded Termination property
	even if we assume that \emph{only the readers can be faulty, and at most one of them can fail}.
More precisely:
	
\BoundedImpoResult*

\begin{proof}
Let $n \ge 3$.
Suppose, for contradiction, that there is an 
	implementation {\AWB}
	of a $\Reg{1}{n}$ $\REG$ from atomic $\Reg{1}{n-1}$s
	that is linearizable (i.e., it satisfies the Register Linearizabilty property)
	and satisfies the Bounded Termination property
	in a system where the writer $w$ of $\REG$
	is correct and at most one of the $n$ readers of $\REG$ can be {\ml}.

\begin{figure}[!t]
\vspace{-6mm} 
\minipage{0.48\textwidth}
     \centering 
    \includegraphics[width=0.9\textwidth]{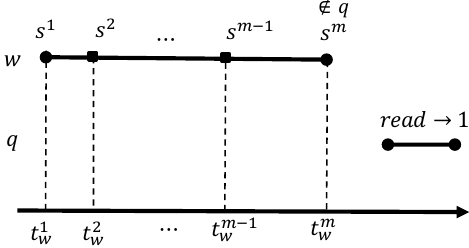}
    \caption{{\Run} $A_m$} 
    \label{Bounded-S}
\endminipage
\minipage{0.48\textwidth}
    \centering 
    \includegraphics[width=0.9\textwidth]{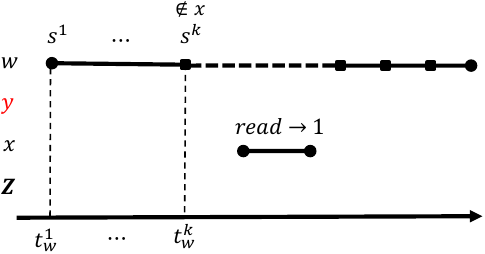}
    \caption{A {\run} with property $P_k$} 
    \label{Bounded-E}
\endminipage\hfill
\newline

\minipage{0.48\textwidth}
    \centering 
    \includegraphics[width=0.9\textwidth]{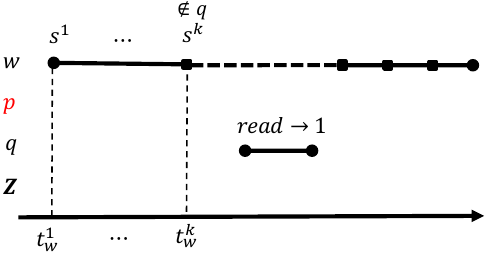}
    \caption{{\Run} $A_{k}$} 
    \label{Bounded-Ak}
\endminipage
\minipage{0.48\textwidth}
    \centering 
    \includegraphics[width=0.9\textwidth]{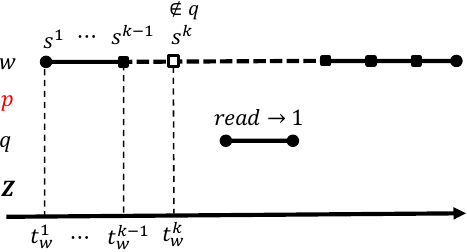}
    \caption{{\Run} $B_{k-1}$} 
    \label{Bounded-Bk-1}
\endminipage\hfill
\newline

\minipage{0.48\textwidth}%
    \centering 
    \includegraphics[width=0.9\textwidth]{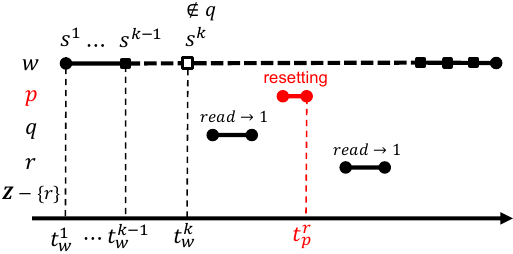}
    \caption{{\Run} $C_{k-1}^r$} 
    \label{Bounded-Ck-1}
\endminipage
\minipage{0.48\textwidth}
    \centering 
    \includegraphics[width=0.9\textwidth]{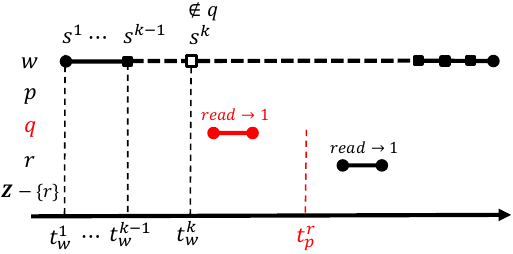}
    \caption{{\Run} $D_{k-1}^r$} 
    \label{Bounded-Dk-1}
\endminipage\hfill
\newline

\minipage{0.48\textwidth}%
    \centering 
    \includegraphics[width=0.9\textwidth]{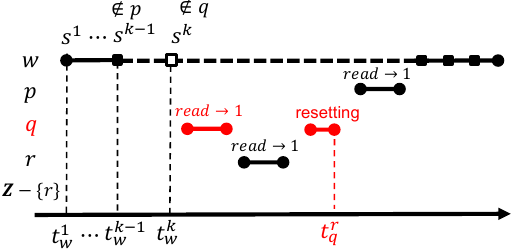}
    \caption{{\Run} $E_{k-1}^r$} 
    \label{Bounded-Ek-1}
\endminipage
\minipage{0.48\textwidth}%
    \centering 
    \includegraphics[width=0.9\textwidth]{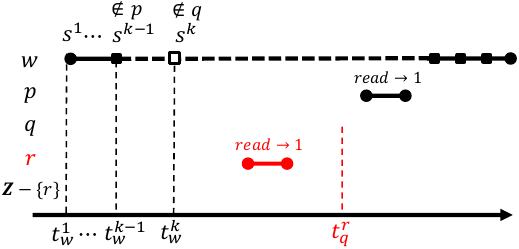}
    \caption{{\Run} $F_{k-1}^r$} 
    \label{Bounded-Fk-1}
\endminipage\hfill
\newline

\minipage{0.48\textwidth}%
    \centering 
    \includegraphics[width=0.9\textwidth]{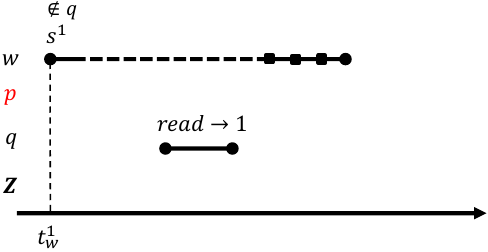}
    \caption{{\Run} $A_1$} 
    \label{Bounded-A1}
\endminipage
\minipage{0.48\textwidth}
    \centering 
    \includegraphics[width=0.9\textwidth]{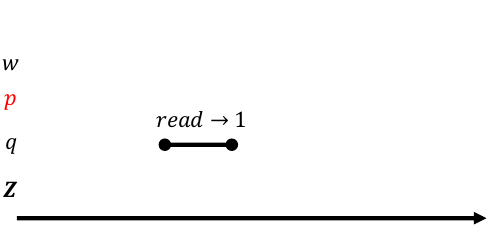}
    \caption{{\Run} $A_0$} 
    \label{Bounded-A0}
\endminipage\hfill
\end{figure}

We now construct a sequence of {\run}s of {\AWB} that leads to a contradiction.
In all these {\run}s, the initial value of the implemented $\REG$ is $0$,
	the writer $w$ invokes only one operation into~$\REG$, namely a write of $1$,
	and each reader reads $\REG$ at most once (i.e., $\REG$ is only a ``one-shot'' binary register).

In all these {\run}s: (a) the writer is correct and
	 (b) there is at most one {\ml} reader (the other $n-1$ readers~are~correct).
Thus, these runs of {\AWB}
	must satisfy the linearizability Properties 1 and 2 of
	Register Linearizability (Definition~\ref{LinearizableByz}),
	and Bounded Termination (Definition~\ref{Btermination}), i.e., 
	every correct reader must complete any read operation that it invokes in a bounded number of steps.

\begin{definition}\label{Bounded-invisible}
Let $s$ be any step that the writer $w$ takes when executing the implementation {\AWB} of $\REG$.
Step $s$ is \emph{invisible
	to a reader $p$} if $s$ is either a local step of $w$, or 
	the reading or the writing of an atomic $\Reg{1}{n-1}$ that is not readable by $p$.
\end{definition}

Since there are $n$ readers,
	and the registers that $w$ can write 
	are atomic $\Reg{1}{n-1}$s,
	every write by $w$ into one of these registers
	is invisible to one of the readers.
So: 

\begin{observation}\label{Bounded-invisibletoareader}
Let $s$ be any step that the writer $w$ takes when executing the implementation {\AWB} of $\REG$. 
Step $s$ is invisible to at least one of the $n$ readers.
\end{observation}

Let $A_m$ be the following {\run} of {\AWB} (see Figure~\ref{Bounded-S}):

\begin{compactitem}  
\item The writer $w$ and all the readers are correct.

\item The writer $w$ invokes an operation to write 1 on $\REG$.
	By the Bounded Termination property of {\AWB},
 	$w$ completes this operation.
	
During this write operation, 
	$w$ takes a sequence of steps $s^1,...,s^{m}$
	such that each $s^i$ is either a local step, or
	the reading or the writing of an atomic $\Reg{1}{n-1}$
	($s^0$ is the invocation step of the write operation, and $s^m$ is the response step of this operation).
	Let~$t_w^i$ be the time when step $s^i$ occurs.
	
\item After taking the step $s^{m}$ at time $t_w^m$, the writer $w$ stops taking steps
	(it has completed its write operation on $\REG$).
	
\item Let $q$ be a reader such that step $s^m$ is invisible to $q$
	(by Observation~\ref{Bounded-invisibletoareader}, this reader exists).

	 After the time $t_w^{m}$,
	correct reader $q$ invokes a read operation on $\REG$.
	By the Bounded Termination property of {\AWB},
	$q$ completes its read operation.
	By the linearizability properties of {\AWB}, this read operation on $\REG$ returns 1.

\item All the other readers take no steps.
		
\end{compactitem}
\begin{definition}\label{Bounded-pk}
For every $k$,
 	$1 \le k \le m$,
 	a {\run} of {\AWB} has property $\pk{k}$
	 if the following holds:
\begin{compactenum}
\item Up to and including time $t_w^k$, all processes behave exactly as in $A_m$, that is:
\begin{compactitem}	
	\item $w$ takes steps $s^0, s^1, \ldots, s^k$
	\item All the readers take no steps.
\end{compactitem}

\item After taking the step $s^k$ at time $t_w^k$, the correct writer $w$ behaves as follows:
	\begin{compactitem}	
	\item If $k=m$, $w$ stops taking steps: it has completed its write operation on $\REG$.
	\item If $k<m$, $w$ \emph{temporarily} stops taking steps. 
	\end{compactitem}

\item There is a reader $x$ that is correct such that step $s^k$ is invisible to $x$.
After time $t_w^k$, reader~$x$ starts and completes a read operation on $\REG$ that returns~$1$.

\item There is a reader $y \neq x$ that may be correct or malicious. After time $t_w^k$, reader $y$ may or may not take steps.

\item There is a set {\sety} of $n-2$ distinct readers other than $x$ and $y$ that are correct and~take~no~steps.

\item If $k < m$, after the reader $x$ reads 1 from $\REG$, the correct writer $w$ resumes taking steps and completes its write operation on $\REG$.
\end{compactenum}
\end{definition}

Note that since $n\ge 3$, the set {\sety} contains at least one reader.
Furthermore, all the readers that take steps do so after time $t_w^k$. 

A {\run} of {\AWB} with property $\pk{k}$ 
	 is shown in Figure~\ref{Bounded-E}.
In this figure and all the subsequent ones,
	correct readers are in black font,
	while the reader that may be {\ml} is colored~{\color{red} red}
	(this~reader may have taken some steps after time $t_w^k$, but these are \emph{not} shown in the figure).
The ``$\notin x$'' on top of a step $s^i$ means that $s^i$
	is invisible to the reader~$x$.

Note that the {\run} $A_m$ of {\AWB} satisfies property $\pk{m}$:
	the reader denoted $x$ in property $\pk{m}$ is the reader $q$ of {\run} $A_m$,
	the reader $y$ of $\pk{m}$ is an arbitrary reader other than $q$ in $A_m$,
	and the set $Z$ of $\pk{m}$ is the set of the remaining $n-2$ readers in $A_m$.
So we have:

\begin{observation}\label{Bounded-baserun}
Run $A_m$ of {\AWB} has property $\pk{m}$.
\end{observation}

\begin{claim}~\label{Bounded-induction}
	 For every $k$, $1 \le k \le m$, there is a {\run} of {\AWB} that has property $\pk{k}$.
\end{claim}
\begin{proof}
We prove the claim by a backward induction on $k$, starting from $k=m$.

\smallskip\noindent
\textbf{Base Case:} $k=m$. This follows directly from Observation~\ref{Bounded-baserun}.

\noindent
\textbf{Induction Step}: Let $k$ be such that $1 < k \le m$.

\RN{$A_{k}$.} Suppose
	there is a {\run} $A_{k}$
	of {\AWB} that has property $\pk{k}$ (this is the induction hypothesis).
We now show that there is a {\run} $A_{k-1}$ of {\AWB} that has property $\pk{k-1}$.

Since {\run} $A_{k}$ of {\AWB} satisfies $\pk{k}$, the following holds in $A_{k}$ (see Figure~\ref{Bounded-Ak}):
\begin{compactitem}  

\item Up to and including time $t_w^k$, all processes behave exactly as in $A_m$.

\item After taking the step $s^k$ at time $t_w^k$, the correct writer $w$ behaves as follows:
	\begin{compactitem}	
	\item If $k=m$, $w$ stops taking steps: it has completed its write operation on $\REG$.
	\item If $k<m$, $w$ {\tsts}. 
	\end{compactitem}
	
\item There is a reader $q$ that is correct such that step $s^k$ is invisible to $q$.
After time $t_w^k$, 
reader~$q$ starts and completes a read operation on $\REG$ that returns~$1$.

\item There is a reader $p \neq q$ that may be correct or malicious. After time $t_w^k$, reader $p$ may or may not take steps.\footnote{These steps are not shown in Figure~\ref{Bounded-Ak}.}

\item There is a set {\sety} of $n-2$ distinct readers other than $p$ and $q$ that are correct and~take~no~steps.

\item If $k < m$, after $q$ reads 1 from $\REG$, the correct writer $w$ resumes taking steps and completes its write operation on $\REG$.

\end{compactitem}

\RN{$B_{k-1}$.}
From the {\run} $A_{k}$ of {\AWB} we construct the following {\run} $B_{k-1}$ of {\AWB} (Figure~\ref{Bounded-Bk-1}).
Intuitively, $B_{k-1}$ is the same as $A_{k}$ except that: 
(a) after taking step $s_{k-1}$ at time $t_w^{k-1}$, the writer $w$ {\tsts}, and (b) $w$ resumes taking steps only after the reader $q$ completes its read of 1. 
This {\run} is possible because
	even though $p$ may have ``noticed'' that $w$ ``pauses''
	after taking step $s^{k-1}$,
	$p$~may be {\ml} (all the other readers are correct in this {\run}),
	and $p$ behaves exactly as in $A_k$,
and 
(2) $q$ cannot distinguish between $A_k$ and $B_{k-1}$
		because step $s^k$ is invisible to $q$, and $p$ and all the readers in {\sety} behave as in $A_k$;
		so $q$ behaves as in $A_k$, and in particular $q$ reads 1 in $B_{k-1}$ as in $A_{k}$.
After $q$ reads 1, $w$~completes its write operation on $\REG$.

More precisely in $B_{k-1}$:

\begin{compactitem}   
\item All processes behave exactly as in $A_k$ up to and including time $t_w^{k-1}$.

\item After taking step $s^{k-1}$ at time $t_w^{k-1}$, $w$ {\tsts}. 
	
\item All the readers in {\sety} are correct and take no steps, exactly as in $A_k$.

\item $p$ behaves exactly as in $A_k$.
	This is possible because
		even though $p$ may have ``noticed'' that $w$ {\tsts} after step $s^{k-1}$,  
		$p$ may be {\ml} (all the other readers are correct in this {\run}).
		
\item $q$ behaves exactly as in $A_k$.
	In particular, after time $t_w^k$, 
	$q$~starts and completes a read operation on $\REG$ that returns~$1$.
	This is possible because $q$ cannot distinguish between $A_k$ and $B_{k-1}$:
	 $s^k$ is invisible to $q$, and $p$ and all the readers in {\sety} behave exactly as in $A_k$.

\item After $q$ reads 1 from $\REG$, the correct writer $w$ resumes taking steps and completes its write operation on $\REG$.

\end{compactitem}

Note that in $B_{k-1}$ all processes behave exactly as in $A_m$ up to and including time $t_w^{k-1}$.

There are two cases:

\textbf{Case 1:} 
\emph{$s^{k-1}$ is invisible to $q$.}
Then $B_{k-1}$ is a {\run} of {\AWB} that has the property $\pk{k-1}$, as we wanted to show.

\textbf{Case 2:} 
\emph{$s^{k-1}$ is visible to $q$.}
Then, by Observation~\ref{Bounded-invisibletoareader}, $s^{k-1}$ is invisible to $p$ or to~some~$r' \in ~{\sety}$.

\noindent
\RN{$C_{k-1}^r$.}
Let $r$ be \emph{any} reader in {\sety}.
From the {\run} $B_{k-1}$ of {\AWB} we construct the following {\run} $C_{k-1}^r$ of {\AWB} (Figure~\ref{Bounded-Ck-1}).
$C_{k-1}^r$ is the same as $B_{k-1}$ up to the time when $q$ completes its read operation on~$\REG$.
After the correct reader $q$ reads~1, {\ml} process $p$
	wipes out any trace of the write steps that it may have taken so far,
	and then correct reader $r\in \sety$  
	starts a read operation~on~$\REG$.
	By the Bounded Termination property of {\AWB}, 
		this read operation by $r$ must complete (without waiting for the correct writer $w$ to complete its write operation\footnote{Even though $r$ ``knows'' 
		that $w$ is correct and so $w$ will eventually take all the steps necessary to complete its write operation, $r$ cannot wait for them: this would violate the \emph{Bounded} Termination property of {\AWB}.}).
	Since 
		 $q$ previously read~1,
		by the linearizability of {\AWB}, 
		$r$ also reads 1.
After~$r$~reads~1, $w$~completes its write operation on $\REG$.

More precisely in $C_{k-1}^r$:

\begin{compactitem}   
\item All processes behave exactly as in $B_{k-1}$ up to and including the time when $q$ completes its read operation on $\REG$.

\item All the readers in $\sety-\{r\}$ are correct and take no steps\footnote{If $n=3$, then the set $\sety-\{r\}$ is empty.}.

\item After the correct reader $q$ completes its read operation on $\REG$:

\begin{compactitem}   

	\item $q$ takes no steps.

	\item 
	$p$ resets all the atomic registers that it can write to their initial values.
	Process $p$ can do so because it may be {\ml} (all the other readers are correct in this {\run}).
	Let~$t_p^r$~be the time when
	$p$ completes all the register resettings.

	\item Correct reader $r$
	starts a read operation on $\REG$ after time $t_p^r$.
	It takes no steps before~this~read.
	By the Bounded Termination property of {\AWB},
		$r$ completes its read operation (without waiting for correct $w$ to resume taking its steps).
	Since $w$ is correct, and
		the read operation by correct $q$ precedes the read operation by $r$ and returns~$1$,
		by the linearizability of~{\AWB},
		the read operation by correct reader $r$~also~returns~$1$.
	
	\item After $r$ reads 1 from $\REG$, the correct writer $w$ resumes taking steps and completes its write operation on $\REG$.	
\end{compactitem}
\end{compactitem}

Note that in $C_{k-1}^r$ all processes behave exactly as in $A_m$ up to and including time $t_w^{k-1}$.

\RN{$D_{k-1}^r$.}
We can now construct the following {\run} $D_{k-1}^r$ of {\AWB} (Figure~\ref{Bounded-Dk-1}).
Intuitively,	we obtain $D_{k-1}^r$ from $C_{k-1}^r$ by removing all the steps of $p$.
So reader $p$ (which was {\ml} in $C_{k-1}^r$)
	is now a correct process that takes no steps.
Despite the removal of $p$'s steps, $q$~behaves exactly as in $C_{k-1}^r$ because $q$
	(which was correct in $C_{k-1}^r$) may now be {\ml}.
Up to and including time~$t_w^{k-1}$, the writer $w$ also behaves exactly
	as in $C_{k-1}^r$ because it cannot see the removal of $p$'s steps:
	they all occur after time $t_w^{k-1}$.
Correct reader $r$ behaves exactly as in $C_{k-1}^r$
	because it also cannot see the removal of $p$'s steps:
	in both $C_{k-1}^r$ and $D_{k-1}^r$,
	$r$ does not ``see'' any steps of $p$.
So $r$ reads 1 in $D_{k-1}^r$ as in $C_{k-1}^r$.
After $r$ reads 1, $w$~completes its write operation on $\REG$.

More precisely in $D_{k-1}^r$:

\begin{compactitem}   
\item After taking step $s^{k-1}$ at time $t_w^{k-1}$, $w$ {\tsts}, as~in~$C_{k-1}^r$.

\item All the readers in $\sety-\{r\}$ are correct and take no steps, as in $C_{k-1}^r$.

\item $p$ is correct and it takes no steps.
	So all the
	atomic
	registers that it can write retain their initial~values.

\item $q$ behaves exactly as in $C_{k-1}^r$.
	  This is possible because
		even though $q$ may have ``noticed'' the removal of $p$'s steps,
		$q$ may be {\ml} (all the other readers are correct in this {\run}).
	
\item $r$ behaves exactly as in $C_{k-1}^r$.
	In particular,
	after time $t_p^r$ reader $r$~starts and completes a read operation on $\REG$ that returns~$1$.
	This is possible because $r$ cannot distinguish between $C_{k-1}^r$ and~$D_{k-1}^r$:
	$r$ cannot see the removal of $p$'s steps, 
		and $q$ and all the readers in $\sety-\{r\}$ behave exactly as in $C_{k-1}^r$.
		
\item After $r$ reads 1 from $\REG$, the correct writer $w$ resumes taking steps and completes its write operation on $\REG$.	

\end{compactitem}

Note that in $D_{k-1}^r$ all processes behave exactly as in $A_m$ up to and including time $t_w^{k-1}$.

If $s^{k-1}$ is invisible to reader $r$, it is clear that the {\run} $D_{k-1}^r$ of {\AWB} has property $\pk{k-1}$.

Recall that (1) the reader $r$ above is an \emph{arbitrary} reader in {\sety},
	and
	(2) $s^{k-1}$ is invisible to~$p$ or to some reader $r' \in \sety$.
So there are two cases:

\textbf{Subcase 2a:} \emph{$s^{k-1}$ is invisible to some reader $r' \in \sety$.}
In the above we proved that the {\run} $D_{k-1}^{r'}$ of {\AWB} has property $\pk{k-1}$,
	as we wanted to show.
	
\textbf{Subcase 2b:} \emph{$s^{k-1}$ is invisible to $p$.}

\RN{$E_{k-1}^r$.}
From the {\run} $D_{k-1}^r$ of {\AWB} we construct the following {\run} $E_{k-1}^r$ of {\AWB} (Figure~\ref{Bounded-Ek-1}).
$E_{k-1}^r$ is the same as $D_{k-1}^r$ up to the time when $r$ completes its read operation on~$\REG$.
After $r$ reads~1, {\ml} process $q$
	wipes out any trace of the write steps that it may have taken so far,
	and then correct reader $p$  
	starts a read operation~on~$\REG$.
	By the Bounded Termination property of {\AWB}, 
		this read operation by $p$ must complete (without waiting for the correct writer $w$ to complete its write operation).
	Since $r$ previously read~1,
		by the linearizability~of~{\AWB}, 
	$p$ also reads 1.
After~$p$~reads~1, $w$~completes its write operation on $\REG$.

More precisely in $E_{k-1}^r$:
\begin{compactitem} 

\item All processes behave exactly as in $D_{k-1}^r$ up to and including the time when $r$ completes its read operation on $\REG$.

\item All the readers in $\sety-\{r\}$ are correct and take no steps, as in $D_{k-1}^r$.

\item After the correct reader $r$ completes its read operation on $\REG$:

\begin{compactitem}   

	\item $r$ takes no steps.

	\item 
	$q$ resets all the atomic registers that it can write to their initial values.
	Process $q$ can do so because it may be {\ml} (all the other readers are correct in this {\run}).
	Let~$t_q^r$~be the time when
	$q$ completes all the register resettings.

	\item Correct reader $p$
	starts a read operation on $\REG$ after time $t_q^r$.
	It takes no steps before this read.
	By the Bounded Termination property of {\AWB},
		$p$ completes its read operation (without waiting for correct $w$ to resume taking its steps).
	Since $w$ is correct,
		and the read operation by correct $r$ precedes the read operation by $p$ and returns~$1$,
		by the linearizability of {\AWB},
		the read operation by correct reader $p$~also~returns~$1$.
			
	\item After $p$ reads 1 from $\REG$, the correct writer $w$ resumes taking steps and completes its write operation on $\REG$.	

\end{compactitem}

\end{compactitem}

Note that in $E_{k-1}^r$ all processes behave exactly as in $A_m$ up to and including time $t_w^{k-1}$.

\RN{$F_{k-1}^r$.}
Finally, we construct the {\run} $F_{k-1}^r$ of {\AWB} by removing
	all the steps of $q$ from $E_{k-1}^r$ (see Figure~\ref{Bounded-Fk-1}).
So $q$ (which was {\ml} in $E_{k-1}^r$)
	is now a correct process that takes no steps.
Despite the removal of $q$'s steps, $r$ behaves exactly as in $E_{k-1}^r$ because $r$ (which was correct in $E_{k-1}^r$) may now be {\ml}.
Up to and including time $t_w^{k-1}$, the writer $w$ also behaves exactly as in $E_{k-1}^r$ because it cannot see the removal of $q$'s steps:
	they all occur after time $t_w^{k-1}$.
Correct reader $p$ behaves exactly as in $E_{k-1}^r$ because it also cannot see the removal of $q$'s steps:
	in both $E_{k-1}^r$ and $F_{k-1}^r$,
	$p$ does not ``see'' any steps of $q$. 
So $p$ reads 1 in $F_{k-1}^r$ as in $E_{k-1}^r$.
After~$p$ reads 1, $w$~completes its write operation on $\REG$.

More precisely in $F_{k-1}^r$:

\begin{compactitem} 

\item After taking step $s^{k-1}$ at time $t_w^{k-1}$, $w$ {\tsts}, as~in~$E_{k-1}^r$. 

\item All the readers in $\sety-\{r\}$ are correct and take no steps, as in $E_{k-1}^r$.

\item $q$ is correct and it takes no steps. So all the atomic registers that it can write retain their initial values.
	
\item $r$ behaves exactly as in $E_{k-1}^r$.
	  This is possible because
		even though $r$ may have ``noticed'' the removal of $q$'s steps,
		$r$ may be {\ml} (all the other readers are correct in this {\run}).
	
\item $p$ behaves exactly as in $E_{k-1}^r$.
	In particular, after time $t_q^r$ reader
	$p$~starts and completes a read operation on $\REG$ that returns~$1$.
	This is possible because $p$ cannot distinguish between $E_{k-1}^r$ and~$F_{k-1}^r$:
	$p$ cannot see the removal of $q$'s steps, 
		and $r$ and all the readers in $\sety-\{r\}$ behave exactly as in $E_{k-1}^r$.

\item After $p$ reads 1 from $\REG$, the correct writer $w$ resumes taking steps and completes its write operation on $\REG$.	

\end{compactitem}

Note that in $F_{k-1}^r$ all processes behave exactly as in $A_m$ up to and including time $t_w^{k-1}$.

Since $s^{k-1}$ is invisible to $p$, it is clear that the {\run} $F_{k-1}^r$ of {\AWB} has property $\pk{k-1}$.

The above concludes the proof of the Induction Step of Claim~\ref{Bounded-induction}:
	we proved that, in all possible cases, there is a {\run} of {\AWB} that has property $\pk{k-1}$, as we needed to show.
\end{proof}

By the Claim~\ref{Bounded-induction} that we just proved,
	the implementation {\AWB} of $\REG$ has a {\run} $A_1$ with property~$\pk{1}$.
By this property, the following holds in $A_{1}$ (see Figure~\ref{Bounded-A1}):

\begin{compactitem}  

\item Up to and including time $t_w^1$, all processes behave exactly as in $A_m$.

\item After taking the step $s^1$ at time $t_w^1$, the correct writer $w$ {\tsts}.
	
\item There is a reader $q$ that is correct such that step $s^1$ is invisible to $q$.
After time $t_w^1$, 
reader~$q$~starts and completes a read operation on $\REG$ that returns~$1$.

\item There is a reader $p \neq q$ that may be correct or malicious. After time $t_w^1$, reader $p$ may or may not take steps.

\item There is a set {\sety} of $n-2$ distinct readers other than $p$ and $q$ that are correct and~take~no~steps.

\item After $q$ reads 1 from $\REG$, the correct writer $w$ resumes taking steps and completes its write operation on $\REG$.

\end{compactitem}

From the {\run} $A_{1}$ of {\AWB} we construct the following {\run} $A_0$ of {\AWB} (Figure~\ref{Bounded-A0}).
Intuitively, $A_0$~is the same as $A_1$ except that the correct writer $w$ does not take any steps
	(i.e., $w$ does~not invoke a write 1 operation on $\REG$), but all the readers
	behave the same as in $A_1$ and so $q$ still reads 1.
This {\run} of {\AWB} is possible because:
	(1) even though $p$ may have ``noticed'' that $w$~does not take any steps,
	$p$~may be {\ml} (all the other readers are correct in this {\run}),
	and $p$ behaves exactly as in $A_1$,
	and 
	(2) $q$ cannot distinguish between $A_1$ and $A_0$
	because $s^1$ is invisible to $q$, and $p$ and all the readers in {\sety} behave as in $A_1$.
So $q$ reads 1 from $\REG$ in $A_0$ exactly as in $A_{1}$.
Since the initial value of the implemented register $\REG$ is $0$, {\run} $A_0$ of the implementation {\AWB}
	of $\REG$ violates the linearizability~of~{\AWB}
	--- a contradiction that concludes the proof of Theorem~\ref{Bounded-Theo-Impossibility-Result}.
\end{proof}

It is easy to verify that the above proof
	holds (without any change) even if all the readers have atomic $\Reg{1}{n}$s that they can write and all processes can read.
Thus:

\begin{theorem}\label{Bounded-STheo-Impossibility-Result}
For all $n \ge 3$, in a system with $n+1$ processes that are subject to Byzantine failures,
	there is \emph{no}
	linearizable implementation of a 
	$\Reg{1}{n}$
	that satisfies Bounded Termination,
	even under the assumption that:
	
\begin{compactitem}

\item The writer $w$ of the implemented $\Reg{1}{n}$ is correct and at most one reader can be {\ml}, and

\item  $w$ has atomic $\Reg{1}{n-1}$s, and every reader has atomic $\Reg{1}{n}$s.

\end{compactitem}
\end{theorem} 	